\documentclass[11pt,final]{article}

\usepackage[T1]{fontenc}
\usepackage{hyperref}
\usepackage{amsmath}
\usepackage{amssymb}
\usepackage{amsthm}
\usepackage{mathtools}
\usepackage{lmodern}
\usepackage{dsfont}
\usepackage{mathrsfs}
\usepackage{tikz}
\usepackage{subcaption}
\usepackage[bottom=1in,top=1in,left=1in,right=1in]{geometry}
\usepackage[ruled]{algorithm2e}
\usepackage[compact]{titlesec}
\usepackage{times}

\usepackage{thmtools,thm-restate} 

\usetikzlibrary{shapes}
\usetikzlibrary{calc}
\usetikzlibrary{decorations.pathmorphing}
\usetikzlibrary{matrix,fadings}

\usepackage{enumitem}
\usepackage{xspace}

\usepackage[colorinlistoftodos,textsize=tiny,textwidth=2cm,color=red!25!white,obeyFinal]{todonotes}
 
\newcommand{\jmcom}[1]{\todo[color=orange!25!white]{JM: #1}\xspace}

\newcommand{\msnote}[1]{\todo[color=cyan!25!white]{MS: #1}\xspace}

\newcommand{\ama}{\[ \begin{aligned}}
\newcommand{\ema}{\end{aligned} \]}
\newcommand{\set}[1]{\left\{#1\right\}}

\newcommand{\ceil}[1]{\left\lceil#1\right\rceil}

\newcommand{\R}{\mathds{R}}

\newtheorem{theorem}{Theorem}

\newtheorem{definition}[theorem]{Definition}

\newtheorem{lemma}[theorem]{Lemma}
\newtheorem{fact}[theorem]{Fact}
\newtheorem{lemdef}[theorem]{Lemma and Definition}
\newtheorem{observation}[theorem]{Observation}
\newtheorem{corollary}[theorem]{Corollary}

\newtheorem{invariants}[theorem]{Invariants}
\newtheorem{invariant}[theorem]{Invariant}

\newcommand{\eg}{e.g.,\xspace}
\newcommand{\etal}{et~al.\xspace}
\newcommand{\ie}{i.e.,\xspace}

\usepackage{xspace}

\newcommand{\cset}[2]{\{#1 \mid #2\}}

\newcommand{\srsum}[1]{\smashoperator[r]{\sum_{#1}}}

\newcommand{\slrsum}[1]{\smashoperator[lr]{\sum_{#1}}}

\newcommand{\frakm}{\mathfrak{m}}
\newcommand{\frakn}{\mathfrak{n}}
\newcommand{\lenl}{\mathfrak{l}}
\renewcommand{\lenl}{\ell}

\newcommand{\eqS}{{\mathfrak{S}}}
\newcommand{\terms}{\mathfrak{T}}

\newcommand{\sse}{\subseteq}
\newcommand{\Rp}{\R_{\geq 0}}
\newcommand{\ALG}{\ensuremath{\mathscr{A}}}
\newcommand{\OPT}{\text{OPT}}

\newcommand{\REFF}{\ensuremath{\mathscr{F}}}
\newcommand{\tforall}{\text{for all}}
\DeclareMathOperator{\dist}{dist}

\newcommand{\maxit}{\ensuremath{\mathfrak{index}}}
\newcommand{\innS}{{\text{In}}_{\scriptstyle S}}
\newcommand{\innSp}{{\text{In}}_{\scriptstyle S'}}
\newcommand{\contract}{\diagup}
\newcommand{\indicator}{\mathds{1}}

\newcommand{\bFin}{{\bar F}_{\circlearrowright}}
\newcommand{\bFbetw}{{\bar F}_{\leftrightarrow}}
\newcommand{\Fin}{{F}_{\circlearrowright}}
\newcommand{\Fbetw}{{F}_{\leftrightarrow}}

\newcommand{\initOneLiners}{%
    \setlength{\itemsep}{0pt}
    \setlength{\parsep }{0pt}
    \setlength{\topsep }{0pt}
		\setlength{\leftmargin}{0.2in}
}
\newenvironment{OneLiners}[1][\ensuremath{\bullet}]
    {\begin{list}
        {#1}
        {\initOneLiners}}
    {\end{list}}
		
\newcommand{\edgeedge}{edge/edge\xspace}
\newcommand{\edgeedgeswap}{edge/edge swap\xspace}
\newcommand{\Edgeedgeswap}{Edge/edge swap\xspace}
\newcommand{\edgeset}{edge/set\xspace}
\newcommand{\edgesetswap}{edge/set swap\xspace}

\newcommand{\pathset}{path/set\xspace}
\newcommand{\pathsetswap}{path/set swap\xspace}

\newcommand{\moveSet}{\mathfrak{M}}


\makeatletter
\def\@fnsymbol#1{\ensuremath{\ifcase#1\or *\or \dagger\or \ddagger\or
   \mathsection\or \mathparagraph\or \uparrow\or \|\or \sharp\or \natural\else\@ctrerr\fi}}

\makeatother

\usepackage{afterpage}
\usepackage{wrapfig}

\begin{document}

\title{A Local-Search Algorithm for Steiner Forest}  

\author{Martin Gro\ss\thanks{Institut f\"ur Mathematik, Technische Universit\"at Berlin, \url{gross@math.tu-berlin.de}.}~\thanks{Supported by the DFG within project A07 of CRC TRR 154.} 
\and Anupam Gupta\thanks{Department of Computer Science, Carnegie Mellon University, \url{anupamg@cs.cmu.edu}.} 
\and Amit Kumar\thanks{Department of Computer Science and Engineering, Indian Institute of Technology, \url{amitk@cse.iitd.ernet.in}.} 
\and Jannik Matuschke\thanks{ TUM School of Management, Technische Universit\"{a}t M\"{u}nchen, \url{jannik.matuschke@tum.de}.}\hspace*{0.2cm}\thanks{Partly supported by the German Academic Exchange Service (DAAD).} 
\and Daniel R. Schmidt\thanks{Institut f{\"u}r Informatik, Universit{\"a}t zu K{\"o}ln, \url{schmidt@informatik.uni-koeln.de}.}~\footnotemark[6]
\and Melanie Schmidt\thanks{Institut f\"ur Informatik, Universit\"{a}t Bonn, \url{melanieschmidt@uni-bonn.de}.}~\footnotemark[6] 
\and \and Jos\'{e} Verschae\thanks{Facultad de Matem\'{a}ticas \& Escuela de Ingenier\'ia, Pontificia Universidad Cat\'{o}lica de Chile, \url{jverschae@uc.cl}. \hspace*{0.5cm} Partly supported by Nucleo Milenio Informaci\'on y Coordinaci\'on en Redes ICM/FIC P10-024F.}
}
\maketitle
\thispagestyle{empty}

\begin{abstract}
  In the \emph{Steiner Forest} problem, we are given a graph and a collection of 
source-sink pairs, and the goal is to find a subgraph of minimum total length 
such that all pairs are connected. The problem is APX-Hard and can be 
$2$-approximated by, e.g., the elegant primal-dual algorithm of Agrawal, Klein, 
and Ravi from 1995.

\medskip\noindent We give a local-search-based constant-factor approximation for the problem. 
Local search brings in new techniques to an area that has for long not seen any 
improvements and might be a step towards a combinatorial algorithm for the more 
general survivable network design problem. Moreover, local search was an 
essential tool to tackle the dynamic MST/Steiner Tree problem, whereas dynamic 
Steiner Forest is still wide open.

\medskip\noindent It is easy to see that any constant factor local search algorithm requires steps 
that add/drop many edges together. We propose natural local moves which, at each step, 
either (a) add a shortest path in the current graph and then drop a bunch of 
inessential edges, or (b) add a set of edges to the current solution. This 
second type of moves is motivated by the potential function we use to measure 
progress, combining the cost of the solution with a penalty for each connected 
component. Our carefully-chosen local moves and potential function work in 
tandem to eliminate bad local minima that arise when using more traditional 
local moves.

\medskip\noindent Our analysis first considers the case where the local optimum is a single tree, 
and shows optimality w.r.t.\ moves that add a single edge (and drop a set of 
edges) is enough to bound the locality gap. For the general case, we show how to 
``project'' the optimal solution onto the different trees of the local optimum 
without incurring too much cost (and this argument uses optimality w.r.t. both 
kinds of moves), followed by a tree-by-tree argument. We hope both the potential 
function, and our analysis techniques will be useful to develop and analyze 
local-search algorithms in other contexts. 
\end{abstract}

\newpage

\setcounter{page}{1}

\section{Introduction}
\label{sec:intro}

The Steiner Forest problem is the following basic network design
problem: given a graph $G = (V,E)$ with edge-lengths $d_e$, and a
set of source-sink pairs $\{ \{s_i, t_i\} \}_{i = 1}^k$, find a 
subgraph $H$ of minimum total length such that each $\{s_i, t_i\}$
pair lies in the same connected component of $H$. This problem
generalizes the Steiner Tree problem, and hence is APX-hard.  The
Steiner Tree problem has a simple $2$-approximation, namely the
minimum spanning tree on the terminals in the metric completion;
however, the forest version does not have such obvious
algorithms.

Indeed, the first approximation algorithm for this problem was a
sophisticated and elegant primal-dual $2$-approximation due to Agrawal,
Klein, and Ravi~\cite{AKR95}. Subsequently, Goemans and Williamson
streamlined and generalized these ideas to many other constrained
network design problems~\cite{GW95}. These results prove an integrality
gap of $2$ for the natural cut-covering LP. Other proofs of this
integrality gap were given in~\cite{Jain98, CS08}.  No better LP
relaxations are currently known (despite attempts in,
e.g.,~\cite{KLS05,KLSZ08}), and improving the approximation guarantee of
$2$ remains an outstanding open problem. Note that all known
constant-factor approximation algorithms for Steiner Forest were based on linear
programming relaxations, until a recent greedy algorithm~\cite{GK15}. In
this paper, we add to the body of techniques that give constant-factor
approximations for Steiner Forest.
The main result of this paper is the following:
\begin{restatable}{theorem}{maintheorem}
  \label{thm:main}
	There is a (non-oblivious) local search algorithm for Steiner Forest with a constant locality gap. It can be implemented to run in polynomial time.
\end{restatable}

The Steiner Forest problem is a basic network problem whose
approximability has not seen any improvements in some time. We explore
new techniques to attacking the problem, with the hope
that these will give us more insights into its structure. Moreover,
for many problems solved using the constrained forest approach
of~\cite{GW95}, the only constant factor approximations known are via
the primal-dual/local-ratio approach, and it seems useful to bring in
new possible techniques.
Another motivation for our work is to make progress towards obtaining
combinatorial algorithms for the survivable network design problem.  In
this problem, we are given connectivity requirements between various
source-sink pairs, and we need to find a minimum cost subset of edges
which provide this desired connectivity. Although we know a
2-approximation algorithm for the survivable network design
problem~\cite{Jain98} based on iterative rounding, obtaining a
combinatorial constant-factor approximation algorithm for this problem
remains a central open
problem~\cite{williamson2011design}.
So far, all approaches of extending primal-dual or greedy algorithms to
survivable network design have only had limited success.  Local search
algorithms are more versatile in the sense that one can easily
\emph{propose} algorithms based on local search for various network
design problems. Therefore, it is important to understand the power of
such algorithms in such settings. We take a step towards this goal by
showing that such ideas can give constant-factor approximation
algorithms for the Steiner Forest problem.

Finally, we hope this is a step towards solving the \emph{dynamic
  Steiner Forest} problem. In this problem, terminal pairs arrive online
and we want to maintain a constant-approximate Steiner Forest while
changing the solution by only a few edges in each update. Several of the
approaches used for the Steiner \emph{Tree} case (e.g., in~\cite{MSVW12,
  GuGK13, LOPSZ15}) are based on local-search, and we hope our local-search
algorithm for Steiner Forest in the offline setting will help solve the
dynamic Steiner Forest problem, too.

\subsection{Our Techniques}

One of the challenges with giving a local-search algorithm for Steiner
Forest is to find the right set of moves. Indeed, it is easy to see that
simple-minded approaches like just adding and dropping a constant number of edges at each step is not
enough. E.g., in the example of Figure~\ref{fig:intro}, the improving moves must add an edge
and remove multiple edges. (This holds even if we take the metric
completion of the graph.) 
We therefore consider a natural generalization of simple edge swaps in which we allow to add paths and remove multiple edges from the induced cycle.

\textbf{Local Moves:} Our first task is to find the ``right''
 moves that add/remove many edges in each ``local'' step. At any step of the algorithm, our
algorithm has a feasible forest, and performs one of these local moves
(which are explained in more detail in \S\ref{sec:LS-moves}):
\begin{OneLiners}
\item \textbf{\edgesetswap{}s}: Add an edge to a tree in the forest, and
  remove one or more edges from the cycle created.
\item \textbf{\pathsetswap{}s}: Add a shortest-path between two vertices of
  a tree in the forest, and again remove edges from the cycle created.
\item \textbf{connecting moves}: Connect some trees of the current forest by adding edges between them.
\end{OneLiners}
At the end of the algorithm, we apply the following post-processing step to the local optimum:
\begin{OneLiners}
\item \textbf{clean-up}: Delete all inessential edges. 
  (An edge is \emph{inessential} if dropping it does not alter the
  feasibility of the solution.)
\end{OneLiners}
Given these local moves, the challenge is to bound the locality gap of
the problem: the ratio between the cost of a local optimum and that of
the global optimum.

\textbf{The Potential:}
The connecting moves may seem odd, since they only increase the
length of the solution. However, a crucial insight behind our
algorithm is that we do not perform local search with respect to the
total length of the solution.  Instead we look to improve a different
potential $\phi$. (In the terminology of~\cite{Alimonti94,KhannaMSV98},
our algorithm is a \emph{non-oblivious} local search.)  The potential
$\phi(T)$ of a tree $T$ is the total length of its edges, plus the
distance between the furthest source-sink pair in it, which we call its
\emph{width}. The potential of the forest $\ALG$ is the sum of the
potentials of its trees. We only perform moves that cause the potential
of the resulting solution to decrease. 

In \S\ref{sec:bad-gap} we give an example where performing the
above moves with respect to the total length of the solution gives us
local optima with cost $\Omega(\log n)\cdot \OPT$ --- this example is useful
for intuition for why using this potential helps. Indeed, if we have a
forest where the distance between two trees in the forest is much less
than both their widths, we can merge them and reduce the potential (even
though we increase the total length). So the trees in a local optimum
are ``well-separated'' compared to their widths, an important property for our analysis.

\textbf{The Proof:} We prove the constant locality gap in two conceptual
steps. 

\begin{wrapfigure}{r}{0.7\textwidth}
  \begin{center}
	  \vspace*{-2.2\baselineskip}
    \begin{tikzpicture}
  \tikzstyle{node}=[circle,shade,top color=gray!30,bottom color=gray!70,draw=gray]
  \tikzstyle{optedge}=[very thick,blue,-,dashed]
  \tikzstyle{algedge}=[very thick,black,-,solid]
	\tikzstyle{both}=[postaction={draw,blue!80,dash pattern= on 3pt off 5pt,dash phase=4pt,very thick},black,dash pattern= on 3pt off 5pt,very thick]
  \begin{scope}[yshift=2.5mm,scale=0.6] 
   \matrix (v) [matrix of math nodes, nodes={circle,shade,top color=gray!30,bottom color=gray!70,draw=gray,inner sep=1pt,anchor=base,text depth=.5ex,text height=2ex,minimum width=1.2em,text centered},column sep=5mm,row sep=4.5mm] {
    s_1 & s_2 & t_2 & s_3 & t_3 & |[draw=none,fill=none,shade=none,scale=2.0]| \dots & s_\ell & t_\ell & t_1 \\
   };
   \draw[algedge] (v-1-1) edge node[auto,black] {$\frac{\ell}{k}$} (v-1-2);
   \draw[algedge,transform canvas={yshift=0.6pt}] (v-1-2) -- node[auto] {$1$} (v-1-3);
	 \draw[optedge,transform canvas={yshift=-0.6pt}] (v-1-2) -- (v-1-3);	
	
   \draw[algedge] (v-1-3) edge node[auto,black] {$\frac{\ell}{k}$} (v-1-4);
   \draw[algedge,transform canvas={yshift=0.6pt}] (v-1-4) -- node[auto] {$1$} (v-1-5);
	 \draw[optedge,transform canvas={yshift=-0.6pt}] (v-1-4) -- (v-1-5);	
   \draw[algedge,path fading=east] (v-1-5) -- ++(0.6,0);
   \draw[algedge,path fading=west] (v-1-7) -- ++(-0.6,0);
   \draw[algedge,transform canvas={yshift=0.6pt}] (v-1-7) -- node[auto] {$1$} (v-1-8);
	 \draw[optedge,transform canvas={yshift=-0.6pt}] (v-1-7) -- (v-1-8);	
   \draw[algedge] (v-1-8) edge node[auto,black] {$\frac{\ell}{k}$} (v-1-9);
   \draw[optedge] (v-1-1) -- ++(0,-1) -- node[auto,swap,black,above,pos=0.55] {$\ell$} ($(v-1-9)+(0,-1)$) -- (v-1-9);
  \end{scope}
 \end{tikzpicture}
	  \vspace*{-1.05\baselineskip}
  \end{center}
  \caption{\footnotesize \it The black edges (continuous lines) are the current solution. If $\ell \gg k$, we should move to the blue forest (dashed lines), but any improving move must change $\Omega(k)$ edges. Details can be found in \S\ref{appendix:introexample}\label{fig:intro}.}
	  \vspace*{-0.95\baselineskip}
\end{wrapfigure}
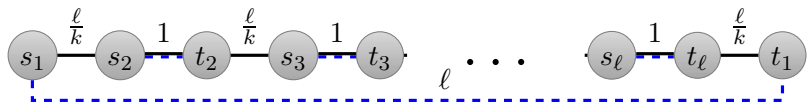
As the first step, we assume that the local optimum happens to be a
single tree. In this case we show that the essential edges of this tree
$T$ have cost at most $\mathcal{O}(\OPT)$---hence the final removal of
inessential edges gives a good solution. To prove this, we need to
charge our edges to $\OPT$'s edges. However, we cannot hope to charge
single edges in our solution to single edges in $\OPT$---we need to
charge multiple edges in our solution to edges of $\OPT$. (We may just
have more edges than $\OPT$ does. More concretely, this happens in the
example from Figure~\ref{fig:intro}, when $\ell = \Theta(k)$ and we are
at the black tree and $\OPT$ is the blue forest.)
So 
we consider \edgesetswap{}s that try to swap some subset $S$ of $T$'s
edges for an edge $f$ of $\OPT$. Since such a swap is non-improving at a
local optimum, the cost of $S$ is no more than that of $f$. Hence, we
would like to partition $T$'s edges into groups and find an $O(1)$-to-$1$
map of groups to edges of $\OPT$ of no less cost. Even if we cannot find
an explicit such map, it turns out that Hall's theorem is the key to
showing its existence.

Indeed, the intuition outlined above works out quite nicely if we
imagine doing the local search with respect to the total length instead
of the potential. The main technical ingredient is a partitioning of our
edges into equivalence classes that behave (for our purposes) ``like
single edges'', allowing us to apply a Hall-type argument.  This idea is
further elaborated in \S\ref{sec-treecase-d} with detailed proofs in
\S\ref{appendix:addproofs}. However, if we go back to considering the
potential, an \edgesetswap adding $f$ and removing $S$ may create
multiple components, and thus increase the width part of the potential.
Hence we give a more delicate argument showing that similar charging
arguments work out: basically we now have to charge to the width of the
globally optimal solution as well. A detailed synopsis is presented in
\S\ref{sec-treecase-phi}, and the proofs are in \S\ref{sec-tree-phi}.

The second conceptual step is to extend this argument to the case where
we can perform all possible local moves, and the local optimum is a
forest $\ALG$. If $\OPT$'s connected components are contained in those of
$\ALG$, we can do the above analysis for each $\ALG$-component
separately. So imagine that $\OPT$ has edges that go between vertices in
different components of $\ALG$. We simply give an algorithm that takes
$\OPT$ and ``projects'' it down to another solution $\OPT'$ of comparable
cost, such that the new projected solution $\OPT'$ has connected
components that are contained in the components of $\ALG$.  We find the
existence of a cheap projected solution quite surprising; our proof
crucially uses the optimality of the algorithm's solution under both
\pathsetswap{}s and connecting moves. Again, a summary of our approach
is in \S\ref{sec:forest}, with proofs in \S\ref{sec:app-forest}.

\medskip \textbf{Polynomial-time Algorithm.} 
The locality gap with respect to the above moves is at most $46$. 
Finally, we show that the swap moves can be implemented in polynomial time, and the connecting moves can be approximated to within constant factors. 
Indeed, a $c$-approximation for \emph{weighted} $k$-MST gives a $23(1+c)+ \varepsilon$-approximation. 
Applying a weighted version of Garg's $2$-approximation~\cite{G05,G16} yields $c=2$. The resulting approximation guarantee is $69$ (compared to $96$ for~\cite{GK15}).
Details on this can be found in \S\ref{sec:polytime}.

\subsection{Related Work}
\label{sec:related-work}

Local search techniques have been very successful for providing good
approximation guarantees for a variety of problems: e.g., network design
problems such as low-degree spanning trees~\cite{FurerR94}, min-leaf
spanning trees~\cite{LR96,SO98}, facility location and related
problems, both uncapacitated~\cite{KPR00,AGKMMP04} and
capacitated~\cite{PalTW01}, geometric $k$-means~\cite{KMNPSW04}, mobile
facility location~\cite{AFS13}, and scheduling problems~\cite{PS15}.
Other examples can be found in, e.g., the book of Williamson and Shmoys~\cite{williamson2011design}. 
More recent are applications to optimization problems on planar and low-dimensional instances~\cite{CM15,CG15}.
In particular, the new PTAS for low dimensional k-means in  is based on local search~\cite{cohen2016local,friggstad2016local}.

Local search algorithms have also been very successful in practice --
e.g., the widely used Lin-Kernighan heuristic~\cite{LK73} for the
travelling salesman problem, which has been experimentally shown to
perform extremely well~\cite{H00}.

Imase and Waxman~\cite{IW91} defined the dynamic Steiner tree problem
where vertices arrive/depart online, and a few edge changes are performed 
 to maintain a near-optimal solution. Their analysis was
improved by~\cite{MSVW12, GuGK13, GK14-steiner, LOPSZ15}, but extending it to
Steiner Forest 
 remains wide open.

\section{Preliminaries}
\label{sec:preliminaries}

Let $G=(V,E)$ be an undirected graph with non-negative edge weights $d_e
\in \Rp$. Let $n:=|V|$. For $W \subseteq V$, let $G[W]=(W,E[W])$ be the
vertex-induced subgraph, and for $F \subseteq E$, $G[F]=(V[F],F)$ the 
edge-induced subgraph, namely the graph consisting of the edges in $F$ and
the vertices contained in them. A forest is a set of edges $F \sse E$
such that $G[F]$ is acyclic. 

For a node set $W \subseteq V$ and an edge set $F \subseteq E$, let
$\delta_F(W)$ denote the edges of $F$ leaving $W$.
Let $\delta_F(A:B) := \delta_F(A) \cap \delta_F(B)$ for two disjoint node sets
$A, B \subseteq V$ be the set of edges that go between $A$ and
$B$. For forests $F_1,F_2 \subseteq E$ we use
$\delta_F(F_1:F_2) := \delta_F(V[F_1]:V[F_2])$. We may drop the
subscript if it is clear from the context.


Let $\terms \subseteq \{ \{v,\bar{v}\} \mid v,\bar v \in V \}$ be a set
of terminal pairs.  Denote the shortest-path distance between $u$ and
$\bar{u}$ in $(G,d)$ by $\dist_d(u,\bar{u})$.  Let $n_t$ be the number
of terminal pairs. We number the pairs according to non-decreasing
shortest path distance (ties broken arbitrarily). Thus, $\terms = \{
\{u_1,\bar{u}_1\}, \ldots, \{u_{n_t},\bar{u}_{n_t}\}\}$ and $i < j$
implies $\dist_d(u_i,\bar{u}_i) \le \dist_d(u_j,\bar{u}_j)$. This
numbering ensures consistent tie-breaking throughout the paper.
We say that $G=(V,E)$, the weights $d$ and $\terms$ form a \emph{Steiner
  Forest instance}.  
We often use $\ALG$ to denote a feasible Steiner forest held by our
algorithm and $\REFF$ to denote an optimal/good feasible solution to
which we compare $\ALG$.



\paragraph{Width.}
Given a connected set of edges $E'$, the \emph{width} $w(E')$ of $E'$ is
the maximum distance (in the original graph) of any terminal pair connected
by~$E'$: \ie $w(E')=\max_{\set{u,\bar{u}}\in\terms, u,\bar{u}\in V[E']}
\dist_d(u,\bar{u})$. Notice that $w(E')$ is the width of the pair
$\{u_i,\bar{u}_i\}$ with the largest $i$ among all pairs in $V[E']$. We
set $\maxit(E') := \max\cset{i}{u_i, \bar{u}_i \in V[E']}$, \ie $w(E') =
\dist_d(u_{\maxit(E')},\bar{u}_{\maxit(E')})$.
	
For a subgraph $G[F]=(V[F],F)$ given by $F\subseteq E$ with connected
components $E_1,\dots,E_l \subseteq F$, we define the \emph{total width}
of $F$ to be the sum $w(F) := \sum_{i=1}^l w(E_i)$ of the widths of
its connected components. Let $d(F) := \sum_{e \in F} d_e$ be the sum of
edge lengths of edges in $F$ and define
\[ \phi(F) := d(F) + w (F). \]
By the definition of the width, it follows that $d(F) \leq \phi(F) \leq 2d(F)$.
  

\section{The Local Search Algorithm} 
\label{sec:LS-moves}

Our local-search algorithm starts with a feasible solution $\ALG$, and
iteratively tries to improve it. Instead of looking at the actual edge
cost $d(\ALG)$, we work with the potential $\phi(\ALG)$ 
and decrease it over time. 

In the rest of the paper, we say a move changing $\ALG$ into $\ALG'$ is
\emph{improving} if $\phi(\ALG') < \phi(\ALG)$. A solution $\ALG$ is
\texttt{<move>}-\emph{optimal} with respect to certain kind of move if no
moves of that kind are improving.

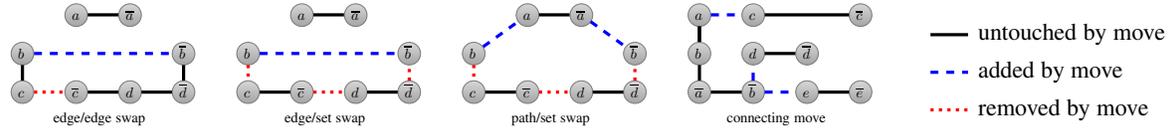
\begin{figure}[htbp]
 \centering
 \begin{tikzpicture}
  \tikzstyle{node}=[circle,shade,top color=gray!30,bottom color=gray!70,draw=gray]
  \tikzstyle{addedge}=[very thick,blue,-,dashed]
	\tikzstyle{removeedge}=[very thick, red, dotted] 
  \tikzstyle{algedge}=[very thick,black,-,solid]
	\tikzstyle{both}=[postaction={draw,blue!80,dash pattern= on 3pt off 5pt,dash phase=4pt,very thick},black,dash pattern= on 3pt off 5pt,very thick]
  \begin{scope}[yshift=2.5mm,scale=0.5,every node/.style={scale=0.5}]
   \matrix (v) [matrix of math nodes, nodes={circle,shade,top color=gray!30,bottom color=gray!70,draw=gray,inner sep=1pt,anchor=base,text depth=.5ex,text height=2ex,text width=0.6em,minimum size=6mm},column sep=4mm,row sep=2mm] {
      & a & \overline{a} &              \\
		b &   &              & \overline{b} \\
		c & \overline{c} & d & \overline{d} \\
   };
   \draw[algedge] (v-1-2) edge (v-1-3);
	 \draw[algedge] (v-2-1) edge (v-3-1);
	 \draw[removeedge] (v-3-1) edge (v-3-2);
	 \draw[algedge] (v-3-2) edge (v-3-3);
	 \draw[algedge] (v-3-3) edge (v-3-4);
	 \draw[algedge] (v-3-4) edge (v-2-4);
	 \draw[addedge] (v-2-1) edge (v-2-4);
	 \node[anchor=north,xshift=4mm,yshift=-2mm] (text) at (v-3-2.south east) {\edgeedgeswap};
  \end{scope}
  \begin{scope}[yshift=2.5mm,xshift=3.0cm,scale=0.5,every node/.style={scale=0.5}]
   \matrix (v) [matrix of math nodes, nodes={circle,shade,top color=gray!30,bottom color=gray!70,draw=gray,inner sep=1pt,anchor=base,text depth=.5ex,text height=2ex,text width=0.6em,minimum size=6mm},column sep=4mm,row sep=2mm] {
      & a & \overline{a} &              \\
		b &   &              & \overline{b} \\
		c & \overline{c} & d & \overline{d} \\
   };
   \draw[algedge] (v-1-2) edge (v-1-3);
	 \draw[removeedge] (v-2-1) edge (v-3-1);
	 \draw[algedge] (v-3-1) edge (v-3-2);
	 \draw[removeedge] (v-3-2) edge (v-3-3);
	 \draw[algedge] (v-3-3) edge (v-3-4);
	 \draw[removeedge] (v-3-4) edge (v-2-4);
	 \draw[addedge] (v-2-1) edge (v-2-4);
	 \node[anchor=north,xshift=4mm,yshift=-2mm] (text) at (v-3-2.south east) {\edgesetswap};
  \end{scope}	
  \begin{scope}[yshift=2.5mm,xshift=6cm,scale=0.5,every node/.style={scale=0.5}]
   \matrix (v) [matrix of math nodes, nodes={circle,shade,top color=gray!30,bottom color=gray!70,draw=gray,inner sep=1pt,anchor=base,text depth=.5ex,text height=2ex,text width=0.6em,minimum size=6mm},column sep=4mm,row sep=2mm] {
      & a & \overline{a} &              \\
		b &   &              & \overline{b} \\
		c & \overline{c} & d & \overline{d} \\
   };
   \draw[algedge] (v-1-2) edge (v-1-3);
	 \draw[removeedge] (v-2-1) edge (v-3-1);
	 \draw[algedge] (v-3-1) edge (v-3-2);
	 \draw[removeedge] (v-3-2) edge (v-3-3);
	 \draw[algedge] (v-3-3) edge (v-3-4);
	 \draw[removeedge] (v-3-4) edge (v-2-4);
	 \draw[addedge] (v-2-1) edge (v-1-2);
	 \draw[addedge] (v-1-3) edge (v-2-4);
	 \node[anchor=north,xshift=4mm,yshift=-2mm] (text) at (v-3-2.south east) {\pathsetswap};
  \end{scope}		
  \begin{scope}[yshift=2.5mm,xshift=9.0cm,scale=0.5,every node/.style={scale=0.5}]
   \matrix (v) [matrix of math nodes, nodes={circle,shade,top color=gray!30,bottom color=gray!70,draw=gray,inner sep=1pt,anchor=base,text depth=.5ex,text height=2ex,text width=0.6em,minimum size=6mm},column sep=4mm,row sep=2mm] {
    a            & c            &              & \overline{c} \\
		b            & d            &  \overline{d}  &          \\
		\overline{a} & \overline{b} &  e & \overline{e}\\
   };
   \draw[algedge] (v-1-1) edge (v-2-1);
	 \draw[algedge] (v-2-1) edge (v-3-1);
	 \draw[algedge] (v-3-1) edge (v-3-2);
	 \draw[algedge] (v-1-1) edge (v-2-1);
	 \draw[algedge] (v-1-2) edge (v-1-4);
	 \draw[algedge] (v-2-2) edge (v-2-3);
	 \draw[algedge] (v-3-3) edge (v-3-4);		
	 \draw[addedge] (v-1-1) edge (v-1-2);
	 \draw[addedge] (v-3-2) edge (v-3-3);
	 \draw[addedge] (v-3-2) edge (v-2-2);
	 \node[anchor=north,xshift=4mm,yshift=-2mm] (text) at (v-3-2.south east) {connecting move};
  \end{scope}			
	\begin{scope}[xshift=11.0cm,yshift=0.5cm]
	 \draw[algedge] (0,0) -- (0.5,0); \node[anchor=west] at (0.5,0) {\footnotesize untouched by move};
	 \draw[addedge] (0,-0.5) -- (0.5,-0.5); \node[anchor=west] at (0.5,-0.5) {\footnotesize added by move};
	 \draw[removeedge] (0,-1) -- (0.5,-1); \node[anchor=west] at (0.5,-1) {\footnotesize removed by move};
	\end{scope}
 \end{tikzpicture}
 \caption{Our different moves.}
 \label{fig:prelimswaps}
\vspace*{-0.5cm}
\end{figure}  

\paragraph{Swaps.}
Swaps are moves that start with a cycle-free feasible solution $\ALG$, add some
edges and remove others to get to another cycle-free feasible solution
$\ALG'$.
\begin{itemize}[leftmargin=0.2in]
\item The most basic swap is: \emph{add an edge $e$ creating a cycle, remove
    an edge $f$ from this cycle}. This is called an \emph{\edgeedgeswap} $(e,f)$. 

\item We can slightly generalize this: \emph{add an edge $e$ creating a
    cycle, and remove a subset $S$ of edges from this cycle
    $C(e)$}. This is called the \emph{\edgesetswap} $(e,S)$. \Edgeedgeswap{}s are a special case of \edgesetswap{}s, so \edgesetswap-optimality implies \edgeedgeswap-optimality.

  There may be many different subsets of $C(e)$ we could remove. A
  useful fact is that if we fix some edge $f \in C(e)$ to remove, this
  uniquely gives a maximal set $R(e,f) \sse C(e)$ of edges that can be
  removed along with $f$ after adding $e$ without violating
  feasibility. Indeed, $R(e,f)$ contains $f$, and also all edges on
  $C(e)$ that can be removed in $\ALG \cup\{e\}\backslash\{f\}$ without
  destroying feasibility.  (See
  Lemma~\ref{lem:comp-class-all-or-nothing} for a formalization.)

  Moreover, given a particular $R(e,f)$, we could remove any subset $S
  \subseteq R(e,f)$. If we were doing local search w.r.t.\ $d(\ALG)$,
  there would be no reason to remove a proper subset. But since the
  local moves try to reduce $\phi(\ALG)$, removing a subset of $R(e,f)$
  may be useful. If $e_1,\ldots,e_{\ell}$ are the edges in $R(e,f)$ in
  the order they appear on $C(e)$, we only need swaps where $S$ consists
  of edges $e_i,\ldots,e_j$ that are consecutive in the above
  order.   There are $\mathcal{O}(n^2)$ sets $S \subseteq R(e,f)$ that are
  consecutive.\footnote{In fact, we only need five different swaps
    $(e,S)$ for the following choices of consecutive sets $S$: The set
    $S=\{f\}$, the complete set $S=R(e,f)$, and three sets of the form
    $S=\{e_1,\ldots,e_{i}\}$, $S=\{e_{i+1},\ldots,e_{j}\}$ and
    $S=\{e_{j+1},\ldots,e_{\ell}\}$ for specific indices $i$ and
    $j$. How to obtain the values for $i$ and $j$ is explained in
    Section~\ref{sec-treecase-d}.} Moreover, there are at most $n-1$
  choices for $e$ and $O(n)$ choices for $f$, so the number of \edgesetswap{}s is polynomial.

\item 
  A further generalization: we can pick two vertices $u, v$ lying in
  some component $T$, add a shortest-path between them (in the current
  solution, where all other components are shrunk down to single points,
  and the vertices/edges in $T\setminus \{u,v\}$ are removed). This
  creates a cycle, and we want to remove some edges. We now imagine that
  we added a ``virtual'' edge $\{u,v\}$, and remove a subset of consecutive edges from some $R(\{u,v\}, f) \sse C(\{u,v\})$, just as if we'd have executed an \edgesetswap with the ``virtual'' edge $\{u,v\}$.
	We call such a swap a \emph{\pathsetswap} $(u,v,S)$.

  Some subtleties: Firstly, the current solution $\ALG$ may already
  contain an edge $\{u,v\}$, but the $uv$-shortest-path we find may be
  shorter because of other components being shrunk. So this move would
  add this shortest-path and remove the direct edge $\{u,v\}$---indeed,
  the cycle $C(uv)$ would consist of two parallel edges, and we'd remove
  the actual edge $\{u,v\}$. 
  Secondly, although the cycle contains edges from many components, only
  edges within $T$ are removed. Finally, there are a polynomial number
  of such moves, since there are $O(n^2)$ choices for $u,v$, $O(n)$
  choices for $f$, and $O(n^2)$ consecutive removal sets $S$.
\end{itemize}

Note that \edgesetswap{}s never decrease the number of connected
components of $\ALG$, but \pathsetswap{}s may increase or decrease the
number of connected components.

\paragraph*{Connecting moves.} Connecting moves reduce the number of
connected components by adding a set of edges that connect some of the
current components. Formally, let $G_{\ALG}^{\text{all}}$ be the
(multi)graph that results from contracting all connected components of
$\ALG$ in $G$, deleting loops and keeping parallel edges.  A
\emph{connecting move} (denoted $conn(T)$) consists of picking a tree in
$G_{\ALG}^{\text{all}}$, and adding the corresponding edges to
$\ALG$. 
The number of possible connecting moves can be large, but
Section~\ref{sec:howto-treemove} discusses how to do this approximately,
using a $k$-MST procedure.
%

Note that connecting moves cause $d(\ALG') > d(\ALG)$, but since our
notion of improvement is with respect to the potential $\phi$, such a
move may still cause the potential to decrease. 

In addition to the above moves, the algorithm runs the following post-processing step at the end.

\paragraph*{Clean-up.} Remove the unique maximal edge set $S \subseteq \ALG$ such that $\ALG \setminus S$ is feasible, i.e., erase all unnecessary edges.  This might increase $\phi(\ALG)$, but it will never increase $d(\ALG)$.



Checking whether an improving move exists is polynomial except for connecting moves, which we can do approximately (see \S\ref{sec:howto-treemove}). Thus, the local search algorithm can be made to run in polynomial time by using standard methods (see \S~\ref{sec:polytime}). \msnote{added this sentence, moved section in appendix}



\section{In Which the Local Optimum is a Single Tree}

We want to bound the cost of a forest that is locally optimal with respect to the moves defined above.
To start, let us consider a simpler case: suppose we were to find a single tree $T$
that 
is optimal with respect to just the \emph{\edgeedge}
and \emph{\edgeset} swaps. (Recall that \edgeset swaps add an edge and
remove a consecutive subset of the edges on the resulting cycle, while
maintaining feasibility. Also, recall that optimality means that no such
moves cause the potential $\phi$ to decrease.) Our main result of this
section is the following:

\begin{restatable}{corollary}{treecaseresult}
  \label{treecaseresult}
  Let $G=(V,E)$ be a graph, let $d_e$ be the cost of edge $e \in E$ and
  let $\terms \subseteq V\times V$ be a set of terminal pairs.  Let
  $\ALG, \REFF \subseteq E$ be two feasible Steiner forests for
  $(G,d,\terms)$ with $V[\ALG]=V[\REFF]$.  Assume that $\ALG$ is a tree
  and that $\ALG$ is swap-optimal with respect to $\REFF$ and $\phi$
  under \edgeedge and \edgeset swaps.  Denote by $\ALG'$ the modified
  solution where all inessential edges have been dropped from $\ALG$.
  Then,
  \begin{align*}
    d(\ALG') \leq 10.5\cdot d(\REFF) + w(\REFF) \le 11.5 \cdot d(\REFF).
  \end{align*}
\end{restatable}

The actual approximation guarantee is~42 for this case: indeed,
Corollary~\ref{treecaseresult} assumes $V[\ALG]=V[\REFF]$, which can be
achieved (by taking the metric completion on the terminals) at the cost of a factor $2$.

The intuition here comes from a proof for the optimality of
\edgeedgeswap{}s for the Minimum Spanning tree problem.  Let $\ALG$ be
the tree produced by the algorithm, and $\REFF$ the reference (i.e.,
optimal or near-optimal) solution, with $V[\ALG]=V[\REFF]$. Suppose we
were looking for a minimum spanning tree instead of a Steiner forest: one way to show that \edgeedgeswap{}s
lead to a global optimum is to build a bipartite graph whose vertices are
the edges of $\ALG$ and $\REFF$, and which contains edge $(e,f)$ when $f \in \REFF$
can be swapped for $e \in \ALG$ and $d_e \leq d_f$. Using the fact that
all edge/edge swaps are non-improving, we can show that there exists a
perfect matching between the edges in $\ALG$ and $\REFF$, and hence the
cost of $\ALG$ is at most that of $\REFF$.

Our analysis is similar in spirit. Of course, we now have to (a)
consider edge/\emph{set} swaps, (b) do the analysis with respect to the
potential $\phi$ instead of just edge-lengths, and (c) we cannot hope to
find a perfect matching because the problem is NP-hard. These issues make the
proofs more complicated, but the analogies still show through.  

\subsection{An approximation guarantee for trees and \texorpdfstring{$d$}{d}}\label{sec-treecase-d}

In this section, we conduct a thought-experiment where we imagine that we get
a connected tree on the terminals which is optimal for \edgesetswap{}s
\emph{with respect to just the edge lengths, not the potential}.  In
very broad strokes, we define an equivalence relation on the edges of
$\ALG$, and show a constant-to-1 cost-increasing map from the resulting equivalence classes
to edges of $\REFF$---again mirroring the MST analysis---and hence
bounding the cost of $\ALG$ by a constant times the cost of $\REFF$.
The analysis of the real algorithm in \S\ref{sec-treecase-phi} builds on the insights we develop here.

\textbf{Some Definitions.}  The crucial equivalence relation is defined as
follows: For edges $e,f \in \ALG$, let $T_{e,f}$ be the connected
component of $\ALG \setminus \{e,f\}$ that contains the unique $e$-$f$-path in
$\ALG$. We say $e$ and $f$ are \emph{compatible w.r.t.\ $\REFF$} if
$e=f$ or if there are no \REFF-edges leaving~$T_{e,f}$, and denote it by
$e \sim_{cp} f$. In Lemma~\ref{lem:compatibility-is-transitive} we show
that $\sim_{cp}$ is an equivalence relation, and denote the set of
equivalence classes by $\eqS$.


An edge is \emph{essential} if dropping it makes the solution infeasible.
If $T_1, T_2$ are the connected components of $\ALG\setminus\{e\}$, then
$e$ is called \emph{safe} if at least one edge from $\REFF$ crosses between
$T_1$ and $T_2$. Observe that any essential edge is safe, but the
converse is not true: safe edges can be essential or
inessential. However, it turns out that the set $S_{u}$ of all unsafe edges
in $\ALG$ forms an equivalence class of $\sim_{cp}$. Hence, all other
equivalence classes in $\eqS$ contain only safe edges. Moreover, these
equivalence classes containing safe edges behave like single edges in
the following sense. (Proof in Lemma~\ref{lem:summary:pathlemmata} in \S\ref{appendix:addproofs}.)

\begin{OneLiners}
  \item[(1)] Each equivalence class $S$ lies on a path in $\ALG$.
  \item[(2)] For any edge $f \in \REFF$, either $S$ is completely contained
    in the fundamental cycle $C_\ALG(f)$ obtained by adding $f$ to $\ALG$, or
    $S \cap C_\ALG(f) = \emptyset$.
  \item[(3)] If $(\ALG\setminus\{e\})\cup\{f\}$ is feasible, and $e$
    belongs to equivalence class $S$, then
    $(\ALG\setminus S)\cup\{f\}$ is feasible. (This last property also trivially holds
    for $S=S_u$.)
\end{OneLiners}

\medskip \textbf{Charging.} We can now give the bipartite-graph-based
charging argument sketched above.

	
\begin{theorem}\label{thm:four-approx-for-trees-hall}
  Let $I=(V,E,\terms,d)$ be a Steiner Forest instance and let $\REFF$ be
  a feasible solution for $I$.  Furthermore, let $\ALG \subseteq E$ be a
  feasible \emph{tree} solution for $I$.  Assume that
  $V[\REFF]=V[\ALG]$.  Let $\Delta: \eqS \to \R$ be a cost function that
  assigns a cost to all $S \in \eqS$.  Suppose that $\Delta(S) \leq d_f$
  for all pairs of $S\in \eqS\setminus \{S_u\}$ and $f \in \REFF$ such
  that the cycle in $\ALG\cup \{f\}$ contains $S$.  Then, 
  \[
  \sum_{S \in \eqS\setminus \{S_u\}} \Delta(S) \leq \frac{7}{2}\cdot \sum_{f\in \REFF} d_f.
  \]
\end{theorem}
	
\newcommand{\ALGprime}{\ALG'}
\newcommand{\Aprime}{X}
\begin{proof}
  Construct a bipartite graph $H = (A\cup B, E(H))$ with nodes
  $A:=\cset{a_S}{S \in \eqS\setminus\{S_u\}}$ and $B:=\cset{b_f}{f \in
    \REFF}$. Add an edge $\{a_S, b_f\}$ whenever $f$ closes a cycle in
  $\ALG$ that contains $S$.  By our assumption,  if $\{a_S,
  b_f\} \in E(H)$ then $\Delta(S) \leq d_f$.  Suppose that we can show that $\frac{7}{2} \cdot |N(\Aprime)| \geq |\Aprime|$ for all $\Aprime \subseteq A$, where $N(\Aprime) \subseteq B$ is the set of neighbors of nodes in $\Aprime$.  By a generalization of Hall's Theorem (stated as Fact~\ref{fact:halltype} in \S\ref{appendix:addproofs}), this
  condition implies that there is an assignment $\alpha: E \to \mathbb{R}_+$ such that $\sum_{e \in \delta_H(a)} \alpha(e) \geq 1$ for all $a \in A$ and $\sum_{e \in \delta_H(b)} \alpha(e) \leq \frac{7}{2}$ for all $b \in B$. Hence
    \begin{align*}
      \smashoperator[r]{\sum_{S \in \eqS\setminus\{S_u\}}} \Delta(S)
       \leq \sum_{S \in \eqS\setminus\{S_u\}} \smashoperator[r]{\sum_{e \in \delta_H(a_S)}} \alpha(e) \Delta(S)
       =\sum_{f \in \REFF} \smashoperator[r]{\sum_{e \in \delta_H(b_f)}} \alpha(e) \Delta(S)
       \leq \sum_{f \in \REFF} \smashoperator[r]{\sum_{e \in \delta_H(b_f)}} \alpha(e) d_f
       \leq \frac{7}{2} \sum_{f \in \REFF} d_f.
  \end{align*}
  It remains to show that $\frac{7}{2} \cdot |N(\Aprime)| \geq
  |\Aprime|$ for all $\Aprime \subseteq A$. To that aim, fix $\Aprime
  \subseteq A$ and define $\eqS' := \{S \mid a_S \in \Aprime\}$. In a
  first step, contract all $e \in U:=\bigcup_{S \in \eqS \setminus
    \eqS'} S$ in $\ALG$, and denote the resulting tree by $\ALGprime :=
  \ALG\contract U$.\footnote{Formally, we define the graph $G[T]/ e
    =(V[T]/ e, T/ e)$ for a tree $T$ by $V[T]/ e := V[T] \cup \{ uv \}
    \setminus \{u,v\}$ and $T / e := T \setminus \delta(\{u,v\}) \cup
    \{\{w,uv\} \mid \{u,w\} \in T \,\vee\, \{v,w\} \in T \}$ for an edge
    $e=\{u, v\}\in E$, then set $G/ U := G/ e_1 / e_2 / \ldots / e_k$
    for $U=\{e_1,\ldots,e_k\}$ and let $T/ U$ be the edge set of this
    graph. If $U \subseteq T$, then the contraction causes no loops or
    parallel edges, otherwise, we delete all loops or parallel edges.}
  Note that edges in each equivalence class are either all contracted or
  none are contracted. Also note that all unsafe edges are contracted,
  as $S_u \notin \eqS'$. Apply the same contraction to $\REFF$ to obtain
  $\REFF' := \REFF \contract U$, from which we remove all loops and
  parallel edges.
  
  Let $f \in \REFF'$. Since $\ALGprime$ is a tree, $f$ closes a cycle $C$
  in $\ALGprime$ containing at least one edge $e \in \ALGprime$. Denoting
  the equivalence class of $e$ by $S_e$, observing that all edges in
  $\ALGprime$ are safe, and using property (2) given above
  (formally, using Lemma~\ref{lemsum:property:cycle-all-or-nothing}), we
  get that cycle $C$ contains $S_e$. Hence the node $b_f \in B$
  corresponding to $f$ belongs to $N(a_{S_e}) \sse N(\Aprime)$. Thus,
  $|N(\Aprime)| \geq |\REFF'|$ and it remains to show that
  $\frac{7}{2}|\REFF'| \geq |\Aprime|$.
  
  We want to find a unique representative for each $a_S \in \Aprime$. So
  we select an arbitrary root vertex $r \in V[\ALGprime]$ and orient all
  edges in $\ALGprime$ away from $r$. Every non-root vertex now has
  exactly one incoming edge. Every equivalence class $S \in \eqS'$
  consists only of safe edges, so it lies on a path (by
  Lemma~\ref{lemsum:property:compatible-path}). Consider the two
  well-defined \emph{endpoints} which are the outermost vertices of $S$
  on this path. For at least one of them, the unique incoming edge must
  be an edge in $S$. We represent $S$ by one of the endpoints which has
  this property and call this representative $r_S$. Let $R \subseteq
  V[\ALGprime]$ be the set of all representative nodes. Since every vertex
  has an unique incoming edge, $S \neq S'$ implies that $r_S \neq
  r_{S'}$. Hence $|R| = |\eqS'| = |\Aprime|$.
  Moreover, let $R_1$ and $R_2$ be the representatives with degrees $1$
  and $2$ in $\ALGprime$, and $L$ be the set of leaves of $\ALGprime$. As
  the number of vertices of degree at least $3$ in a tree is bounded by
  the number of its leaves, the number of representatives of degree at
  least $3$ in $\ALGprime$ is bounded by $|L|$. So
  $|\Aprime| \leq |R_1| + |R_2| + |L|$.
  
   We now show that every $v \in R_2 \cup L$ is incident to an edge in $\REFF'$. First, consider any $v \in L$ and let $e$ be the only edge in $\ALGprime$ incident to $v$. As $e$ is safe, there must be an edge $f \in \REFF'$ incident to $v$. Now consider any $r_S \in R_2$ and let $e_1, e_2 \in \ALGprime$ be the unique edges incident to $r_S$. Because $r_S$ is the endpoint of the path corresponding to the equivalence class $S$, the edges $e_1$ and $e_2$ are not compatible. Hence there must be an edge $f \in \REFF'$ incident to $r_S$. Because $R_2$ and $L$ are disjoint and every edge is incident to at most two vertices, we conclude that 
	$|\REFF'| \geq (|R_2| + |L|)/2$. 
   
   Finally, we show that $|\REFF'| \geq \frac{2}{3} |R_1|$. Let $\mathcal{C}$ be the set of connected components of $\REFF'$ in $G \contract U$. Let $\mathcal{C}' := \{T \in \mathcal{C} \mid |V[T] \cap R_1| \leq 2\}$ and $\mathcal{C}'' := \{T \in \mathcal{C} \mid |V[T] \cap R_1| > 2\}$. Note that no representative $r_S \in R_1$ is a singleton as every leaf of $\ALGprime$ is incident to an edge of $\REFF'$. 
   We claim that $|T| \geq |V[T] \cap R_1|$ for every $T \in \mathcal{C}'$.
   Assume by contradiction that this was not true and let $T \in \mathcal{C}'$ with $|T| < |V(T) \cap R_1|$. This means that $V[T] \cap R_1$ contains exactly two representatives $r_{S}, r_{S'} \in R_1$ and $T$ contains only the edge $\{r_S, r_{S'}\}$. Let $e \in S$ and $e' \in S'$ be the edges of $\ALGprime$ incident to $r_{S}$ and $r_{S'}$, respectively. As $e$ and $e'$ are not compatible, there must be an edge $f \in \REFF'$ with exactly one endpoint in $\{r_{S}, r_{S'}\}$, a contradiction as this edge would be part of the connected component $T$. We conclude that $|T| \geq |V[T] \cap R_1|$ for every $T \in \mathcal{C}'$. Additionally, we have that $|T| \geq |V[T]| - 1 \geq \frac{2}{3}|V[T]|$ for all $T \in \mathcal{C}''$ as $|V[T]| > 2$. Therefore,
   \begin{align*}
     |\REFF'| & = \sum_{T \in \mathcal{C}} |T| \geq \sum_{T \in \mathcal{C}'} |V[T] \cap R_1| + \sum_{T \in \mathcal{C}''} \frac{2}{3}|V[T] \cap R_1| \geq \frac{2}{3} |R_1|.
   \end{align*}
	The three bounds together imply $|\Aprime| \leq |R_1| + |R_2| + |L| \leq \frac{3}{2}|\REFF'| + 2|\REFF'| = \frac{7}{2} |\REFF'|$. 
  \end{proof}

  We obtain the following corollary of Theorem~\ref{thm:four-approx-for-trees-hall}.
  \begin{corollary}\label{cor:treecase}
    Let $I=(V,E,\terms,d)$ be a Steiner Forest instance and let 
    $\OPT$ be a solution for $I$ that minimizes $d(\OPT) = \sum_{e \in \OPT} d_e$.
    Let $\ALG \subseteq E$ be feasible tree solution for $I$ that does not 
    contain inessential edges. 
		Assume $V[\ALG]=V[\OPT]$.
    If $\ALG$ is \edgeedge and \edgeset swap-optimal with respect to $\OPT$ and $d$, 
    then it holds that $\sum_{e \in \ALG} d_e \leq (7/2)\cdot \sum_{e\in \OPT} d_e$.
  \end{corollary}
	\begin{proof}
	Since there are no inessential edges, $S_u=\emptyset$. 
	We set $\Delta(S) := \sum_{e\in S} d_e$ for all $S \in \eqS$.
	Let $f \in \OPT$ be an edge that closes a cycle in $\ALG$ that contains $S$. 
	Then, $(\ALG \setminus \{e\}) \cup \{f\}$ is feasible for any single edge $e \in S$ because it is still a tree.
	By Lemma~\ref{lemsum:property:comp-class-all-or-nothing}, this implies that $(\ALG \setminus S) \cup \{f\}$ is also feasible. Thus, we consider the swap that adds $f$ and deletes $S$. It was not improving with respect to $d$, because $\ALG$ is \edgesetswap-optimal with respect to edges from $\OPT$ and $d$. Thus, $\Delta(S)=\sum_{e\in S} d_e \le d_f$, and we can apply Theorem~\ref{thm:four-approx-for-trees-hall} to obtain the result.	
	\end{proof}

\subsection{An approximation guarantee for trees and \texorpdfstring{$\phi$}{width objective}}
\label{sec-treecase-phi} 

We now consider the case where a connected tree $\ALG$ is output by the
algorithm when considering the \edgesetswap{}s, but now with respect to
the potential $\phi$ (instead of just the total length as in the
previous section). These swaps may increase the number of components,
which may have large widths, and hence \edgesetswap{}s that are improving
purely from the lengths may not be improving any more. This requires a
more nuanced analysis, though relying on ideas we have developed in the
previous section. 

Here is the high-level summary of this section.
Consider some equivalence class $S \in \eqS\backslash\{S_u\}$ of safe edges: these lie on a path (by Lemma~\ref{lemsum:property:compatible-path}), hence look like this:\vspace*{-0.4\baselineskip}
\begin{center}
\begin{tikzpicture}[yscale=0.95]
\useasboundingbox (-7.5,-0.4) rectangle (7.5,0.3);
\node {$w_{0} \xrightarrow{e_{1}} v_{1} \to \dots \to w_{1} \xrightarrow{e_{2}} v_{2} \to \dots \to w_{i-1}  \xrightarrow{e_{i}} v_{i} \to \dots w_{\lenl(S)-1} \xrightarrow{e_{\lenl(S)}} v_{\lenl(S)}$,};
\draw ($(-5.05,-0.4)+(-1.1,0)$) -- ($(-2.45,-0.4)+(-3.3,0)$) -- ($(-2.45,0.3)+(-3.3,0)$) -- ($(-5.05,0.3)+(-1.1,0)$);
\draw (-5.1,-0.4) rectangle (-2.5,0.3);
\draw ($(-2.45,-0.4)+(4.9,0)$) -- ($(-5.05,-0.4)+(6.6,0)$) -- ($(-5.05,0.3)+(6.6,0)$) -- ($(-2.45,0.3)+(4.9,0)$);
\draw ($(-5.05,-0.4)+(8.1,0)$) -- ($(-2.45,-0.4)+(6.75,0)$) -- ($(-2.45,0.3)+(6.75,0)$) -- ($(-5.05,0.3)+(8.1,0)$);
\begin{scope}[xshift=-3.45cm]
\draw ($(-2.45,-0.4)+(4.9,0)$) -- ($(-5.05,-0.4)+(6.6,0)$) -- ($(-5.05,0.3)+(6.6,0)$) -- ($(-2.45,0.3)+(4.9,0)$);
\draw ($(-5.05,-0.4)+(8.2,0)$) -- ($(-2.45,-0.4)+(6.85,0)$) -- ($(-2.45,0.3)+(6.85,0)$) -- ($(-5.05,0.3)+(8.2,0)$);
\end{scope}
\draw ($(-2.45,-0.4)+(8.6,0)$) -- ($(-5.05,-0.4)+(10.3,0)$) -- ($(-5.05,0.3)+(10.3,0)$) -- ($(-2.45,0.3)+(8.6,0)$);
\end{tikzpicture}
\end{center}
\vspace*{-0.5\baselineskip}
where there are $\ell(S)$ edges and hence $\ell(S)+1$ components formed
by deleting $S$.  
We let $\innS$ be set of the ``inner'' components (the ones containing $v_1, \ldots,
v_{\ell(S) - 1}$), and $\innSp$ be the inner components except the
two with the highest widths. Just taking the definition of $\phi$, and
adding and subtracting the widths of these ``not-the-two-largest'' inner
components, we get
\begin{align*}
 \phi(\ALG) = w(\ALG) + \sum\limits_{e\in S_u} d_e
             + \underbrace{\sum_{S\in\eqS\backslash\{S_u\}} \big( \sum_{i=1}^{\lenl(S)} d_{e_i}
               - \sum_{K \in \innSp} w(K)\big)}_{%
               \leq 10.5 \cdot d(\REFF)\ \text{by Corollary~\ref{cor:bound-weird-term}}
             } 
             + \underbrace{\sum_{S\in\eqS\backslash\{S_u\}} \sum_{K \in \innSp} w(K)}_{%
               \leq w(\REFF)\ \text{by Lemma~\ref{lem:bound-width-term}}
             }.
\end{align*}
As indicated above, the argument has two parts. For the first summation,
look at the cycle created by adding edge $f \in \REFF$ to our solution
$\ALG$. Suppose class $S$ is contained in this cycle. We prove that
\edgesetswap optimality implies that $\sum_{i=1}^{\lenl(S)} d_{e_i} -
\sum_{K \in \innSp} w(K)$ is at most $3 d_f$. (Think of this bound as
being a weak version of the facts in the previous section, which did not
have a factor of $3$ but did not consider weights in the analysis.)
Using this bound in Theorem~\ref{thm:four-approx-for-trees-hall} from
the previous section gives us a matching that bounds the first summation
by $3 \cdot (7/2) \cdot d(\REFF)$. (A couple of words about the proofs: the bound above
follows from showing that three different swaps must be non-improving,
hence the factor of $3$. Basically, we break the above path into three
at the positions of the two components of highest width, since for
technical reasons we do not want to disconnect these high-width
components. Details are in \S\ref{sec:tree:phi:middleterm}.)

For the second summation, we want to sum up the widths of all the
``all-but-two-widest'' inner components, over all these equivalence
classes, and argue this is at most $w(\REFF)$. This is where our
notions of safe and compatible edges comes into play. The crucial
observations are that (a) given the inner components corresponding to
some class $S$, the edges of some class $S'$ either avoid all these
inner components, or lie completely within some inner component; (b)~the
notion of compatibility ensures that these inner components correspond
to distinct components of $\REFF$, 
so we can get more width to charge to; and (c)~since we don't charge to
the two largest widths, we don't double-charge these widths. The details
are in \S\ref{sec:tree:phi:lastterm}.

\section{In Which the Local Optimum may be a Forest}
\label{sec:forest}

\paragraph{The main theorem.}
In the general case, both $\ALG$ and $\REFF$ may have multiple connected components. 
We assume that the distance function $d$ is a metric. 
The first thing that comes to mind is to apply the previous analysis to the components individually.
Morally, the main obstacle in doing so is in the proof of Theorem~\ref{thm:four-approx-for-trees-hall}: 
There, we assume implicitly that no edge from $\REFF$ goes between connected components of $\ALG$.\footnote{More precisely, we need the slightly weaker condition that for each node $t \in V[\ALG]$, there is an $\REFF$-edge incident to $t$ that does not leave the connected component of $\ALG$ containing $t$.}
This is vacuously true if $\ALG$ is a single tree, but may be false if $\ALG$ is disconnected.
In the following, our underlying idea is to replace $\REFF$-edges that cross between the components of $\ALG$ by edges that lie within the components of $\ALG$, thereby re-establishing the preconditions of Theorem~\ref{thm:four-approx-for-trees-hall}.
We do this in a way that $\REFF$ stays feasible, and moreover, its cost increases by at most a constant factor.
This allows us to prove that the local search has a constant locality gap.

\paragraph{Reducing to local tree optima.}

Suppose $\REFF$ has no inessential edges to
start. Then we convert $\REFF$ into a collection of cycles (shortcutting
non-terminals), losing a factor of 2 in the cost. 
Now observe that each ``offending'' $\REFF$-edge (i.e., one that goes between different components of $\ALG$)
must be part of a path $P$ in $\REFF$ that connects some $s, \bar s$,
and hence starts and ends in the same
component of $\ALG$. 
This path $P$ may connect several terminal pairs, and for each such pair $s,\bar{s}$, there is a component of $\ALG$ that contains $s$ and $\bar{s}$.
Thus, $P$ could be replaced by direct connections between $s,
  \bar s$ within the components of~$\ALG$. This would get rid of these
  ``offending'' edges, since the new connections would stay within
  components of $\ALG$. 
The worry is, however, that this replacement is too expensive. 
We show how to use connecting tree moves to bound the cost of the
replacement. 

Consider one cycle from $\REFF$, regarded as a circuit $C$ in the graph $G_\ALG$ where the connected components $A_1,\dots,A_p$ of $\ALG$ are shrunk to single nodes, \ie $C$ consists of offending edges. The graph $G_\ALG$ might contain parallel edges and $C$ might have
repeated vertices. So $C$ is a circuit, meaning that it is a potentially non-simple cycle, or, in other words, a Eulerian multigraph. The left and middle of Figure~\ref{fig:charging-example-compact} are an example.

Index the  $A_j$'s  such that $w(A_1) \leq \dots \leq w(A_p)$ and say that node $j$ in $G_\ALG$ corresponds to $A_j$.
Suppose $C$ visits the nodes $v_1,\dots,v_{|C|}, v_1$ (where several of these nodes may correspond to the same component $A_j$) and that component $A_j$ is visited $n_j$ times by $C$.
In the worst case, we may need to insert $n_j$ different $s,\bar s$ connections into component $A_j$ of $\ALG$, for all $j$. 
The key observation is that the total cost of our direct connections is at most $\sum_{i=1}^{|C|} n_i w(A_i)$. 
We show how to pay for this using the length of $C$.\footnote{We also need to take care of the additional width of the modified solution, but this is the easier part.}

To do so, we use optimality with respect to all moves, in particular connecting moves. The idea is simple: We cut $C$ into a set of trees that each define a valid connecting move. For each tree, the connecting move optimality bounds the widths of some components of $\ALG$ by the length of the tree. E.g., $w(A_1) + w(A_4)$ is at most the length of the tree connecting $A_1, A_4, A_5$ in Figure~\ref{fig:charging-example-compact}. Observe that we did not list $w(A_5)$: Optimality against a connecting move with tree $T$ relates the length of $T$ to the width of all the components that $T$ connects, \emph{except} for the component with maximum width. We say a tree \emph{pays} for $A_j$ if it hits $A_j$, and also hits another $A_j$ of higher width. So we need three properties: (a) the trees should collectively pack into the edges of the Eulerian multigraph $C$, (b) each tree hits each component $A_j$ at most once, and (c) the number of trees that pay for $A_j$ is at least $n_j$.

Assume that we found such a tree packing. For circuit $C$, if $A_{j^\star}$ is the component with greatest width hit by $C$, then using connecting move optimality for all the trees shows that 
\[ 
\srsum{j: A_j \text{ hit by } C, j \neq j^\star} n_j\, w(A_j) \leq d(C). 
\]
In fact, even if we have $c$-approximate connection-move optimality, the right-hand side just gets multiplied by $c$. But what about $n_{j^\star} w(A_{j^\star})$? We can cut $C$ into sub-circuits, such that each subcircuit $C'$ hits $A_{j^\star}$ exactly once. To get this one extra copy of $w(A_{j^\star})$, we use \pathsetswap optimality which tells us that the missing connection cannot be more expensive than the length of $C$. Thus, collecting all our bounds (see Lemma~\ref{lem:main-forest:reduction-to-tree}), adding all the extra connections to $\REFF$ increases the cost to at most $2(1+c) d(\REFF)$: the factor $2$ to make $\REFF$ Eulerian, $(1+c)$ to add the direct connections, using $c$-approximate optimality with respect to connecting moves and optimality with respect to \pathsetswap{}s. 
\S\ref{sec:wkmst} discusses that $c \le 2$. 

Now each component $A_j$ of $\ALG$ can be treated separately, \ie we can use Corollary~\ref{treecaseresult} on each $A_j$ and the portion of $\REFF$ that falls into $A_j$. By combining the conclusions for all connected components, we get that 
\[
d(\ALG') \stackrel{\text{Cor.}~\ref{treecaseresult}}{\le} 10.5 d(\REFF') +
w(\REFF') \le  11.5 d(\REFF') \stackrel{\text{Lem.~}\ref{lem:main-forest:reduction-to-tree}}{\le} 23 (1+c) \cdot d(\REFF) \le 69 \cdot d(\REFF) 
\]
for any feasible solution $\REFF$. This proves Theorem~\ref{thm:main}.

\paragraph{Obtaining a decomposition into connecting moves.}
It remains to show how to take $C$ and construct the set of trees. 
If $C$ had no repeated vertices (it is a simple cycle) then taking a spanning subtree would suffice. 
And even if $C$ has some repeated vertices, decomposing it into suitable trees can be easy: E.g., if $C$ is the ``flower'' graph on $n$ vertices, with vertex $1$ having two edges to each vertex $2,\dots,n$. Even though $1$ appears multiple times, we find a good decomposition (see Figure~\ref{flowergraph}). Observe, however, that breaking $C$ into simple cycles and than doing something on each simple cycle would not work, since it would only pay $1$ multiple times and none of the others.

The flower graph has a property that is a generalization of a simple cycle: We say that $C$ is minimally guarded if (a) the largest vertex is visited only once (b) between two occurrences of the same (non-maximal) vertex, there is at least one larger number. 
The flower graph and the circuit at the top of Figure~\ref{fig:charging-algo-example-compact} have this property. We show that every minimally guarded circuit can be decomposed suitably by providing Algorithm~\ref{alg:charging}.
%
%
It iteratively finds trees that pay for all (non maximal) $j$ with $j \le z$ for increasing $z$. Figure~\ref{fig:charging-algo-example-compact} shows how the set of trees $\moveSet_5$ is converted into $\moveSet_6$ in order to pay for all occurrences of $6$. Intuitively, we look where $6$ falls into the trees in $\moveSet_5$. Up to one occurrence can be included in a tree. If there are more occurrences, the tree has to be split into multiple trees appropriately. \S\ref{sec:partitioningalgo} explains the details.

Finally, Lemma~\ref{lem:guardedcycles} shows how to go from minimally guarded circuits to arbitrary $C$ 
in a recursive fashion.
\afterpage{
\begin{figure}[htp!]
  \centering
  \begin{tikzpicture}[
    scale=0.95,thick, 
    mynode/.style={draw, circle,inner sep=0cm, minimum size=0.15cm},
    mybnode/.style={draw, circle,inner sep=0cm, minimum size=0.6cm},
    markededge/.style={draw, very thick, dashed}
  ]
\begin{scope}[scale=0.6]
\draw[rounded corners=5mm] (0,0) rectangle (3,3); 

\draw[rounded corners=4mm] (3,4) rectangle (5.5,6.5); 
\draw[rounded corners=3mm] (3.5,1.5) rectangle (5.5,3.5); 

\draw[rounded corners=2mm] (6,0.5) rectangle (7.5,2); 
\node [mynode] (n1) at (2,2) {};
\node [mynode] (n2) at (4,5) {};
\node [mynode] (n3) at (4,2.5) {};
\node [mynode] (n4) at (5,2.5) {};
\node [mynode] (n5) at (5,5) {};
\node [mynode] (n6) at (6.75,1.5) {};
\node [mynode] (n7) at (6.75,1) {};
\node [mynode] (n8) at (2,1) {};
\draw [markededge] (n1) -- (n3) -- (n2);
\draw [markededge] (n2) -- (n5);
\draw [markededge] (n5) -- (n4) -- (n6);
\draw [markededge] (n6) -- (n7);
\draw [markededge] (n7) -- (n8);
\draw [markededge] (n8) -- (n1);
\node at (6.75,2.5) {$A_1$};
\node at (4.5,1.9) {$A_4$};
\node at (2.5,5) {$A_5$};
\node at (1.5,3.5) {$A_7$};
\end{scope}

\begin{scope}[xshift=6cm,scale=0.6]
\node [mybnode] (m1) at (6.75,1.25) {1};
\node [mybnode] (m2) at (4.25, 2.5)  {4};
\node [mybnode] (m3) at (4.25, 5.75) {5};
\node [mybnode] (m4) at (1.5, 1.25) {7};
\draw [markededge] (m4) to (m2);
\draw [markededge,bend left=15] (m3) to (m2);
\draw [markededge, bend left=15] (m2) to (m3);
\draw [markededge] (m2) to (m1) to (m4);
\node at (2,4) {$C$};
\end{scope}

\begin{scope}[xshift=12cm, scale=0.6]
\draw[rounded corners=5mm] (0,0) rectangle (3,3); 

\draw[rounded corners=4mm] (3,4) rectangle (5.5,6.5); 
\draw[rounded corners=3mm] (3.5,1.5) rectangle (5.5,3.5); 

\draw[rounded corners=2mm] (6,0.5) rectangle (7.5,2); 
\node [mynode] (n1) at (2,2) {};
\node [mynode] (n2) at (4,5) {};
\node [mynode] (n3) at (4,2.5) {};
\node [mynode] (n4) at (5,2.5) {};
\node [mynode] (n5) at (5,5) {};
\node [mynode] (n6) at (6.75,1.5) {};
\node [mynode] (n7) at (6.75,1) {};
\node [mynode] (n8) at (2,1) {};
\draw [markededge] (n1) -- (n3) -- (n2);
\draw [markededge] (n5) -- (n4) -- (n6);
\draw [markededge] (n7) -- (n8);
\node at (6.75,2.5) {$A_1$};
\node at (4.5,1.9) {$A_4$};
\node at (2.5,5) {$A_5$};
\node at (1.5,3.5) {$A_7$};
\end{scope}

\end{tikzpicture}
\caption{\label{fig:charging-example-compact}
 \footnotesize The charging argument with four components $A_1, A_4, A_5$ and $A_7$ of $\ALG$. 
  The area of each component corresponds to its width.
  \textit{On the left.} A cycle in $\REFF$. 
  \textit{In the middle.} The corresponding circuit (non-simple cycle) in $G_\ALG$.
  \textit{On the right.} A suitable decomposition into connecting moves.
}
\end{figure}
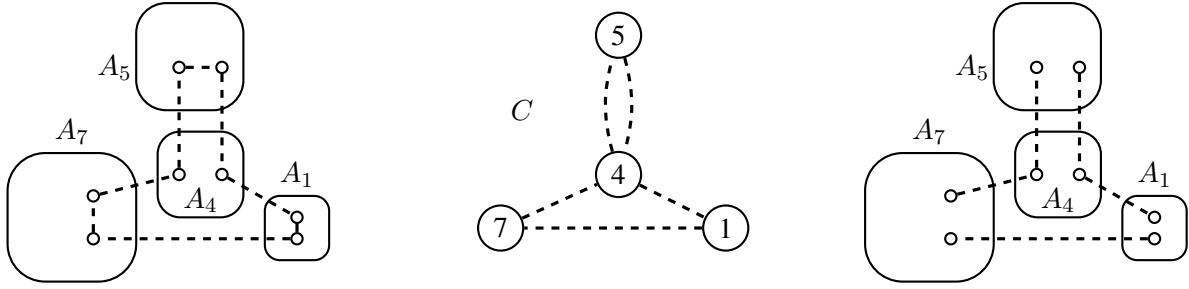
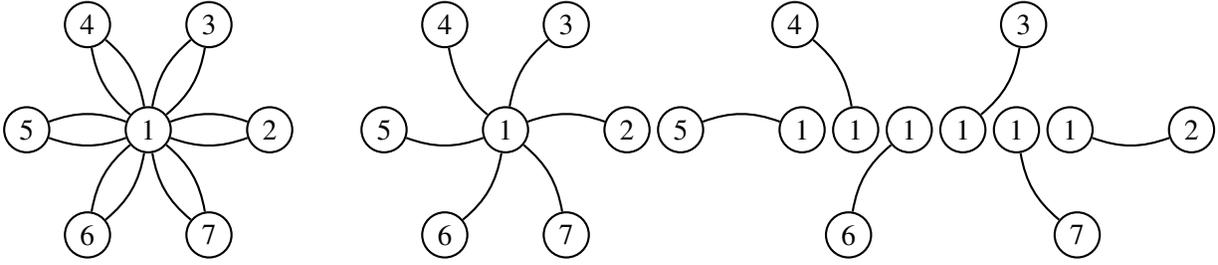
\begin{figure}[htp!]
\begin{center}
\begin{tikzpicture}[
    scale=0.95,thick, 
    mynode/.style={draw, circle,inner sep=0cm, minimum size=0.15cm},
    mybnode/.style={draw, circle,inner sep=0cm, minimum size=0.6cm},
    markededge/.style={draw, very thick, dashed}
  ]
\node [mybnode] (n1) at (0,0) {1};
\foreach \x/\i in {0/2,60/3,120/4,180/5,240/6,300/7} 
{
 \node [mybnode] (n\i) at (\x:1.7) {\i};
 \draw [bend left=20] (n1) to (n\i);
 \draw [bend right=20] (n1) to (n\i);
}
\begin{scope}[xshift=5cm]
 \node [mybnode] (m1) at (0,0) {1};
 \foreach \x/\i in {0/2,60/3,120/4,180/5,240/6,300/7} 
 {
  \node [mybnode] (m\i) at (\x:1.7) {\i};
  \draw [bend left=20] (m1) to (m\i);
 }
\end{scope}

\begin{scope}[xshift=-0.6cm]
\begin{scope}[xshift=13.5cm]
 \node [mybnode] (m1) at (0,0) {1};
 \foreach \x/\i in {0/2} 
 {
  \node [mybnode] (m\i) at (\x:1.7) {\i};
  \draw [bend right=20] (m1) to (m\i);
 }
\end{scope}

\begin{scope}[xshift=12cm,yshift=0cm]
 \node [mybnode] (m1) at (0,0) {1};
 \foreach \x/\i in {60/3} 
 {
  \node [mybnode] (m\i) at (\x:1.7) {\i};
  \draw [bend right=20] (m1) to (m\i);
 }
\end{scope}

\begin{scope}[xshift=10.5cm,yshift=0cm]
 \node [mybnode] (m1) at (0,0) {1};
 \foreach \x/\i in {120/4} 
 {
  \node [mybnode] (m\i) at (\x:1.7) {\i};
  \draw [bend right=20] (m1) to (m\i);
 }
\end{scope}

\begin{scope}[xshift=9.75cm,yshift=0cm]
 \node [mybnode] (m1) at (0,0) {1};
 \foreach \x/\i in {180/5} 
 {
  \node [mybnode] (m\i) at (\x:1.7) {\i};
  \draw [bend right=20] (m1) to (m\i);
 }
\end{scope}

\begin{scope}[xshift=11.25cm,yshift=0cm]
 \node [mybnode] (m1) at (0,0) {1};
 \foreach \x/\i in {240/6} 
 {
  \node [mybnode] (m\i) at (\x:1.7) {\i};
  \draw [bend right=20] (m1) to (m\i);
 }
\end{scope}

\begin{scope}[xshift=12.75cm,yshift=0cm]
 \node [mybnode] (m1) at (0,0) {1};
 \foreach \x/\i in {300/7} 
 {
  \node [mybnode] (m\i) at (\x:1.7) {\i};
  \draw [bend right=20] (m1) to (m\i);
 }
\end{scope}
\end{scope}
\end{tikzpicture}
\caption{A flower graph. Even though the graph is a non-simple cycle, we can easily decompose it into trees that pay for each $j\neq 6$ at least $n_j$ times ($1$ is payed for $7=n_1+1$ times).\label{flowergraph}}
\end{center}
\end{figure}
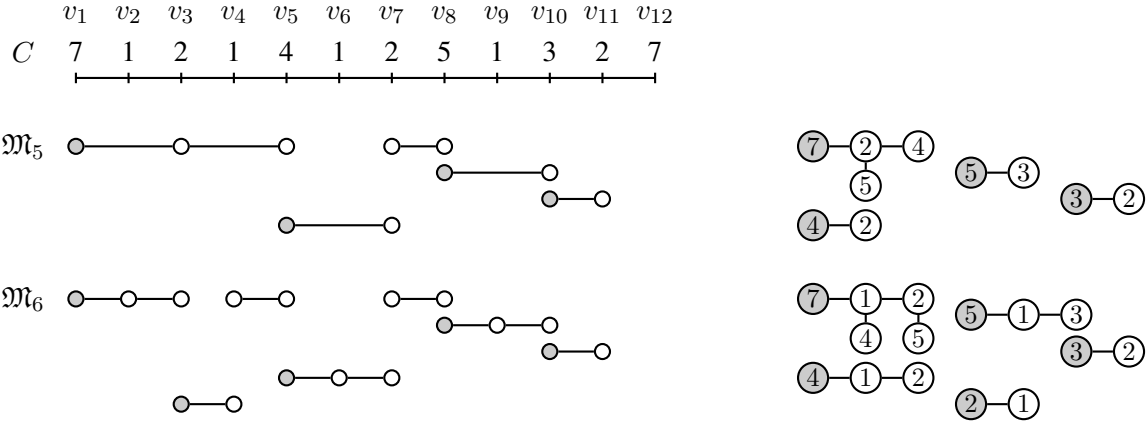
\begin{figure}[htp!]
\begin{tikzpicture}[
  thick, 
  scale=0.7, 
  treenode/.style={draw, circle,inner sep=0cm, minimum size=0.4cm},streenode/.style={draw, circle,inner sep=0cm, minimum size=0.2cm}
]
\foreach \x/\xi in {1/7,2/1,3/2,4/1,5/4,6/1,7/2,8/5,9/1,10/3,11/2,12/7}{
 \draw (\x,-0.1) -- (\x,0.1);
 \node at (\x,0.5) {\xi};
 \node at (\x,1.2) {$v_{\x}$};
}
\draw (0,0.5) node {$C$};
\draw (1,0) -- (12,0);

\begin{scope}[yshift=-0.3cm]
\node at (0,-1) {$\moveSet_5$};
\node [streenode,fill=black!20] (nt41) at (1,-1) {}; \node [streenode] (nt41b) at (3,-1) {}; 
\node [streenode] (nt42) at (5,-1) {}; \node [streenode] (nt42b) at (7,-1) {};  \node [streenode] (nt43) at (8,-1) {};
\node [streenode,fill=black!20] (nt44) at (8,-1.5) {}; \node [streenode] (nt45) at (10,-1.5) {};
\node [streenode,fill=black!20] (nt46) at (10,-2) {}; \node [streenode] (nt47) at (11,-2) {};
\node [streenode,fill=black!20] (nt48) at (5,-2.5) {}; \node [streenode] (nt49) at (7,-2.5) {}; 
\draw [-] (nt41) -- (nt41b); \draw [-] (nt41b) -- (nt42); \draw [-] (nt42b) -- (nt43); \draw [-] (nt44) -- (nt45); \draw [-] (nt46)--(nt47);
\draw [-] (nt48) -- (nt49); 
\node [treenode,fill=black!20] (n41) at (15,-1) {{\small $7$}};
\node [treenode] (n42) at (16,-1) {{\small $2$}};
\node [treenode] (n43) at (17,-1) {{\small $4$}};
\node [treenode] (n43b) at (16,-1.75) {{\small $5$}};
\node [treenode,fill=black!20] (n44) at (18,-1.5) {{\small $5$}};
\node [treenode] (n45) at (19,-1.5) {{\small $3$}};
\node [treenode,fill=black!20] (n46) at (20,-2) {{\small $3$}};
\node [treenode] (n47) at (21,-2) {{\small $2$}};
\node [treenode,fill=black!20] (n48) at (15,-2.5) {{\small $4$}};
\node [treenode] (n49) at (16,-2.5) {{\small $2$}};
\draw [-] (n41)--(n42); \draw [-] (n42) --(n43); \draw [-] (n42) --(n43b); 
\draw [-] (n44) -- (n45); \draw [-](n46)--(n47); \draw [-] (n48) -- (n49);
\end{scope}

\begin{scope}[yshift=-3.2cm]
\node at (0,-1) {$\moveSet_6$};
\node [streenode,fill=black!20] (nt51) at (1,-1) {}; \node [streenode] (nt51a) at (2,-1) {};  \node [streenode] (nt51b) at (3,-1) {}; 
\node [streenode] (nt51c) at (4,-1) {}; \node [streenode] (nt52) at (5,-1) {};  \node [streenode] (nt52b) at (7,-1) {};  \node [streenode] (nt53) at (8,-1) {};
\node [streenode,fill=black!20] (nt54) at (8,-1.5) {}; \node [streenode] (nt54b) at (9,-1.5) {}; \node [streenode] (nt55) at (10,-1.5) {};
\node [streenode,fill=black!20] (nt56) at (10,-2) {}; \node [streenode] (nt57) at (11,-2) {};
\node [streenode,fill=black!20] (nt58) at (5,-2.5) {}; \node [streenode] (nt58b) at (6,-2.5) {};  \node [streenode] (nt59) at (7,-2.5) {}; 
\node [streenode,fill=black!20] (nt59a) at (3,-3) {};  \node [streenode] (nt59b) at (4,-3) {};
\draw [-] (nt51) -- (nt51a); \draw [-] (nt51a) -- (nt51b); 
\draw [-] (nt51c) -- (nt52); \draw [-] (nt52b) -- (nt53); \draw [-] (nt54) -- (nt54b); 
\draw [-] (nt54b) -- (nt55); \draw [-] (nt56)--(nt57);
\draw [-] (nt58) -- (nt58b);  \draw [-] (nt58b) -- (nt59); 
\draw [-] (nt59a) -- (nt59b);
\node [treenode,fill=black!20] (n50) at (15,-1) {{\small $7$}};
\node [treenode] (n51) at (16,-1) {{\small $1$}};
\node [treenode] (n52) at (17,-1) {{\small $2$}};
\node [treenode] (n53) at (16,-1.75) {{\small $4$}};
\node [treenode] (n53b) at (17,-1.75) {{\small $5$}};
\node [treenode,fill=black!20] (n54) at (18,-1.3) {{\small $5$}};
\node [treenode] (n55) at (19,-1.3) {{\small $1$}};
\node [treenode] (n55b) at (20,-1.3) {{\small $3$}};
\node [treenode,fill=black!20] (n56) at (20,-2) {{\small $3$}};
\node [treenode] (n57) at (21,-2) {{\small $2$}};
\node [treenode,fill=black!20] (n58) at (15,-2.5) {{\small $4$}};
\node [treenode] (n58b) at (16,-2.5) {{\small $1$}};
\node [treenode] (n59) at (17,-2.5) {{\small $2$}};
\node [treenode,fill=black!20] (n59b) at (18,-3) {{\small $2$}};
\node [treenode] (n59c) at (19,-3) {{\small $1$}};
\draw [-] (n50)--(n51); \draw [-] (n51)--(n52); \draw [-] (n51) --(n53); \draw [-] (n52) --(n53b); 
\draw [-] (n54) -- (n55); \draw [-] (n55) -- (n55b); \draw [-](n56)--(n57); \draw [-] (n58) -- (n58b); \draw [-] (n58b) -- (n59);
\draw [-] (n59b) -- (n59c);
\end{scope}

\end{tikzpicture}
\caption{\label{fig:charging-algo-example-compact}Two iterations of an example run of the Algorithm~\ref{alg:charging}.}
\end{figure}
}


\section{Proofs: Bounds on \texorpdfstring{$d$}{d} when the Local Optimum is a Tree}\label{appendix:addproofs}

Let $\ALG,\REFF\subseteq E$ be two feasible Steiner forests with respect to a set $\mathfrak{T}$ of terminal pairs.
We think of $\REFF$ as an optimum or near optimum Steiner forest and of $\ALG$ as a feasible solution computed by our algorithm. 
We assume that $V[\ALG]=V[\REFF]$ (this will be true when we use the results from this section later on).
Throughout this section, we assume that $\ALG$ is a tree (thus, it is a spanning tree on $V[\ALG]=V[\REFF]$). 
The following definition is crucial for our analysis.
\begin{definition}
Let $e=\{s,v_1\}, f=\{v_l, t\} \in \ALG$ be two edges.
Consider the unique (undirected) path $P = s \xrightarrow{e} v_1 \to \dots \to v_l \xrightarrow{f} t$ in $\ALG$ that connects $e$ and $f$. Let $T_{e,f}$ be the connected component in $\ALG\setminus\{e,f\}$ that contains $P \setminus \{e, f\}$.
We say that $e$ and $f$ are \emph{compatible with respect to $\REFF$} if $e=f$ or if there are no \REFF-edges leaving~$T_{e,f}$.
In this case, we write $e \sim_{cp} f$.
\end{definition}
Observe that $\sim_{cp}$ is a reflexive and symmetric relation. 
We show that the relation is also transitive, \ie it is an equivalence relation.
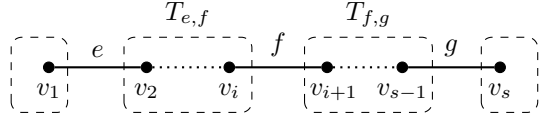
\begin{lemdef}\label{lem:compatibility-is-transitive}
Let $\ALG$ be feasible and let $e, f \in \ALG$ and $f, g \in \ALG$ be two pairs of compatible edges. 
Then, $e$ and $g$ are also compatible.
In particular, $\sim_{cp}$ is an equivalence relation.
We denote the set of all equivalence classes of $\sim_{cp}$ by $\eqS$.
\end{lemdef}
\begin{proof}
Consider the unique path $P$ that connects $e$ and $g$ in $\ALG$.
This path has the form $v_1 \xrightarrow{e} v_2 \to \dots \to v_{s-1} \xrightarrow{g} v_s$.
We distinguish three cases (see Figure~\ref{fig:compatibility-is-transitive}). 
  \begin{figure}
    \begin{minipage}{0.45\textwidth}
    \begin{subfigure}{\textwidth}
        \centering
        \begin{tikzpicture}[
        graphnode/.style={label distance=0.0mm, fill, draw, circle, inner sep=0pt, minimum size=4pt},
        cc/.style={draw, rounded corners, dashed},
        graphedge/.style={thick},
        pathedge/.style={graphedge, dotted},
        font=\small
      ]
        \draw (-3.9, 0) node[graphnode, label={below:$v_1$}] (w1) {};
        \draw (-2.6, 0)  node[graphnode, label={below:$v_2$}] (w2) {};
        \draw (-1.5, 0)  node[graphnode, label={below:$v_i$}] (v1) {};
        \draw (-0.2, 0) node[graphnode, label={below:$\ \ v_{i+1}$}] (v2) {};
        \draw ( 0.8, 0)  node[graphnode, label={below:$v_{s-1}$}] (vss) {};
        \draw ( 2.1, 0) node[graphnode, label={below:$v_s$}] (vs) {};

        \path[cc] (w1) ++(-0.50, -0.6) rectangle ++(0.75, 1.0);
        \path[cc] (w2) ++(-0.3, -0.6) rectangle node[above=0.5cm]{$T_{e,f}$} ++( 1.7, 1.0);        
        \path[cc] (v2) ++(-0.3, -0.6) rectangle node[above=0.5cm]{$T_{f,g}$} ++( 1.7, 1.0);        
        \path[cc] (vs) ++(-0.25, -0.6) rectangle ++(0.75, 1.0);        
        
        \draw[graphedge] (w1) -- node[above]{$e$} (w2);
        \draw[pathedge]  (w2) -- (v1);
        \draw[graphedge] (v1) -- node[above]{$f$} (v2);
        \draw[pathedge]  (v2) -- (vss);
        \draw[graphedge] (vss) -- node[above]{$g$} (vs);
         \end{tikzpicture}
      \caption{First case: Edge $f$ lies on the path from $e$ to $g$.
        If an $\REFF$-edge is cut by $T_{f,g}$, the edges $f$ and $g$ cannot
        be compatible. Likewise, $e$ and $f$ cannot be compatible if an 
        $\REFF$-edge is cut by $T_{e,f}$. Thus, no $\REFF$-edge is cut 
        by $T_{e,f} \cup T_{f,g}$ and $e$ and $g$ are compatible.
      }
    \end{subfigure}
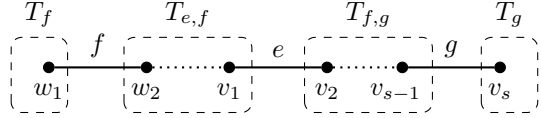
\vspace{\baselineskip}
    
    \begin{subfigure}{\textwidth}
      \centering
        \begin{tikzpicture}[
        graphnode/.style={label distance=0.0mm, fill, draw, circle, inner sep=0pt, minimum size=4pt},
        cc/.style={draw, rounded corners, dashed},
        graphedge/.style={thick},
        pathedge/.style={graphedge, dotted},
        font=\small
      ]
        
        \draw ( -3.9, 0) node[graphnode, label={below:$w_1$}] (w1) {};
        \draw ( -2.6, 0) node[graphnode, label={below:$w_2$}] (w2) {};
        \draw ( -1.5, 0) node[graphnode, label={below:$v_1$}] (v1) {};
        \draw ( -0.2, 0) node[graphnode, label={below:$v_2$}] (v2) {};
        \draw (  0.8, 0) node[graphnode, label={below:$v_{s-1}\ \ $}] (vss) {};
        \draw (  2.1, 0) node[graphnode, label={below:$v_s$}] (vs) {};
        
        \path[cc] (w1) ++(-0.50, -0.6) rectangle node[above=0.5cm]{$T_{f}$} ++(0.75, 1.0);
        \path[cc] (w2) ++(-0.3, -0.6) rectangle node[above=0.5cm]{$T_{e,f}$} ++( 1.7, 1.0);        
        \path[cc] (v2) ++(-0.3, -0.6) rectangle node[above=0.5cm]{$T_{f,g}$} ++( 1.7, 1.0);        
        \path[cc] (vs) ++(-0.25, -0.6) rectangle node[above=0.5cm]{$T_{g}$} ++(0.75, 1.0);             
        
        \draw[graphedge] (w1) -- node[above]{$f$} (w2);
        \draw[pathedge]  (w2) -- (v1);
        \draw[graphedge] (v1) -- node[above]{$e$} (v2);
        \draw[pathedge]  (v2) -- (vss);
        \draw[graphedge] (vss) -- node[above]{$g$} (vs);
         \end{tikzpicture}
      \caption{Second case: The path connecting $f$ to the path $P$ from $e$ to $g$ touches an outer node of $P$.
      There can be no $\REFF$-edge between $T_{f,g}$ and $T_f \cup T_{e,f}$ because $e$ and $f$ are compatible. 
			Likewise, no $\REFF$-edge can cross from $T_{f,g}$ to $T_g$ because $f$ and $g$ are compatible. 
			This shows that no $\REFF$-edge is cut by $T_{f,g}$ and thus, $e$ and $g$ are compatible.
      }
    \end{subfigure}    
    \end{minipage}\qquad
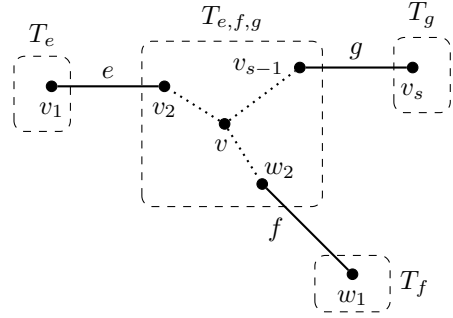
    \begin{subfigure}{0.45\textwidth}
      \centering
      \begin{tikzpicture}[
      graphnode/.style={label distance=0.0mm, fill, draw, circle, inner sep=0pt, minimum size=4pt},
      cc/.style={draw, rounded corners, dashed},
      graphedge/.style={thick},
      pathedge/.style={graphedge, dotted},
      font=\small
      ]
  
      \draw (-1.7,  0.0 )  node[graphnode, label={below:$v$\ \ }]              (v)   {};
      \draw (-2.5,  0.5 )  node[graphnode, label={below:$v_2$}]            (vRE) {};
      \draw (-0.7,  0.75)  node[graphnode, label={180:$v_{s-1}$}]          (vRF) {};
      \draw (-4.0,  0.5 )  node[graphnode, label={below:$v_1$}]            (vER) {};
      \draw ( 0.8,  0.75)  node[graphnode, label={below:$v_s$}]            (vFR) {};
      \draw (-1.2, -0.8 )  node[graphnode, label={[rectangle, inner sep=0mm]75:$w_2$}]               (vRG) {};
      \draw ( 0.0, -2.0 )  node[graphnode, label={below:$w_1$}]            (vGR) {};
            
      \path[cc] (vER) ++(-0.5, -0.6) rectangle node[above=0.5cm]{$T_e$} ++(0.75, 1.0);  
      \path[cc] (vRE) ++(-0.3, -1.6) rectangle node[above=1.1cm]{$T_{e,f,g}$} ++(2.4, 2.2);    
      \path[cc] (vGR) ++(-0.5, -0.5) rectangle node[right=0.5cm]{$T_f$} ++(1.0, 0.75);
      \path[cc] (vFR) ++(-0.25, -0.6) rectangle node[above=0.5cm]{$T_g$} ++(0.75, 1.0);
      
      \draw[pathedge]  (v)   -- (vRE);
      \draw[pathedge]  (v)   -- (vRF);
      \draw[pathedge]  (v)   -- (vRG);
      \draw[graphedge] (vRE) -- node[rectangle, inner sep=0mm, above=1mm]{$e$} (vER);
      \draw[graphedge] (vRF) -- node[rectangle, inner sep=0mm, above=1mm]{$g$} (vFR);
      \draw[graphedge] (vRG) -- node[rectangle, inner sep=0mm, left=3mm]{$f$}  (vGR);
    \end{tikzpicture}
    \caption{Third case: The path connecting $f$ to the path $P$ from $e$ to $g$ touches an inner node $v$ of $P$.
    Here, no $\REFF$-edge can cross from $T_{e,f,g}$ to $T_e$ nor from $T_f$ to $T_g$ since $e$ and $f$ are compatible. 
    The edges $f$ and $g$ being compatible, there also cannot be $\REFF$-edges from $T_{e,f,g}$ to $T_g$ nor from  $T_e$ to $T_f$.
    Thus, no $\REFF$-edge is cut by $T_{e,f,g} \cup T_f$ and thus, the edges $e$ and $g$ are compatible.    
    }
    \end{subfigure}
    \caption{\label{fig:compatibility-is-transitive}%
      The situation in Lemma~\ref{lem:compatibility-is-transitive}.}
    \end{figure}

    
For the first case, suppose that $f=\{v_i, v_{i+1}\}$ lies on $P$, for some $i \in \{2,\dots,s-2\}$.
Let $T_{e,f}$ and $T_{f,g}$ be the connected components of $\ALG\setminus\{e,f,g\}$ that contain $v_2,\dots,v_i$ and $v_{i+1},\dots,v_{s-1}$, respectively.
Assume by contradiction that $\delta_{\REFF}(T_{e,g}) \not= \emptyset$. 
Since $\delta_{\REFF}(T_{e,g}) \not= \emptyset$, it follows that $\delta_{\REFF}(T_{e,f}) \not=\emptyset$ or $\delta_{\REFF}(T_{f,g}) \not= \emptyset$.
This is a contradiction to the assumption that $e,f$ and $f,g$ are compatible, respectively.
    
For the second case, suppose that $f$ does not lie on $P$. 
Let $R$ be the unique, possibly empty, path that connects $f$ to $P$ and suppose that $R$ meets $P$ at one of the outer nodes (w.l.o.g. assume that $P$ and $R$ meet in $v_1$). 
Thus, we are in the situation where $R \cup P$ has the form $w_1 \xrightarrow{f} w_2 \to \dots \to w_{t} \to v_1 \xrightarrow{e} v_2 \dots \to v_{s-1} \xrightarrow{g} v_{s}$.
Consider the four connected components that $\ALG\setminus\{e,f,g\}$ decomposes into:
The component $T_{f}$ contains $w_1$, the component $T_{f,e}$ contains $w_2,\dots,w_t,v_1$, 
the component $T_{e,g}$ contains $v_2,\dots,v_{s-1}$ and $T_{g}$ contains $v_s$.
If $R$ is empty, the nodes $w_2$, $w_{t}$ and $v_1$ coincide and $T_{f,e}$ contains only $v_1$.
Both the cuts $\delta_{\REFF}(T_{e,g} : T_g)$ and $\delta_{\REFF}(T_{e,g}:T_f)$ must be empty because $f$ and $g$ are compatible. 
The cut $\delta_{\REFF}(T_{e,g} : T_{f,e})$ must be empty because $e$ and $f$ are compatible. 
It follows that $e$ and $g$ are compatible. 

For the last remaining case, suppose that $f$ does not lie on $P$ and that the unique path $R$ connecting $f$ to $R$ meets $P$ in an inner node $v_i$, $i \in\{2,\dots,s-1\}$. 
Here, $R$ has the form $w_1 \xrightarrow{f} w_2 \to \dots \to w_t \to v_i$ and $\ALG\setminus\{e,f,g\}$ decomposes into four components:
The component $T_e$ contains $v_1$, the component $T_{e,f,g}$ contains $v_2,\dots,v_{s-1}$ and $w_2,\dots,w_t$,
the component $T_{g}$ contains $v_s$ and the component $T_{f}$ contains $w_1$.
The cuts $\delta_{\REFF}(T_f : T_g)$ and $\delta_{\REFF}(T_{e,f,g} : T_e)$ must be empty because $e$ and $f$ are compatible.
The cuts $\delta_{\REFF}(T_f : T_e)$ and $\delta_{\REFF}(T_{e,f,g} : T_g)$ must be empty because $f$ and $g$ are compatible.
Thus, the edges $e$ and $g$ are compatible. This concludes the proof.
\end{proof}

  \begin{figure}
    \centering
      \begin{tikzpicture}[
      graphnode/.style={label distance=0.0mm, fill, draw, circle, inner sep=0pt, minimum size=4pt},
      cc/.style={draw, rounded corners, dashed},
      graphedge/.style={thick},
      pathedge/.style={graphedge, dotted},
      font=\small
      ]
  
      \draw (-1.7,  0.0 )  node[graphnode, label={below:$v_0$\ \ }]              (v)   {};
      \draw (-2.5,  0.5 )  node[graphnode]            (vRE) {};
      \draw (-0.7,  0.75)  node[graphnode]          (vRF) {};
      \draw (-4.0,  0.5 )  node[graphnode, label={below:$v_2$}]            (vER) {};
      \draw ( 0.8,  0.75)  node[graphnode, label={below:$v_3$}]            (vFR) {};
      \draw (-1.2, -0.8 )  node[graphnode]               (vRG) {};
      \draw ( 0.0, -2.0 )  node[graphnode, label={below:$v_1$}]            (vGR) {};
            
      \path[cc] (vER) ++(-0.5, -0.6) rectangle node[above=0.5cm]{$T_2$} ++(0.75, 1.0);  
      \path[cc] (vRE) ++(-0.3, -1.6) rectangle node[above=1.1cm]{$T_0$} ++(2.4, 2.2);    
      \path[cc] (vGR) ++(-0.5, -0.5) rectangle node[right=0.5cm]{$T_1$} ++(1.0, 0.75);
      \path[cc] (vFR) ++(-0.25, -0.6) rectangle node[above=0.5cm]{$T_3$} ++(0.75, 1.0);
      
      \draw[pathedge]  (v)   -- (vRE);
      \draw[pathedge]  (v)   -- (vRF);
      \draw[pathedge]  (v)   -- (vRG);
      \draw[graphedge] (vRE) -- node[rectangle, inner sep=0mm, above=1mm]{$e_2$} (vER);
      \draw[graphedge] (vRF) -- node[rectangle, inner sep=0mm, above=1mm]{$e_3$} (vFR);
      \draw[graphedge] (vRG) -- node[rectangle, inner sep=0mm, left=3mm]{$e_1$}  (vGR);
    \end{tikzpicture}
    \caption{\label{fig:lemma-compatibility-path}
      The situation from the proof of Lemma~\ref{lem:compatibility-path}. Since $e_1$, $e_2$, and $e_3$ 
      are pairwise compatible, no \REFF-edge can leave $T_1$.}
  \end{figure}
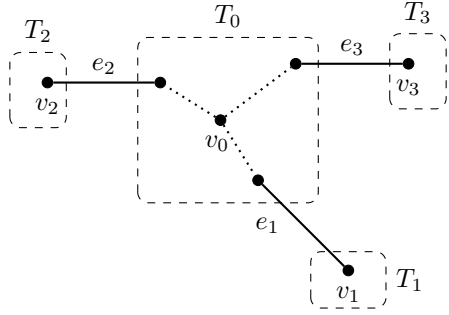
	
We classify the edges in $\ALG$ with respect to $\REFF$ and $\terms$. 
Recall that we assume $V[\ALG]=V[\REFF]$, and that $\ALG$ is a tree, thus, it is a spanning tree on the vertices that we are interested in.
\begin{definition}
An edge $e \in \ALG$ is \emph{essential} if $\ALG$ is feasible with respect to $\terms$, but $\ALG\setminus\{e\}$ is infeasible with respect to $\terms$.
\end{definition}
\begin{definition}
Let $e\in \ALG$ and let $T_1, T_2$ be the connected components of $\ALG\setminus\{e\}$.
We say that $e$ is \emph{safe} if $\delta_{\REFF}(T_1) = \delta_{\REFF}(T_2)\not= \emptyset$, i.e.,
if at least one $\REFF$-edge crosses between $T_1$ and $T_2$. 
\end{definition}
By definition, any essential edge is safe, however, the opposite is not true
in general: Safe edges can be essential or inessential.
We thus have classified the edges of $\ALG$ into three classes: safe essential 
edges, safe inessential edges and unsafe (and ergo, inessential) edges. 
 
  \begin{lemma}\label{lem:safe-all-or-none}
    Let $e,f \in \ALG$ be two compatible edges. Then $e$ is safe if and only if $f$ is safe.
  \end{lemma}
  \begin{proof}
    Suppose that $e$ is safe and let $P=v_1 \xrightarrow{e} v_2 \to \dots\to v_{s-1}\xrightarrow{f} v_s$
    be the unique path between $e$ and $f$ in $\ALG$. 
    Let $T_e$, $T_{e,f}$ and $T_f$ be the connected components of $\ALG$ that contain the vertex $v_1$, the vertices $v_2,\dots,v_{s-1}$ and the vertex $v_s$, respectively.
    Since $e$ and $f$ are compatible we know that $\delta_{\REFF}(T_f : T_{e,f})$ 
    is empty. 
    Also, since $e$ is safe, we have that $\delta_{\REFF}(T_e : T_f) = \delta_{\REFF}(T_f : T_e)$
    is non-empty. 
    It follows from $\delta_{\REFF}(T_f) = \delta_{\REFF}(T_f : T_e) \cup \delta_{\REFF}(T_f : T_{e,f})$
    that $f$ is safe. 
    The reverse implication follows since $e$ and $f$ are exchangeable.
  \end{proof}
  \begin{lemma}\label{lem:unsafe-implies-compatible}
    Let $e,f \in \ALG$ be two unsafe edges. Then $e$ and $f$ are compatible.
  \end{lemma}
  \begin{proof}
    Let $P=v_1 \xrightarrow{e} v_2 \to \dots \to v_{s-1} \xrightarrow{f} v_s$
    be the unique path between $e$ and $f$ in $\ALG$. 
    Let $T_e$, $T_{e,f}$ and $T_f$ be the connected components of $\ALG$ that 
    contain the vertex $v_1$, the vertices $v_2,\dots,v_{s-1}$ and the vertex $v_s$, respectively.
    We need to show that $\delta_{\REFF}(T_e : T_{e,f}) = \delta_{\REFF}(T_f : T_{e,f}) = \emptyset$.
    Observe that $\delta_{\REFF}(T_e : T_{e,f}) \subseteq \delta_{\REFF}(T_e)$
    and that $\delta_{\REFF}(T_f : T_{e,f}) \subseteq \delta_{\REFF}(T_f)$. 
    The assumption that $e$ and $f$ are both unsafe implies that 
    $\delta_{\REFF}(T_e) = \delta_{\REFF}(T_f) = \emptyset$ and
    we have thus shown the claim.
  \end{proof}

  We have the following summary lemma showing that $\sim_{cp}$ behaves well.
	\begin{lemma}\label{lem:summary:classification}\label{lem:compatibleessential}
    Let $e, f\in \ALG$ be compatible edges. Then 
		\begin{enumerate}
				\item $e$ is essential if and only if $f$ is essential and
				\item $e$ is safe if and only $f$ is safe.
		\end{enumerate}
    Furthermore, if two edges $e,f \in \ALG$ are unsafe, then they are compatible.
		Thus, the set 
    \begin{equation*}
    S_{u} := \cset{e \in \ALG}{\text{$e$ is unsafe}}
    \end{equation*}
    forms an equivalence class of $\sim_{cp}$. 
    If $S \in \eqS$ is an equivalence class of $\sim_{cp}$, then either all edges $e \in S$ are essential or none. Also, if $S \neq S_u$, then all edges in $S$ are safe.
	\end{lemma}
	\begin{proof}
		We show Statement 1.
		
		Statement 2 and the compatibility of unsafe edges are shown in Lemma~\ref{lem:safe-all-or-none} and Lemma~\ref{lem:unsafe-implies-compatible}. 
		Let $P=v_1 \xrightarrow{e} v_2 \to \dots \to v_{s-1} \xrightarrow{f} v_s$
    be the unique path between $e$ and $f$ in $\ALG$. 
    Let $T_e$, $T_{e,f}$ and $T_f$ be the connected components of $\ALG$ that contain the vertex $v_1$, the vertices $v_2,\dots,v_{s-1}$ and the vertex $v_s$, respectively.
		Suppose that $e$ is essential. 
    We show that there must be a terminal pair $u, \bar{u}$ with $u \in T_{e} \cup T_{e,f}$
    and $\bar{u} \in T_f$. It then follows that $\ALG\setminus\{f\}$ is infeasible.
    
    Since the edge $e$ is essential, there must be a terminal pair $u, \bar{u}$
    with $u \in T_e$ and $\bar{u} \in T_{e,f} \cup T_f$. 
    Since $e$ and $f$ are compatible, however, we know that 
    $\delta_{\REFF}(T_e : T_{e,f}) = \delta_{\REFF}(T_f : T_{e,f}) = \emptyset$
    which implies that $\delta_{\REFF}(T_{e,f}) = \emptyset$. 
    Thus, if $\bar{u} \in T_{e,f}$, the forest $\REFF$ would be infeasible. 
    It follows that $\bar{u} \in T_f$ and thus, that $f$ is essential.
  \end{proof}
	

  We can now show that the compatibility classes in $\mathfrak{S}\setminus S_u$ behave as
  if they were single edges.
  \begin{lemma}\label{lem:compatibility-path}
 Let $K:=\{e_1,\dots,e_l\} \subseteq \ALG$ with $l \geq 2$ be a set of 
 pairwise compatible edges. 
 Furthermore, suppose that $e_i\in K$ is safe for all $i=1,\dots,l$.
 Then there is a path $P \subseteq \ALG$ with $K \subseteq P$.
\end{lemma}
\begin{proof}
For $e, f \in K$ let $P_{e,f}$ be the unique path in $\ALG$ starting with $e$ and ending with $f$. Define $R := \bigcup_{e,f \in K} P_{e,f}$. Observe that $K \subseteq R \subseteq \ALG$ and that $R$ is a tree whose leaves are all incident to edges in $K$. By contradiction assume $R$ is not a path. Then there exists a vertex $v_0$ with $|\delta_R(v_0)| \geq 3$. In particular, there are three leaves $v_1, v_2, v_3$ of $R$ such that the unique $v_0$-$v_i$-paths in $R$ for $i \in \{1, 2, 3\}$ are pairwise edge-disjoint. Let $e_1, e_2, e_3 \in K$ be the edges incident to $v_1, v_2, v_3$ in $R$, respectively. Let $T_0, T_1, T_2, T_3$ be the four components of $\ALG \setminus \{e_1, e_2, e_3\}$ with $v_i \in T_i$ for $i \in \{0, 1, 2, 3\}$, see Figure~\ref{fig:lemma-compatibility-path}. Observe that $\delta_{\REFF}(T_1 : T_0 \cup T_3) = \delta_{\REFF}(T_1 : T_0 \cup T_2) = \emptyset$ because $e_1 \sim_{cp} e_2$ and $e_1 \sim_{cp} e_3$. We deduce that $\delta_{\REFF}(T_1) = \emptyset$, contradicting the fact that $e_1$ is safe.
  \end{proof}
%

  \begin{lemma}\label{lem:comp-class-all-or-nothing}
    Let $S \in \eqS$ be an equivalence class of $\sim_{cp}$ and let $e \in S$.
    Let $f \in \REFF$ such that $\ALG\setminus\{e\}\cup\{f\}$ is feasible.
    Then, $\ALG\setminus S'\cup\{f\}$ is feasible for all $S'\subseteq S$.
  \end{lemma}
  \begin{proof}
    If $e$ is inessential, then all edges in $S$ are inessential by Lemma~\ref{lem:summary:classification}.
    Thus, $\ALG\setminus S$ is feasible even without~$f$ and we are done.
    Otherwise, $e$ is essential, which implies that all edges in $S$ are essential by Lemma~\ref{lem:summary:classification}. So, they are also safe and we can apply ~\ref{lemsum:property:no-F-edges-between-inner-components} from Lemma~\ref{lem:summary:pathlemmata}.
%
    Let $T_0, \dots, T_l$ be the connected components of $\ALG \setminus S$ in the order they are traversed by the path $P$ containing $S$. Statement~\ref{lemsum:property:no-F-edges-between-inner-components} implies that $\delta_{\REFF}(T_i) = \emptyset$ for all $i \in \{1, \dots, l-1\}$. Therefore, for every terminal pair $\{u, \bar{u}\} \in \terms$ there either is an $i \in \{0, \dots, l\}$ with $u, \bar{u} \in T_i$, or $u \in T_0$ and $\bar{u} \in T_l$ (w.l.o.g.). Hence the only terminal pairs that are disconnected by the removal of $S$ are those with $u \in T_0$ and $\bar{u} \in T_l$.
    Statement~\ref{lemsum:property:no-F-edges-between-inner-components} also implies that either $f \in \delta_{\REFF}(T_0 : T_l)$ or both endpoints of $f$ are contained in one of the $T_i$. In the former case, $\ALG \setminus S \cup \{f\}$ is feasible since $T_0 \cup T_l$ is a connected component of this solution. In the latter case, the connected components of $\ALG \setminus \{e\} \cup \{f\}$ and $\ALG \setminus \{e\}$ are the same, which would imply that $A \setminus \{e\}$ is feasible by our assumption, which is a contradiction to $e$ being essential.
  \end{proof}

\begin{lemma}\label{lem:summary:pathlemmata}
Let $S \in \eqS \backslash \{S_u\}$ be an equivalence class of safe edges. Let $f \in \REFF$ be an edge.
\begin{enumerate}[label=\alph*.,ref={\ref{lem:summary:pathlemmata}(\alph*)}]
\item Let $K \subseteq S$. Then there is a path $P \subseteq \ALG$ with $K \subseteq P$. \label{lemsum:property:compatible-path}
\item Let $P \subseteq \ALG$ be the unique minimal path containing $S$ and let $T_0, \dots, T_l$ be the components of $\ALG \setminus S$ in the order they are traversed by $P$. Then either $f \in \delta_{\REFF}(T_0 : T_l)$, or there is an $i \in \{0, \dots, l\}$ such that $T_i$ contains both endpoints of $f$. \label{lemsum:property:no-F-edges-between-inner-components}
\item Let $C$ be the unique cycle in $\ALG \cup \{f\}$.
    Then $S \subseteq C$ or $S \cap C = \emptyset$. \label{lemsum:property:cycle-all-or-nothing}
\item If $\ALG\setminus\{e\}\cup\{f\}$ is feasible for some edge $e \in S$, then $\ALG\setminus S'\cup\{f\}$ is feasible for all $S'\subseteq S$. 
This also holds for $S=S_u$.\label{lemsum:property:comp-class-all-or-nothing}
\end{enumerate}
\end{lemma}
\begin{proof}
The proof of ~\ref{lemsum:property:compatible-path} is done in Lemma~\ref{lem:compatibility-path}. For ~\ref{lemsum:property:no-F-edges-between-inner-components}, denote the edges in $S$ by $e_1,\ldots,e_l$ such that $e_i$ is the edge between $T_{i-1}$ and $T_i$ for all $i \in \{1,\ldots,l\}$. Since for any $i \neq j \in \{1, \dots, l\}$, the edges $e_i$ and $e_j$ are pairwise compatible, no $\REFF$-edge can cross between $T_i$ and $T_j$, for any $i \not= j \in \{1,\dots,l\}$. Nor can there be $\REFF$-edges from $T_0$ to any $T_i$, for $i \in \{1,\dots,l\}$, since $e_1$ and $e_{i+1}$ are compatible. Therefore either both endpoints of $f$ lie within the same component $T_i$, or $f \in \delta_{\REFF}(T_0 : T_l)$. 

Property \ref{lemsum:property:cycle-all-or-nothing} follows from~\ref{lemsum:property:no-F-edges-between-inner-components} for $S \in \eqS \backslash \{ S_u\}$. Notice that for any $f \in \REFF$, all edges on $C$ are automatically safe. Thus, the statement is true for $S_u$ as well.
Statement~\ref{lemsum:property:comp-class-all-or-nothing} is proven in Lemma~\ref{lem:comp-class-all-or-nothing}.
\end{proof}

The following fact is a straightforward generalization of Hall's theorem and an easy consequence of max flow/min cut (see, \eg \cite{chekuri2009longest}).

 \begin{fact}\label{fact:halltype}
 Let $G=(A \cup B, E)$ be a bipartite graph. For any $A' \subseteq A$, let $N(A') \subseteq B$ be the set of neighbors of the nodes in $A'$. If
 \[
 \tforall\ A' \subseteq A: |N(A')| \ge |A'| / c
 \]
 then there exists an assignment $\alpha: E \to \mathbb{R}_+$ such that $\sum_{e \in \delta(a)} \alpha(e) \geq 1$ for all $a \in A$ and $\sum_{e \in \delta(b)} \alpha(e) \leq c$ for all $b \in B$.
 \end{fact}	


\section{Proofs: Bounds on \texorpdfstring{$\phi$}{phi} when the Local Optimum is a Tree}\label{sec-tree-phi}

Let's recall (and formalize) some notation from
\S\ref{sec-treecase-phi}.  Let equivalence class $S \in
\eqS\backslash\{S_u\}$ of safe edges contain $\lenl(S)$ edges.  By
Lemma~\ref{lemsum:property:compatible-path}, the edges of $S = \langle
e_{S,1}, \ldots, e_{S, \ell(S)}\rangle$ lie on a path: let
\label{def:path}
\[
w_{S,0} \xrightarrow{e_{S,1}} v_{S,1} \to \dots \to w_{S,1} \xrightarrow{e_{S,2}} v_{S,2} \to \dots \to w_{S,i-1}  \xrightarrow{e_{S,i}} v_{S,i} \to \dots w_{S,\lenl(S)-1} \xrightarrow{e_{S,\lenl(S)}} v_{S,\lenl(S)}
\]
be this path, where each $e_{S,i}=\{w_{S,i-1},v_{S,i}\}$.  When the
context is clear (as it is in this subsection), we use $e_1,\ldots,
e_{\lenl(S)}$ and $v_1,\ldots,v_{\lenl(S)},w_{0},\ldots,w_{\lenl(S)-1}$
instead.  Notice that $w_{S,i}=v_{S,i}$ is possible, but we always have
$w_{S,i-1}\neq v_{S,i}$.

Removing $S$ decomposes $\ALG$ into $\lenl(S)+1$ connected components.
Let the connected component containing $w_i$ be $(V_{S,i}, E_{S,i})$ and
the connected component that contains $v_{\lenl(S)}$ be
$(V_{S,\lenl(S)}, E_{S,i})$. As in the rest of the paper, we will
associate each components by its edge set.
We think of $E_{S,0}, E_{S,\lenl(S)}$ as
the \emph{outside} of the path that $S$ lies on, and
$E_{S,1},\ldots,E_{S,\lenl(S)-1}$ as the \emph{inner} components.
Note that these components form vertex-disjoint subtrees of $\ALG$.

\begin{observation}\label{obs:AsiAsia-disjoint}
Let $S \in \eqS \backslash\{S_u\}$ and $i, i' \in \{0,\ldots,\lenl(S)\}$, $i\neq i'$. Then $E_{S,i} \cap E_{S,i'} = \emptyset$.
Also, $\ALG$ is the disjoint union of $E_{S,0},\ldots,E_{S,\lenl(S)}$ and $S$.
\end{observation}

Next, we set $\frakm(S)$ and $\frakn(S)$ to be the index of the inner
components with the largest and second-largest widths, respectively.
(We use the indexing $\maxit$ from \S\ref{sec:preliminaries} to break
ties, so $\frakm(S) := \arg\max_{i=1,\ldots,\lenl(S)-1}
\maxit(E_{S,i})$, and $\frakn(S) :=
\arg\max_{i\in\{1,\ldots,\lenl(S)-1\}\backslash\{\frakm(S)\}}\maxit(E_{S,i})$.)
Without loss of generality, we assume that the orientation of the path
is such that $\frakm(S) < \frakn(S)$.
Let $\innS := \{1, 2,\ldots, \lenl(S) - 1\}$ be the indices of the
inner components, and $\innSp := \innS \setminus \{ \frakm(S),
\frakn(S) \}$.  
Now we split the cost of the solution $\ALG$ into three terms. It holds that
\[
\begin{array}{rll}\label{eq:splitted-objective-function}
\phi(\ALG)
=& \multicolumn{2}{l}{ w(\ALG) + \sum\limits_{e \in \ALG} d_f  
    = w(\ALG) + \sum\limits_{e\in S_u} d_e + \sum\limits_{S\in\eqS\backslash\{S_u\}} \sum\limits_{e\in S} d_e}\\
=& \Bigg(w(\ALG) + \sum\limits_{e\in S_u} d_e\Bigg)  
 &  + \Bigg(\sum\limits_{S\in\eqS\backslash\{S_u\}} \bigg( 
\sum\limits_{i=1}^{\lenl(S)} d_{e_i} - \sum\limits_{i\in \innS} w(E_{S,i})\bigg)\Bigg) \\
 &
 & + \Bigg(\sum\limits_{S\in\eqS\backslash\{S_u\}} \sum\limits_{i \in \innS} w(E_{S,i})\Bigg).
\end{array}
\]

\subsection{Bounding the Middle Term}\label{sec:tree:phi:middleterm}

The next lemma and the corollary bound the cost of the middle term.

\begin{lemma}\label{lem:swap-with-three-pieces}
Let $I=(V,E,\terms,d)$ be a Steiner Forest instance and let $\REFF$ be a feasible solution for $I$.
Furthermore, let $\ALG \subseteq E$ be a feasible \emph{tree} solution for $I$ that is 
\edgesetswap-optimal with respect to $\REFF$ and $\phi$.
Set $S \in \eqS\backslash\{S_u\}$ be an equivalence class of safe edges.
Let $f \in \REFF$ be an edge that closes a cycle in $\ALG$ that contains $S$. 
Then
	\[
		d_f \ge \frac{1}{3} \Bigg( \sum_{i=1}^{\lenl(S)} d_{e_i} - \sum_{i\in \innS} w_i\Bigg).
	\] 
\end{lemma}
\begin{proof}
Since $f$ closes a cycle in $\ALG$ that contains $S$, any edge $e$ in $S$ satisfies that $\ALG \setminus\{e\}\cup\{f\}$ is feasible. By Lemma~\ref{lemsum:property:comp-class-all-or-nothing}, this implies that $\ALG \setminus S \cup\{f\}$ is feasible as well. Thus, adding $f$ and removing any subset $S'\subseteq S$ is a feasible \edgesetswap. 
We consider three subsets $S_1':=\{e_1,\ldots,e_{\frakm(S)}\}$, $S_2' :=\{e_{\frakm(S)+1},\ldots,e_{\frakn(S)}\}$ and $S_3':=\{e_{\frakn(S)+1},\ldots,e_{\lenl(S)}\}$.
Since $\ALG$ is \edgesetswap-optimal with respect to $\REFF$ and $\phi$, we know that adding $f$ and removing $S_1'$, $S_2'$ or $S_3'$ is not an improving swap. When adding $f$, we pay $d_f$. When removing $S_1'$, we gain $\sum_{i=1}^{\frakm(S)} d_{e_i}$, but we have to pay the width of the connected components that we create. Notice that we detach the connected components $E_{S,1}$, \ldots $E_{S,\frakm(S)-1}$ from the tree. Since $E_{S,\frakm(S)}$ has the maximum width on the path, the increase in $\phi$ that originates from widths has to be $\sum_{i=1}^{\frakm(S)-1} w(E_{S,i})$. Thus, the fact that the swap is not improving yields that
\[
d_f + \sum_{i=1}^{\frakm(S)-1} w(E_{S,i}) \ge \sum_{i=1}^{\frakm(S)} d_{e_i}
\ \Leftrightarrow\ d_f \ge \sum_{i=1}^{\frakm(S)} d_{e_i} - \sum_{i=1}^{\frakm(S)-1} w(E_{S,i})
\]
Similarly, we gain that
\[
d_f \ge \sum_{i=\frakm(S)+1}^{\frakn(S)} d_{e_i} - \sum_{i=\frakm(S)+1}^{\frakn(S)-1} w(E_{S,i}) \quad \text{and}\quad
d_f \ge \sum_{i=\frakn(S)+1}^{\lenl(S)} d_{e_i} - \sum_{i=\frakn(S)}^{\lenl(S)-1} w(E_{S,i})
\]
from the fact that $(f,S_2')$ and $(f,S_3')$ are not improving. 
The statement of the lemma follows by adding the inequalities and dividing the resulting inequality by three.
\end{proof}

\begin{corollary}\label{cor:bound-weird-term}
    Let $I=(V,E,\terms,d)$ be a Steiner Forest instance and let 
    $\REFF$ be a feasible solution for $I$.
    Furthermore, let $\ALG \subseteq E$ be a feasible \emph{tree} solution for $I$ that is \edgeedge and \edgeset swap-optimal with respect to $\REFF$ and $\phi$. Then,
		\[
		\sum_{S \in \eqS \backslash\{S_u\}} \bigg( \sum_{i=1}^{\lenl(S)} d_{e_i} - \sum_{i\in \innS} w(E_{S,i})  \bigg) \le 10.5 \cdot \sum_{e\in \REFF} d_e.
		\]
\end{corollary}
\begin{proof}
We set $\Delta(S) := \frac{1}{3} \left(\sum_{i=1}^{\lenl(S)} d_{e_i} - \sum_{i\in \innS} w(E_{S,i})\right)$. The corollary then follows by Lemma~\ref{lem:swap-with-three-pieces} and Theorem~\ref{thm:four-approx-for-trees-hall}.
\end{proof}

\subsection{Bounding the First Term}

The following lemma bounds the first term of the cost.

\begin{lemma}\label{lem:bound-for-rem-optimal}
Let $I=(V,E,\terms,d)$ be a Steiner Forest instance and let 
    $\REFF$ be a feasible solution for $I$ with $cc(\REFF)$ connected components $F_1,\ldots,F_{cc(\REFF)}$.
    Furthermore, let $\ALG \subseteq E$ be a feasible \emph{tree} solution for $I$ and recall that $S_u$ is the equivalence class of unsafe edges in $\ALG$. Assume that $\ALG$ is removing swap-optimal.
		It holds that
		\[
		w(\ALG) + \sum_{e\in S_{u}} d_e\le \sum_{i=1}^{cc(\REFF)} w(F_i).
		\]
\end{lemma}
\begin{proof}
$\ALG\backslash S_u$ contains $|S_u|+1$ connected components $E_1,\ldots,E_{|S_u|+1}$, number them by decreasing width, \ie $w(E_1) \ge w(E_2) \ge \ldots \ge w(E_{|S_u|+1})$. Since $\REFF$ does not connect them by definition of $S_u$, its number of connected components satisfies $cc(\REFF) \ge |S_u|+1$. Furthermore, any component $E_i$ contains a component $F_j$ with $w(F_j) = w(E_i)$. Assume that the components are numbered such that $w(F_i)=w(E_i)$ for $i=1,\ldots,|S_u|+1$. 

Consider the removing swap that removes all edges in $S_u$. Before the swap, $\ALG$ paid all edges and $w(\ALG)$. After the swap, $\ALG$ would pay all edges except those in $S_u$ plus $\sum_{i=1}^{|S_u|+1} w(E_i)=\sum_{i=1}^{|S_u|+1} w(F_i)$. Since the swap was not improving, we know that
\begin{align*}
& w(\ALG) + \sum_{e\in\ALG}d_e \le \sum_{i=1}^{|S_u|+1} w(F_i) + \sum_{e\in\ALG}d_e - \sum_{e\in S_u}d_e.
\end{align*}
Thus, $w(\ALG) + \sum_{e\in S_u}d_e \le \sum_{i=1}^{|S_u|+1} w(F_i) \le \sum_{i=1}^{cc(\REFF)} w(F_i)$.
\end{proof}

\subsection{Bounding the Third Term: Sum of Widths}\label{sec:tree:phi:lastterm}
 
For the third term, we need a bit more work. 
The main step is to show that $\maxit(E_{S, i})$ is an injective function of $S, i$ when restricting to $i \in \innS$. 
The following lemma helps to prove this statement.

\begin{figure}
    \begin{subfigure}{\textwidth}
        \centering
        \begin{tikzpicture}[
        graphnode/.style={label distance=0.0mm, fill, draw, circle, inner sep=0pt, minimum size=4pt},
        cc/.style={draw, rounded corners, dashed},
        graphedge/.style={thick},
        pathedge/.style={graphedge, dotted},
        font=\small
      ]
        \draw (-3.9, 0) node[graphnode, label={below:$v_i$}] (w1) {};
        \draw (-2.4, 0)  node[graphnode] (w2) {};
        \draw (-1.3, 0)  node[graphnode] (v1) {};
        \draw ( 0.2, 0) node[graphnode] (v2) {};
        \draw ( 1.3, 0)  node[graphnode] (vss) {};
        \draw ( 2.8, 0) node[graphnode, label={below:$w_i$}] (vs) {};

        \path[cc] (w1) ++(-1.50, -0.6) rectangle node[above=0.5cm]{$E^1_{S,<}$} ++(1.8, 1.0);
        \path[cc] (w2) ++(-0.30, -0.6) rectangle node[above=0.5cm]{$E^1_{S,i}$} ++(1.7, 1.0);        
        \path[cc] (v2) ++(-0.30, -0.6) rectangle node[above=0.5cm]{$E^2_{S,i}$} ++(1.7, 1.0);        
        \path[cc] (vs) ++(-0.25, -0.6) rectangle node[above=0.5cm]{$E^1_{S,>}$} ++(1.8, 1.0);        
        
        \draw[graphedge] (w1) -- node[above]{$e_i$} (w2);
        \draw[pathedge]  (w2) -- (v1);
        \draw[graphedge] (v1) -- node[above]{$e'$} (v2);
        \draw[pathedge]  (v2) -- (vss);
        \draw[graphedge] (vss) -- node[above]{$e_{i+1}$} (vs);
        
        \draw (-5.1, -0.3) node[graphnode, gray] (fs1) {};
        \draw (-4.5,  0.1) node[graphnode, gray] (ft1) {};
        \draw[graphedge, gray] (fs1) -- node[above, inner sep=0mm]{$f'$} (ft1);

        \draw (3.2, 0.1)  node[graphnode, gray] (fs2) {};
        \draw (3.9, -0.3) node[graphnode, gray] (ft2) {};
        \draw[graphedge, gray] (fs2) -- node[above right, inner sep=0mm]{$f'$} (ft2);

         \end{tikzpicture}
      \caption{First case: 
        Removing $e'$ disconnects $v_i$ and $w_i$. 
        Thus, $e_i$, $e'$ and $e_{i+1}$ lie on a path. 
        No $\REFF$-edge can leave $E_{S,i}=E_{S,i}^1\cup E_{S,i}^2$. 
        Also, no $\REFF$-edge can go between $E_{S,i}^1$ and $E_{S,i}^2$ since $e'$ 
          and $f'$ are compatible. 
        Thus, $e'$ is compatible to $e_i$ and $e_{i+1}$.}
    \end{subfigure}\vspace{.25\baselineskip}
    \begin{subfigure}{\textwidth}
      \centering
      \begin{tikzpicture}[
      graphnode/.style={label distance=0.0mm, fill, draw, circle, inner sep=0pt, minimum size=4pt},
      cc/.style={draw, rounded corners, dashed},
      graphedge/.style={thick},
      pathedge/.style={graphedge, dotted},
      font=\small
      ]

      \coordinate (v) at (-1.7, 0.0);  
      \draw (-2.5,  0.5 )  node[graphnode]            (vRE) {};
      \draw (-0.7,  0.75)  node[graphnode]          (vRF) {};
      \draw (-4.0,  0.5 )  node[graphnode, label={below:$v_2$}]            (vER) {};
      \draw ( 0.8,  0.75)  node[graphnode, label={below:$v_3$}]            (vFR) {};
      \draw (-1.2, -0.6 )  node[graphnode]               (vRG) {};
			\draw ( 0.8, -0.6 )  node[graphnode, label={below:$v_1$}]            (vGR) {};
            
      \path[cc] (vER) ++(-1.5, -0.6) rectangle node[above=0.5cm]{$E_{S,<}$} ++(1.75, 1.0);  
      \path[cc] (vRE) ++(-0.3, -1.6) rectangle node[above=1.1cm]{$E^1_{S,i}$} ++(2.4, 2.2);    
      \path[cc] (vGR) ++(-0.5, -0.5) rectangle node[right=0.5cm]{$E^2_{S,i}$} ++(1.0, 0.75);
      \path[cc] (vFR) ++(-0.25, -0.6) rectangle node[above=0.5cm]{$E_{S,>}$} ++(1.75, 1.0);
      
      \draw[pathedge]  (v)   -- (vRE);
      \draw[pathedge]  (v)   -- (vRF);
      \draw[pathedge]  (v)   -- (vRG);
      \draw[graphedge] (vRE) -- node[rectangle, inner sep=0mm, above=1mm]{$e_i$} (vER);
      \draw[graphedge] (vRF) -- node[rectangle, inner sep=0mm, above=1mm]{$e_{i+1}$} (vFR);
			\draw[graphedge] (vRG) -- node[auto,xshift=2mm]{$e'$}  (vGR);
      
        \draw (-5.1,  0.2) node[graphnode, gray] (fs1) {};
        \draw (-4.5,  0.6) node[graphnode, gray] (ft1) {};
        \draw[graphedge, gray] (fs1) -- node[above, inner sep=0mm]{$f'$} (ft1);

        \draw (1.4,  0.8)  node[graphnode, gray] (fs2) {};
        \draw (2.1,  0.4) node[graphnode, gray] (ft2) {};
        \draw[graphedge, gray] (fs2) -- node[above right, inner sep=0mm]{$f'$} (ft2);
    \end{tikzpicture}
      \caption{Second case: The nodes $v_i$ and $w_i$ lie in the same connected component when removing $e'$. Again, no $\REFF$-edge can leave $E_{S,i}$. Also, no edge can go between $E_{S,i}^1$ and $E_{S,i}^2$ because $e'$ and $f'$ are compatible. Thus, $e'$ is compatible to $e_i$ and $e_{i+1}$.
      }
    \end{subfigure}
   
    \caption{\label{fig:s-in-asi-or-not}%
		 Two cases that occur in Lemma~\ref{lem:s-in-asi-or-not}.}
    \end{figure}
		
\begin{lemma}\label{lem:s-in-asi-or-not}
Let $S, S' \in \eqS \backslash \{S_u\}$ with $S\neq S'$. Exactly one of the following two cases holds:
\begin{itemize}
\item $S' \subseteq E_{S,0} \cup E_{S,\lenl(S)}$, \ie $S'$ lies in the outside of the path of $S$
\item There exists $i\in\{1,\ldots,\lenl(S)-1\}$ with $S' \subseteq E_{S,i}$.
\end{itemize}
\end{lemma}
\begin{proof}
Notice that $S \cap S' = \emptyset$ by definition because $S$ and $S'$ are different equivalence classes. 
Let $e',f' \in S'$, \ie $e',f' \notin S$.
Assume that $e' \in E_{S,i}$ for an $i\in\{1,\ldots,\lenl(S)-1\}$, but $f' \notin E_{S,i}$. 
Following the notation on page~\pageref{def:path}, we denote the unique edges in $S$ that touch $E_{S,i}$ by $e_{S,i}$ and $e_{S,i+1}$ and their adjacent vertices in $E_{S,i}$ by $v_i$ and $w_i$. 
Removing $e_{S,i}$ and $e_{S,i+1}$ creates three connected components in $\ALG$. 
The first component contains $E_{S,0},\ldots,E_{S,i-1}$ and the edges $e_{S,1},\ldots,e_{S,i-1}$ and we name it $E_{S,<}$. 
The second component is $E_{S,i}$. 
The third component contains $E_{S,i+1},\ldots,E_{S,\lenl(S)}$ and the edges $e_{S,i+1},\ldots,e_{S,\lenl(S)}$ and we name it $E_{S,>}$. 
Since $f' \notin E_{S,i}$, and $f'\notin S$, we know that $f' \in E_{S,<} \cup E_{S,>}$.
 
Notice that $e'$ and $f'$ are compatible, $e_{S,i}$ and $e_{S,i+1}$ are compatible, but $e'$ is not compatible to either $e_{S,i}$ or $e_{S,i+1}$. 
If we remove $e'$ in addition to $e_{S,i}$ and $e_{S,i+1}$, then $E_{S,i}$ decomposes into two connected components, $E_{S,i}^1$ and $E_{S,i}^2$. 
We assume without loss of generality that $v_{S,i} \in E_{S,i}^1$. 
The vertex $w_{S,i}$ can either be separated from $v_{S,i}$ when removing $e'$, \ie $w_{S,i} \in E_{S,i}^2$, or not, \ie $w_{S,i} \in E_{S,i}^1$. 
The two cases are depicted in Figure~\ref{fig:s-in-asi-or-not}.

If $w_{S,i} \in  E_{S,i}^2$, then the edges $e_{S,i}$, $e'$ and $e_{S,i+1}$ lie on a path. 
No $\REFF$-edge can leave $E_{S,i}=E_{S,i}^1\cup E_{S,i}^2$ because $e_{S,i}$ and $e_{S,i+1}$ are compatible. 
Thus, any $\REFF$-edge leaving $E_{S,i}^1$ has to end in $E_{S,i}^2$. 
On the other hand, consider removing $e'$ and $f'$ from $\ALG$, which creates three components $T_{e'}, T_{e',f'}, T_{f'}$. 
As both $E_{S,i}^1$ and $E_{S,i}^2$ are incident to $e'$, they are contained in two different components, $T_{e'}$ and $T_{e',f'}$.
Since no edge can leave $T_{e',f'}$ because $e'$ and $f'$ are compatible, no edge can cross between $E_{S,i}^1$ and $E_{S,i}^2$.
Thus, no $\REFF$-edges leave $E_{S,i}^1$ or $E_{S,i}^2$, which means that $e'$ is compatible to $e_{S,i}$ and $e_{S,i+1}$.

If $w_{S,i} \in E_{S,i}^1$, the argument is similar. Still, no $\REFF$-edge can leave $E_{S,i}=E_{S,i}^1\cup E_{S,i}^2$ because $e_{S,i}$ and $e_{S,i+1}$ are compatible. 
Also, $E_{S,i}^1 \subseteq T_{e',f'}$ and $E_{S,i}^2 \subseteq T_{e'}$ still are in different connected components when $e'$ and $f'$ are removed. 
Thus, no $\REFF$-edge can connect them, and thus $e'$ is compatible to $e_{S,i}$ and $e_{S,i+1}$.

Both cases end in a contradiction since $e'$ cannot be compatible to any edge in $S$ because $S$ and $S'$ are different equivalence classes. 
Thus, $e' \in E_{S,i}$ implies that $S \subseteq E_{S,i}$. 
Only if no edge in $S'$ lies in any $E_{S,i}$ for $i\in\{1,\ldots,\lenl(S)-1\}$, then it is possible that edges of $S'$ lie in $E_{S,0}$ or $E_{S,\lenl(S)}$. 
However, then $S' \subseteq (E_{S,0} \cup E_{S,\lenl(S)})$.
\end{proof}

We set $\mu(S,i) = \maxit(E_{S,i})$ for all $S \in \eqS\setminus\{S_u\}$ and $i \in \innS$. 

\begin{lemma}\label{lem:mu-is-injective}
Let $S, S' \in \eqS \backslash\{S_u\}$, $i \in \innS$, $i' \in \innSp$. Then 
\[
\mu(S,i) = \mu(S',i') \Rightarrow S=S' \text{ and }i=i',
\]
i.e., the mapping $\mu$ is injective.
\end{lemma}
\begin{proof}
Let $\mu(S,i) = \mu(S',i')$ and denote by $u^\ast :=  u_{\mu(S,i)} = u_{\maxit(E_{S,i})}$ and $\bar{u}^* := \bar{u}_{\mu(S,i)} = \bar{u}_{\maxit(E_{S,i})}$ the corresponding terminal pair. 
By the definition of $\maxit$, it follows that $u^\ast, \bar{u}^\ast \in V_{S,i} \cap V_{S',i'}$ and in particular, we know that $V_{S,i} \cap V_{S',i'} \not= \emptyset$.  
If $S=S'$, then either $i=i'$ (and we are done) or $i\not= i'$ implies that $V_{S,i} \cap V_{S',i'} = \emptyset$ which is a contradiction. 
Thus, assume in the following that $S\not=S'$.
  
Let $P$ be the unique $u^*$-$v_{S, i}$-path in $E_{S,i}$. 
If $P \cap S' = \emptyset$, then $P \subseteq E_{S', i'}$ because $u^* \in V_{S',i'}$. 
Thus, $v_{S,i} \in V_{S',i'}$, which also implies that $e_{S,i} \in E_{S', i'}$. By Lemma~\ref{lem:s-in-asi-or-not}, we get $S \subseteq E_{S', i'}$. 
On the other hand, if $P \cap S' \neq \emptyset$, then $S' \subseteq E_{S,i}$ by Lemma~\ref{lem:s-in-asi-or-not}. 
Thus $S \subseteq E_{S', i'}$ or $S' \subseteq E_{S,i}$.
W.l.o.g., we assume the latter. 
Now let $j \in \{\frakm(S'), \frakn(S')\}$ be such that $E_{S', j} \cap S = \emptyset$ (note that $j$ exists due to Lemma~\ref{lem:s-in-asi-or-not}). 
Because $e_{S',j}, e_{S',j+1} \in S' \subseteq E_{S,i}$ and $E_{S', j} \cap S = \emptyset$, we deduce that $E_{S',j} \subseteq E_{S,i}$. 
This implies that $\mu(S, i) \geq \maxit(E_{S',j}) > \mu(S', i')$, again a contradiction.
\end{proof}

We can now bound the third term of our objective function (see page~\pageref{eq:splitted-objective-function}).
\begin{lemma}\label{lem:bound-width-term}
  As previously, let $\eqS$ denote the set of all equivalence classes of the
  compatibility relation $\sim_{cp}$, let $\ALG$ be a feasible tree and 
  denote by $E_{S,i}$ the $i$-th inner connected component on the path 
  that contains $S \in \eqS\setminus\{S_u\}$.
  Let $\REFF$ be a feasible Steiner Forest with $cc(\REFF)$ connected 
  components $F_1,\dots,F_{cc(\REFF)}$.
  Then,
  \begin{align*}
    \sum_{S\in\eqS\setminus\{S_u\}} \sum_{i \in \innS} w(E_{S,i}) \leq 
    \sum_{i=1}^{cc(\REFF)} w(F_i)
  \end{align*}
\end{lemma}
\begin{proof}
  As previously, we set $\mu(S,i) = \maxit(E_{S,i})$ for all $S \in \eqS\setminus\{S_u\}$ and $i \in \innS$.
  We also recall that by definition, $w(E_{S,i}) = d_G(u_{\maxit(S,i)}, \bar{u}_{\maxit(S,i)})
  = d_G(u_{\mu(S,i)}, \bar{u}_{\mu(S,i)})$
  for all $S \in \eqS\setminus\{S_u\}$ and all $i\in \innS$.
  We thus have
  \begin{align}\label{eq:lem-width-bound-1}
       \sum_{S\in\eqS\setminus\{S_u\}} \sum_{i \in \innS} w(E_{S,i})
    =& \sum_{S\in\eqS\setminus\{S_u\}} \sum_{i \in \innS} d_G(u_{\mu(S,i)}, \bar{u}_{\mu(S,i)}).
  \end{align}
  Let $\chi(S,i)$ denote the index of the connected component $F_{\chi(S,i)}$  
  containing the terminal pair $u_{\maxit(S,i)}, \bar{u}_{\maxit(S,i)}$ in $\REFF$. 
  We claim that $\chi$ is injective. To see this, consider $S,S' \in \eqS$ and $i\in\innS$, $i'\in\innSp$ with $\chi(S,i) = \chi(S',i')$. 
  Since $\delta_{\REFF}(V_{S,i}) = \delta_{\REFF}(V_{S',i'}) = \emptyset$ and $F_{\chi(S,i)}$ is connected, we deduce that $V[F_{\chi(S,i)}] \subseteq V_{S,i} \cap V_{S',i'}$. This implies that $u_{\mu(S, i)}, \bar{u}_{\mu(S, i)} \in V_{S', i'}$ and that $u_{\mu(S', i')}, \bar{u}_{\mu(S', i')} \in V_{S, i}$. Hence $\mu(S, i) = \mu(S', i')$, which implies $S = S'$ and $i = i'$ by Lemma~~\ref{lem:mu-is-injective}.

  Since $d_G(u_{\mu(S,i}, \bar{u}_{\mu(S,i)}) \leq w(F_{\chi(S,i)})$, we can now 
  continue \eqref{eq:lem-width-bound-1} to see that
  \begin{align*}
          \sum_{S\in\eqS\setminus\{S_u\}} \sum_{i \in \innS} d_G(u_{\mu(S,i)}, \bar{u}_{\mu(S,i)})
    &\leq \sum_{S\in\eqS\setminus\{S_u\}} \sum_{i \in \innS} w(F_{\chi(S,i)})
    &\leq \sum_{i=1}^{cc(\REFF)} w(F_i).
  \end{align*}
  Here, the last inequality follows from our argument that $\chi$ is injective.
\end{proof}

\subsubsection{Wrapping Things Up}

We can now prove the main result of this section. 
\begin{theorem}\label{thm:approx-guarantee-tree-case}
Let $G=(V,E)$ be a graph, let $d_e$ be the cost of edge $e \in E$ and let $\terms \subseteq V\times V$ be a terminal set. 
Let $\ALG, \REFF \subseteq E$ be two feasible solutions for $(G,d,\terms)$ with $V[\ALG]=V[\REFF]$.
Furthermore, suppose that $\ALG$ is a tree and that $\ALG$ is 
\edgesetswap optimal with respect to $\REFF$ and $\phi$. 
Then,
 \begin{align*}
   \phi(\ALG) = \sum_{e\in\ALG} d_e + w(\ALG) & \le 10.5 \cdot d(\REFF) + w(\REFF) + \sum_{e \in S_u} d_e + w(\ALG). 
 \end{align*}
In particular, $d(\ALG) \le 10.5 \cdot d(\REFF) + w(\REFF) + \sum_{e \in S_u} d_e$.
\end{theorem}
\begin{proof}
Rewrite $\phi(\ALG)$ as on page~\pageref{eq:splitted-objective-function} to 
  \begin{align*}
  \phi(\ALG) = w(\ALG) + \sum\limits_{e\in S_u} d_e
             + \underbrace{\sum_{S\in\eqS\backslash\{S_u\}} \big( \sum_{i=1}^{\lenl(S)} d_{e_i}
               - \sum_{i\in \innS} w(E_{S,i})\big)}_{%
               \leq 10.5 \cdot d(\REFF)\ \text{by Corollary~\ref{cor:bound-weird-term}}
             } 
             + \underbrace{\sum_{S\in\eqS\backslash\{S_u\}} \sum_{i \in \innS} w(E_{S,i})}_{%
               \leq w(\REFF)\ \text{by Lemma~\ref{lem:bound-width-term}}
             }.
  \end{align*}
	This proves the theorem.
\end{proof}
Additionally applying Lemma~\ref{lem:bound-for-rem-optimal} yields the following reformulation of Theorem~\ref{thm:approx-guarantee-tree-case}.
\begin{corollary}
  Let $G=(V,E)$ be a graph, let $d_e$ be the cost of edge $e \in E$ and 
  let $\terms \subseteq V\times V$ be a terminal set. 
  Let $\ALG, \REFF \subseteq E$ be two feasible solutions for $(G,d,\terms)$ with $V[\ALG]=V[\REFF]$.
  Furthermore, suppose that $\ALG$ is a tree and that $\ALG$ is optimal 
  with respect to $\REFF$ and $\phi$ under \edgeedge, \edgeset and removing swaps.
  Then,
	\[\pushQED{\qed}
	\phi(\ALG) \le 10.5 \cdot d(\REFF) + w(\REFF) + \sum_{e \in S_u} d_e + w(\ALG)
		\le 10.5 \cdot d(\REFF) + 2 w(\REFF) \le 10.5 \cdot \phi(\REFF) \qedhere
	\popQED\]
\end{corollary}
If we want to bound the original objective function of the Steiner Forest problem, we do not need removing swaps.

\treecaseresult*
\begin{proof}
  Let $S_{is}$ be the set of safe but inessential edges. 
  We have $d(\ALG) \leq 10.5 \cdot d(\REFF) + w(\REFF) + \sum_{e \in S_u} d_e$ by Theorem~\ref{thm:approx-guarantee-tree-case}. That implies 
  \begin{align*}
    \sum_{e \in \ALG'} d_e = d(\ALG) - \sum_{e\in S_u} d_e - \sum_{e \in S_{is}} d_e
		\le 10.5 \cdot d(\REFF) + w(\REFF) - \sum_{e \in S_{is}} d_e.\\[-3\baselineskip]
  \end{align*}
\end{proof}

\section{Proofs: Bounds when the Local Optimum is a Forest}
\label{sec:app-forest}

The aim of this section is to transform (a.k.a.\ ``project'') a pair
$(\ALG, \REFF)$ of arbitrary solutions into a pair $(\ALG, \REFF')$ of
solutions to which the results from the previous section apply. I.e.,
each connected component of $\REFF'$ is contained within a connected
component of $\ALG$.
The main lemma is
\begin{restatable*}{lemma}{reductionforesttotree}
\label{lem:main-forest:reduction-to-tree}
  Let $G=(V,E)$ be a complete graph, let $d: E \to \Rp$ be a metric that assigns a cost $d_e$ to every edge $e \in E$ and 
  let $\terms \subseteq V\times V$ be a set of terminal pairs.
  Let $\ALG, \REFF \subseteq E$ be two feasible Steiner Forest solutions for $(G,\terms)$.
  Furthermore, suppose that $\ALG$ is \edgeedge, \edgeset and \pathset swap-optimal with respect to $E$ and $\phi$, that $\ALG$ is $c$-approximate connecting move optimal and that $\ALG$ only uses edges between terminals. Then there exists a feasible solution $\REFF'$ with $d(\REFF') \le 2(1+c)\cdot d(\REFF)$ that satisfies $\REFF'={\REFF'}_{\circlearrowright}$ such that $\ALG$ is \edgeedge and \edgeset swap-optimal with respect to $\REFF'$.
	\end{restatable*}
	
We use the notation $\REFF_{\circlearrowright}$ to denote the set of edges in $\REFF$ that go within components of $\ALG$ and $\REFF_{\leftrightarrow}$ to denote the set of all edges between different components, and $c$ is the approximation guarantee of the approximate connecting moves. Formally, $c$-approximate tree move optimality is defined as follows:
\begin{definition}
A \emph{$c$-approximate connecting move} for some constant $c \geq 1$ is a connecting move $conn(T)$ applied to the current solution $\ALG$ using a tree $T$ in $G_{\ALG}^{\text{all}}$ such that $c \cdot d(T) \le \bar{w}(\ALG) - \bar{w}(\ALG \cup T)$. A solution is \emph{$c$-approximate connecting move optimal}, if there are no $c$-approximate connecting moves.
\end{definition}

Lemma~\ref{lem:main-forest:reduction-to-tree} shows that for any solution $\REFF$, we can find a solution $\REFF'$ with 
\[
d(\REFF') \le 2(1+c)\cdot d(\REFF)
\]
which does not contain edges between different components of $\ALG$, \ie $\REFF'_{\circlearrowright}=\REFF'$.
With Section~\ref{sec:wkmst}, we know that $c$ is at most $2$,
\msnote{Next time that $c=5$ is used}
\ie we can find $\REFF'$ with $d(\REFF') \le 6 \cdot d(\REFF)$.
Every connected component $A_j$ of $\ALG$ can now be treated separately by using Corollary~\ref{treecaseresult} on $A_j$ and the part of $\REFF'$ that falls into $A_j$. By combining the conclusions for all connected components, we get that \msnote{Again, $138$ would become $69$}
\[
d(\ALG') \le 11.5 \cdot d(\REFF') \le 23(1+c)\cdot d(\REFF) \le 69 \cdot d(\REFF)
\]
for any feasible solution $\REFF$. This proves the main theorem.

\maintheorem*

\paragraph*{Proof outline}
The forest case depends on \pathsetswap{}s and connecting moves. Exploiting connecting move-optimality is the main effort of the section, while \pathsetswap-optimality is only used to handle one specific situation.

The goal is to replace $\REFF_{\leftrightarrow}$ by edges that go within components of $\ALG$.

We first convert $\REFF$ into a collection of disjoint cycles at the expense of a factor of 2 in the edge costs. 
Let $F_i$ be one of the cycles. We want to replace $(F_i)_{\leftrightarrow}$. 
To do so, we look at $F_i$ in $G_\ALG$ -- the graph where the edges in $\REFF_{\circlearrowright}$ are contracted and loops are removed. In this graph $G_\ALG$, the set $(F_i)_{\leftrightarrow}$ is a circuit (\ie a possibly non-simple cycle). 
In a first step, we use Algorithm~\ref{alg:charging} to cope with the case that $(F_i)_{\leftrightarrow}$ has a special structure that we call \emph{minimally guarded}.
The second step inductively ensures that this structure is present. 

An example run of Algorithm~\ref{alg:charging} is visualized in Figure~\ref{fig:charging-algo-example}. 
The algorithm partitions $(F_i)_{\leftrightarrow}$ into trees in $G_\ALG$. These trees define connecting moves, so (approximate) connecting move-optimality gives us a lower bound on the total edge weight of each tree.

Let $n_j$ be the number of times that $F_i$ passes through $A_j$, the $j$-th connected component of $\ALG$. 
Then the lower bound that we get is 
\[
d((F_i)_{\leftrightarrow}) \ge \left(\sum_{j=1}^p n_j w(A_j) \right) - n_{j_m} w(A_{j_m}).
\]
Here, $j_m$ is the component with the largest width among all components that $F_i$ touches. 
In a minimally guarded circuit, this component is only visited once. 
The lower bound (for minimally guarded circuits) results from Lemma~\ref{lem:treedecomp-to-lowerbound} and Corollary~\ref{lem:minimallyguarded} (notice that $n_j$ is defined differently in the actual proof, but we need less detail here). The part between the two statements establishes invariants of Algorithm~\ref{alg:charging} that we need to show that it computes trees with the correct properties.

The lower bound means that we can do the following. Assume that we delete $(F_i)_{\leftrightarrow}$. Now some vertices in $\ALG$ do no longer have adjacent $\REFF$-edges which is a problem for applying Corollary~\ref{treecaseresult}. 
We fix this by inserting a direct connection to their terminal partner, which has to be in the same component of $\ALG$. 
In order to keep the modified solution feasible, we insert the new connections a bit differently, but with the same result.
This connection is paid for by $(F_i)_{\leftrightarrow}$ which -- due to our construction -- can give each of these vertices a budget equivalent to the width of its component.
This enables us to use the results from the tree case. 

This argument does not work for the vertices in the largest-width component.  If $F_i$ is minimally guarded, it visits the largest width component exactly once and there are exactly two problematic vertices.
To reconnect these vertices, we use \pathsetswap-optimality and charge $(F_i)_{\leftrightarrow}$ again to pay for directly connecting them (this argument comes later, in the proof of Lemma~\ref{lem:guardedcycles}, when Corollary~\ref{lem:minimallyguarded} is applied).

If otherwise $F_i$ is not minimally guarded, we need a second step. This is taken care of in Lemma~\ref{lem:guardedcycles} which 
extends Corollary~\ref{lem:minimallyguarded} to \emph{guarded} circuits that are not necessarily minimal. 
This is done by applying Corollary~\ref{lem:minimallyguarded} to subcircuits and removing these until the minimality criterion is met. The proof is by induction. Lemma~\ref{lem:guardedcycles} outputs a broken solution $\REFF'$ which is not feasible, but is equipped with the widths from the lower bound. Figures~\ref{fig:ex:1} up to \ref{fig:ex:last} show the recursive process for an example circuit and visualize the broken solution that comes out of it. 
The solution is then repaired in Lemma~\ref{lem:main-forest:reduction-to-tree}, giving the final reduction result.

\subsection{Details: Getting a Good Tree Packing}\label{sec:happytreepacking} 

While the previous sections operated on arbitrary graphs, we now consider the metric case of the Steiner forest problem. 
Consequently, we assume that $G=(V,E)$ is the complete graph on $V$ and that the cost $d_e = d_{vw}$ of each edge $e=\{v,w\} \in E$ is given by a metric $d: V\times V \to \Rp$. 
Together with a set of terminal pairs $\terms \subseteq V$, the graph $G$ and the metric $d$ define an instance of the metric Steiner Forest problem.

The more important change in our setting is, however, that we no longer assume that our feasible solution $\ALG \subseteq E$ is a tree; rather, $\ALG$ can be an arbitrary feasible forest in $G$. 
We write its connected components as $A_1,\dots,A_p \subseteq \ALG$, where the numbering is fixed and such that $w(A_1) \leq w(A_2) \leq \dots \leq w(A_p)$ holds.
As in the previous sections, we compare $\ALG$ to another solution $\REFF$. 

The connected components of $\ALG$ are trees and it is a natural idea to apply the results from the previous sections to these trees by considering each connected component of $\ALG$ individually. 
Morally, the main obstacle in doing so is the following: 
In the proof of Theorem~\ref{thm:four-approx-for-trees-hall}, we assume implicitly that no $\REFF$-edge crosses between connected components of $\ALG$.\footnote{More precisely, we need the slightly weaker condition that for each node $t \in V[\ALG]$, there is an $\REFF$-edge incident to $t$ that does not leave the connected component of $\ALG$ containing $t$.}
This is vacuously true in the case where $\ALG$ is a tree; however, if $\ALG$ is a forest, this assumption is not justified in general. 
In the following, our underlying idea is to replace $\REFF$-edges that cross between the components of $\ALG$ by edges that lie within the components of $\ALG$, thereby re-establishing the preconditions of Theorem~\ref{thm:four-approx-for-trees-hall}.
We show how to do this such that $\REFF$ stays feasible and such that its cost is increased by at most a constant factor.


\subsubsection{Notation}

We start with some normalizing assumptions on $\REFF$. As before, we
denote the connected components of $\REFF$ by $F_1,\dots,F_q \subseteq
\REFF$. First we can assume that each $F_i$ has no inessential edges.
Then, since we are in the metric case, we can convert each $F_i$ into a
simple cycle---this can be done with at most a factor $2$ loss in the
cost of $\REFF$ by taking an Euler tour and short-cutting over repeated
vertices and non-terminals.
This implies that that $V[\REFF]\subseteq V[\ALG]$, since now $\REFF$
only has terminals, which are all covered by $\ALG$.
Assume that $V[\ALG]$ only contains terminals. Then, $V[\REFF]$ and $V[\ALG]$ are equal.
Recall that $\terms$ is the set of terminal pairs, let $V_\terms$ be the set of all terminals.
Henceforth, we assume that $V= V[\REFF] = V[\ALG] = V_\terms$.

\begin{observation}\label{obs:f-assumptions}
Let $\REFF, \ALG$ be feasible solutions and assume that $V[\ALG]=V_\terms$.
Then there exists a solution $\REFF'$ of cost $2\cdot\sum_{e \in \REFF} d_e + \sum_{i=1}^q w(F_q)$ whose connected components are node disjoint cycles and which satisfies $V[\REFF'] = V[\ALG]= V_\terms$.
\end{observation}

Next, we define a convenient notation for the connected components of $\ALG$ and the $\REFF$-cycles that pass through them.
For each $v\in V$, we set $\xi(v) = j$ for the unique $j \in\{1,\dots,p\}$ that satisfies that $v \in A_j$, \ie $\xi(v)$ is the index of the connected component of $\ALG$ that contains $v$.
Using this notation, we define a graph $G_\ALG = G \contract \{A_1,\dots,A_p\} = (V_\ALG, E_\ALG)$ which results from contracting the connected components of $\ALG$ in $F$. We set
\begin{align*}  
  V_\ALG := \{1,\dots,p\}\quad\text{and}\quad E_\ALG:=\cset{e_f}{f=\{v,w\} \in \REFF, \xi(v) \not= \xi(w)}.
\end{align*}
It is important that this definition removes all loops induced by the contraction, but it retains possible parallel edges. 
In this way, the edges of $G_\ALG$ correspond to the edges of $\REFF$ that are crossing between the connected components of $\ALG$, while the nodes correspond to the components.
Notice that $G_\ALG$ can be seen as a subgraph of $G_\ALG^{\text{all}}$ defined in the preliminaries for connecting moves. Thus, every tree in $G_\ALG$ induces a connecting move.
We extend $d$ to $G_\ALG$ by setting $d_{e_f} = d_{f}$ for all $e_f \in E_\ALG$.
We also consider the graph $\hat G_\ALG$ on $V_\ALG$ that is the transitive closure of $G_\ALG$. For all pairs $j_1,j_2 \in V_\ALG$, it contains an additional edge $e'_{j_1j_2}$ whose weight $d_{e'_{j_1j_2}}$ is given by the length of a shortest $j_1$-$j_2$-path in $G_\ALG$.

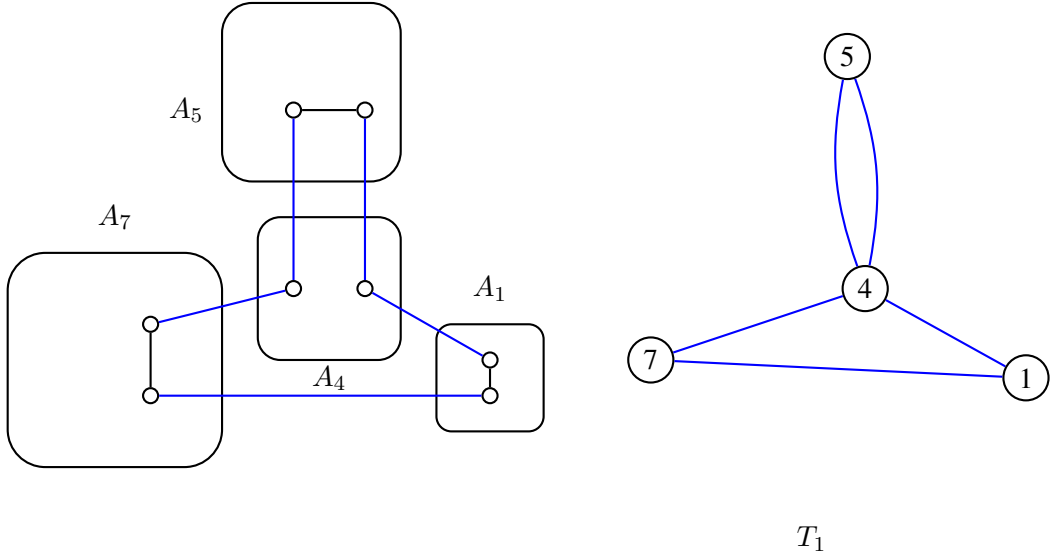
\begin{figure}
\begin{center}
\begin{tikzpicture}[scale=0.95,thick, mynode/.style={draw, circle,inner sep=0cm, minimum size=0.2cm},mybnode/.style={draw, circle,inner sep=0cm, minimum size=0.6cm}]
\draw[rounded corners=5mm] (0,0) rectangle (3,3); 

\draw[rounded corners=4mm] (3,4) rectangle (5.5,6.5); 
\draw[rounded corners=3mm] (3.5,1.5) rectangle (5.5,3.5); 

\draw[rounded corners=2mm] (6,0.5) rectangle (7.5,2); 
\node [mynode] (n1) at (2,2) {};
\node [mynode] (n2) at (4,5) {};
\node [mynode] (n3) at (4,2.5) {};
\node [mynode] (n4) at (5,2.5) {};
\node [mynode] (n5) at (5,5) {};
\node [mynode] (n6) at (6.75,1.5) {};
\node [mynode] (n7) at (6.75,1) {};
\node [mynode] (n8) at (2,1) {};
\draw [blue] (n1) -- (n3) -- (n2);
\draw (n2) -- (n5);
\draw [blue] (n5) -- (n4) -- (n6);
\draw (n6) -- (n7);
\draw [blue] (n7) -- (n8);
\draw (n8) -- (n1);
\node at (6.75,2.5) {$A_1$};
\node at (4.5,1.25) {$A_4$};
\node at (2.5,5) {$A_5$};
\node at (1.5,3.5) {$A_7$};

\begin{scope}[xshift=7.5cm]
\node [mybnode] (m1) at (6.75,1.25) {1};
\node [mybnode] (m2) at (4.5,2.5) {4};
\node [mybnode] (m3) at (4.25,5.75) {5};
\node [mybnode] (m4) at (1.5,1.5) {7};
\draw [blue] (m4) to (m2);
\draw [blue,bend left=15] (m3) to (m2);
\draw [blue,bend left=15] (m2) to (m3);
\draw [blue] (m2) to (m1) to (m4);
\node at (3.75,-1) {$T_1$};
\end{scope}
\end{tikzpicture}
\end{center}
\caption{A simple cycle $F_i$ in $G$ (depicted in blue and black) that induces the blue circuit $C_i$ in $G_\ALG$. Notice that $C_i$ is not simple.
The rounded rectangles represent connected components of $\ALG$, and their size indicates their width.
\label{ex:cyclesinG-GALG}
}
\end{figure}
\begin{figure}
\begin{center}
\begin{tikzpicture}[scale=0.6,thick, mynode/.style={draw, circle,inner sep=0cm, minimum size=0.1cm}]
\draw[rounded corners=5mm] (0,0) rectangle (3,3); 

\draw[rounded corners=4mm] (3,4) rectangle (5.5,6.5); 
\draw[rounded corners=3mm] (3.5,1.5) rectangle (5.5,3.5); 

\draw[rounded corners=2mm] (6,0.5) rectangle (7.5,2); 
\node [mynode] (n1) at (2,2) {};
\node [mynode] (n2) at (4,5) {};
\node [mynode] (n3) at (4,2.5) {};
\node [mynode] (n4) at (5,2.5) {};
\node [mynode] (n5) at (5,5) {};
\node [mynode] (n6) at (6.75,1.5) {};
\node [mynode] (n7) at (6.75,1) {};
\node [mynode] (n8) at (2,1) {};
\draw (n2) -- (n3);
\node at (6.75,2.5) {$A_1$};
\node at (4.5,1) {$A_4$};
\node at (2.5,5) {$A_5$};
\node at (1.5,3.5) {$A_7$};
\node at (3.75,-1) {$T_1$};
\end{tikzpicture}\qquad\quad
\begin{tikzpicture}[scale=0.6,thick, mynode/.style={draw, circle,inner sep=0cm, minimum size=0.1cm}]
\draw[rounded corners=5mm] (0,0) rectangle (3,3); 

\draw[rounded corners=4mm] (3,4) rectangle (5.5,6.5); 
\draw[rounded corners=3mm] (3.5,1.5) rectangle (5.5,3.5); 

\draw[rounded corners=2mm] (6,0.5) rectangle (7.5,2); 
\node [mynode] (n1) at (2,2) {};
\node [mynode] (n2) at (4,5) {};
\node [mynode] (n3) at (4,2.5) {};
\node [mynode] (n4) at (5,2.5) {};
\node [mynode] (n5) at (5,5) {};
\node [mynode] (n6) at (6.75,1.5) {};
\node [mynode] (n7) at (6.75,1) {};
\node [mynode] (n8) at (2,1) {};
\draw (n5) -- (n4) -- (n6);
\node at (7,2.5) {$A_1$};
\node at (6,2.75) {$A_4$};
\node at (2.5,5) {$A_5$};
\node at (1.5,3.5) {$A_7$};
\node at (3.75,-1) {$T_2$};
\end{tikzpicture}\qquad\quad
\begin{tikzpicture}[scale=0.55,thick, mynode/.style={draw, circle,inner sep=0cm, minimum size=0.1cm}]
\draw[rounded corners=5mm] (0,0) rectangle (3,3); 

\draw[rounded corners=4mm] (3,4) rectangle (5.5,6.5); 
\draw[rounded corners=3mm] (3.5,1.5) rectangle (5.5,3.5); 

\draw[rounded corners=2mm] (6,0.5) rectangle (7.5,2); 
\node [mynode] (n1) at (2,2) {};
\node [mynode] (n2) at (4,5) {};
\node [mynode] (n3) at (4,2.5) {};
\node [mynode] (n4) at (5,2.5) {};
\node [mynode] (n5) at (5,5) {};
\node [mynode] (n6) at (6.75,1.5) {};
\node [mynode] (n7) at (6.75,1) {};
\node [mynode] (n8) at (2,1) {};
\draw (n7) -- (n8);
\draw (n1) -- (n3);
\node at (7,2.5) {$A_1$};
\node at (6,2.75) {$A_4$};
\node at (2.5,5) {$A_5$};
\node at (1.5,3.5) {$A_7$};
\node at (3.75,-1) {$T_3$};
\end{tikzpicture}
\end{center}
\caption{A partitioning of the blue circuit in Figure~\ref{ex:cyclesinG-GALG} into trees $T_1,T_2,T_3$ in $G_\ALG$.
If none of the three induced connecting moves is improving, then $d(T_1) \ge w(A_4)$, $d(T_2) \ge w(A_4)+w(A_1)$ and $d(T_3) \ge w(A_4)+w(A_1)$. Thus, we get $d(C_i) \ge 3 w(A_4)+2w(A_1)$.\label{fig:c:baddecomp}}
\end{figure}
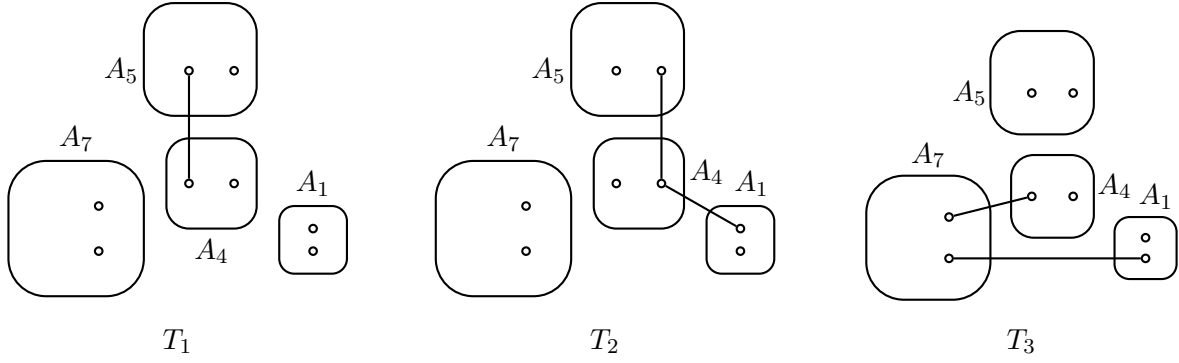

For each $i\in\{1,\dots,q\}$, the simple cycle $F_i$ in $G$ induces a circuit $C_i$ in $G_\ALG$. Figure~\ref{ex:cyclesinG-GALG} shows a cycle $F_i$ and its induced circuit $C_i$. 
The edges of $C_i$ correspond to those edges of $F_i$ that we want to replace. 
Observe that $C_i$ is indeed not necessarily simple: Whenever $F_i$ revisits the connected component $A_j$ of $\ALG$, the induced circuit $C_i$ revisits the same node $j \in V_\ALG$. Assume $C_i$ visits exactly $s$ distinct vertices. Then we name them $\xi_1,\ldots,\xi_s$ and assume without loss of generality that $\xi_1 > \dots > \xi_s$. 
Since the connected components of $\ALG$ are numbered according to their width, we know that $w(A_{\xi_1}) \geq \dots \geq w(A_{\xi_s})$, and thus the $\xi_i$ are ordered according to the widths of the components as well.
Finally, let $n_\ell$ be the number of times that $C_i$ visits $\xi_\ell$. In Figure~\ref{ex:cyclesinG-GALG}, $\xi_1=7, \xi_2=5,\xi_3=4,\xi_4=1$, and $n_1=1$, $n_2=1$, $n_3=2$ and $n_4=1$. 

The crucial idea for the replacement of $C_i$ is to use the connecting move optimality\footnote{We will later use approximate moves, but for simplicity, we forget about approximate optimality during this explanation.} of $\ALG$ to lower bound $d(C_i)$. 
Any subgraph of $C_i$ that is a tree in $G_\ALG$ induces a connecting move. 
For an example, consider Figure~\ref{fig:c:baddecomp}. 
We partitioned the edges in $C_i$ from Figure~\ref{ex:cyclesinG-GALG} into three trees. 
For any $T_i$ among the three trees, connecting move optimality guarantees that the sum of the edges $d(T_i)$ is at least as expensive as the sum of the widths of the components that get connected, except for the largest. 
For example, when adding $T_1$ to the solution, the edge cost increases by $c(T_1)$, but the width cost decreases by $w(A_4)$. Thus, $d(T_1) \ge w(A_4)$ when $\ALG$ is connecting move optimal. 
Since $T_1, T_2$ and $T_3$ are a edge-disjoint partitioning of $C_i$, it holds that $d(C_i) = d(T_1)+d(T_2)+d(T_3)$ and thus we get $d(C_i) \ge 3w(A_4) +2w(A_1)$ by considering all three connecting moves.

Now consider Figure~\ref{fig:c:gooddecomp}. Here, we partitioned the edges of $C_i$ into a different set of trees. 
It turns out that this partitioning provides a better lower bound on $d(C_i)$, namely $w(A_5)+2w(A_4)+2w(A_1)$. 
In fact, this lower bound contains $w(A_{\xi_\ell})$ at least $n_{\ell}$ times for all $\ell \in\{2,3,4\}$.
We observe a sufficient condition for guaranteeing that such a partitioning exists.

\begin{restatable}{definition}{deftreepaysforxi}
We say that a tree \emph{pays} for $\xi_\ell$ (once) if it contains $\xi_\ell$ and at least one vertex $\xi_{\ell'}$ with $\xi_{\ell'} > \xi_{\ell}$.
\end{restatable}
\begin{restatable}{definition}{definitionminguarded}\label{def:minguarded}
  Let $C=(e_1,\dots,e_{|C|})$ be a circuit in~$G_\ALG$ that visits the nodes $v_1,\dots,v_{|C|+1} = v_1$ in this order.
  We say that~$C$ is \emph{guarded} if we have $v_i < v_1$ for all $i\in\{2,\dots,|C|\}$. 
  A circuit~$C$ is \emph{minimally guarded} if it is guarded and no subcircuit $(v_{i_1},\ldots,v_{i_2})$ with $i_1, i_2 \in \{2,\dots,|C|\}$, $i_1 < i_2$ and $v_{i_1} = v_{i_2}$ is guarded.
\end{restatable}
Notice that in any guarded circuit, the highest component number only appears once. 
In Figure~\ref{ex:cyclesinG-GALG}, $C_i$ is minimally guarded because the only component visited between the two visits of $A_4$ is $A_5$, which has a higher index.

\begin{restatable}{lemma}{lemtreedecomptolowerbound}\label{lem:treedecomp-to-lowerbound}
Let $C=(e_1,\dots,e_{|C|})$ be a guarded circuit in $G_\ALG$ that visits the nodes $v_1,\dots,v_{|C|+1}=v_1$ in this order. 
Assume that $v_1=v_{|C|+1} \ge v_i$ for all $j\in\{2,\ldots,|C|\}$ and that $\{v_1,\ldots,v_{|C|}\}$ consists of $s$ disjoint elements $\xi_1 > \xi_2 > \ldots > \xi_s$ (this means that $v_1=\xi_1$).
Furthermore, let $n_\ell$ be the number of times that $C$ visits node $\xi_\ell$, for all $\ell=1,\dots,s$.
  If $\ALG$ is $c$-approximate connecting move optimal and there exists a set of trees  $\moveSet$ in $G_\ALG$ that satisfies that
\begin{enumerate}
\item all trees in $\moveSet$ are edge-disjoint and only contain edges from $C$ and
\item for all $\ell \in \{2,\ldots,s\}$, there are at least $n_\ell$ trees in $\moveSet$ that pay for $\xi_\ell$,
\end{enumerate}
 then it holds that
 \begin{equation*}
    \sum_{i=2}^{|C|} w(A_{v_i})=\sum_{\ell=2}^s n_\ell w(A_{\xi_\ell}) \leq c \cdot \sum_{i=1}^{|C|} d_{e_i} = c \cdot d(C).
  \end{equation*}
  Recall that $A_{v_i}$ is the connected component of $\ALG$ that corresponds to the index $v_i$.
\end{restatable}
\begin{proof}
By the first precondition we know that
  \begin{equation*}
	\sum_{T \in \moveSet} \srsum{e \in T} d_e \le \sum_{i=1}^{|C|} d_{e_i}.
  \end{equation*}
	Now notice that every tree in $G_\ALG$ and thus every tree in $\moveSet$ defines a connecting move.
	Since $\ALG$ is $c$-approximate connecting move optimal, it holds that
  \begin{equation*}
	\srsum{v \in V[T]} w(A_v) -\max_{v\in V[T]} w(A_v) \le c \cdot \slrsum{e \in T} d_e
  \end{equation*}
	for every tree $T$ in $\moveSet$. 
  Let $low(T) = V[T] \setminus\{\max_{\xi_i \in T} \xi_i\}$.
  Then, we have
  \begin{align*}
        \sum_{i=2}^{|C|} w(A_{v_i})
      &= \sum_{\ell=2}^s n_\ell w(A_{\xi_\ell})\\
      &\stackrel{2.}{\leq} \sum_{\ell=2}^s \sum_{T \in \moveSet} \indicator_{low(T)}(\xi_\ell) w(A_{\xi_\ell})\\
      &= \sum_{T \in \moveSet} \sum_{\ell=2}^s \indicator_{low(T)}(\xi_\ell) w(A_{\xi_\ell})\\
      &= \sum_{T \in \moveSet} \srsum{v \in low(T)} w(A_{v})\\
      &= \sum_{T \in \moveSet} \Bigl(\srsum{v \in V[T]} w(A_v) - \max_{v \in V[T]} w(A_v)\Bigr)\\
      &\leq c \cdot \sum_{T \in \moveSet} \ \ \slrsum{e \in T} d_e
   \end{align*}
  and this proves the lemma. \end{proof}


\begin{figure}
\begin{center}
\begin{tikzpicture}[scale=0.6,thick, mynode/.style={draw, circle,inner sep=0cm, minimum size=0.1cm}]
\draw[rounded corners=5mm] (0,0) rectangle (3,3); 

\draw[rounded corners=4mm] (3,4) rectangle (5.5,6.5); 
\draw[rounded corners=3mm] (3.5,1.5) rectangle (5.5,3.5); 

\draw[rounded corners=2mm] (6,0.5) rectangle (7.5,2); 
\node [mynode] (n1) at (2,2) {};
\node [mynode] (n2) at (4,5) {};
\node [mynode] (n3) at (4,2.5) {};
\node [mynode] (n4) at (5,2.5) {};
\node [mynode] (n5) at (5,5) {};
\node [mynode] (n6) at (6.75,1.5) {};
\node [mynode] (n7) at (6.75,1) {};
\node [mynode] (n8) at (2,1) {};
\draw (n1) -- (n3) -- (n2);
\node at (6.75,2.5) {$A_1$};
\node at (4.5,1) {$A_4$};
\node at (2.5,5) {$A_5$};
\node at (1.5,3.5) {$A_7$};
\node at (3.75,-1) {$T_1$};
\end{tikzpicture}\qquad\quad
\begin{tikzpicture}[scale=0.6,thick, mynode/.style={draw, circle,inner sep=0cm, minimum size=0.1cm}]
\draw[rounded corners=5mm] (0,0) rectangle (3,3); 

\draw[rounded corners=4mm] (3,4) rectangle (5.5,6.5); 
\draw[rounded corners=3mm] (3.5,1.5) rectangle (5.5,3.5); 

\draw[rounded corners=2mm] (6,0.5) rectangle (7.5,2); 
\node [mynode] (n1) at (2,2) {};
\node [mynode] (n2) at (4,5) {};
\node [mynode] (n3) at (4,2.5) {};
\node [mynode] (n4) at (5,2.5) {};
\node [mynode] (n5) at (5,5) {};
\node [mynode] (n6) at (6.75,1.5) {};
\node [mynode] (n7) at (6.75,1) {};
\node [mynode] (n8) at (2,1) {};
\draw (n4) -- (n5);
\draw (n4) -- (n6);
\node at (7,2.5) {$A_1$};
\node at (3.25,3.5) {$A_4$};
\node at (2.5,5) {$A_5$};
\node at (1.5,3.5) {$A_7$};
\node at (3.75,-1) {$T_2$};
\end{tikzpicture}\qquad\quad
\begin{tikzpicture}[scale=0.55,thick, mynode/.style={draw, circle,inner sep=0cm, minimum size=0.1cm}]
\draw[rounded corners=5mm] (0,0) rectangle (3,3); 

\draw[rounded corners=4mm] (3,4) rectangle (5.5,6.5); 
\draw[rounded corners=3mm] (3.5,1.5) rectangle (5.5,3.5); 

\draw[rounded corners=2mm] (6,0.5) rectangle (7.5,2); 
\node [mynode] (n1) at (2,2) {};
\node [mynode] (n2) at (4,5) {};
\node [mynode] (n3) at (4,2.5) {};
\node [mynode] (n4) at (5,2.5) {};
\node [mynode] (n5) at (5,5) {};
\node [mynode] (n6) at (6.75,1.5) {};
\node [mynode] (n7) at (6.75,1) {};
\node [mynode] (n8) at (2,1) {};
\draw (n7) -- (n8);
\node at (7,2.5) {$A_1$};
\node at (6,2.75) {$A_4$};
\node at (2.5,5) {$A_5$};
\node at (1.5,3.5) {$A_7$};
\node at (3.75,-1) {$T_3$};
\end{tikzpicture}
\end{center}
\caption{A different partitioning of the blue circuit in~\ref{ex:cyclesinG-GALG} into trees $T_1,T_2,T_3$ in $G_\ALG$. If none of the three induced connecting moves is improving, then $d(T_1) \ge w(A_5)+w(A_4)$, $d(T_2) \ge w(A_4)+w(A_1)$ and $d(T_3) \ge w(A_1)$. Thus, we get $d(C_i) \ge w(A_5)+2w(A_4)+2w(A_1) \ge n_2 w(A_5) + n_3 w(A_4) +n_4 w(A_1)$.\label{fig:c:gooddecomp}}
\end{figure}
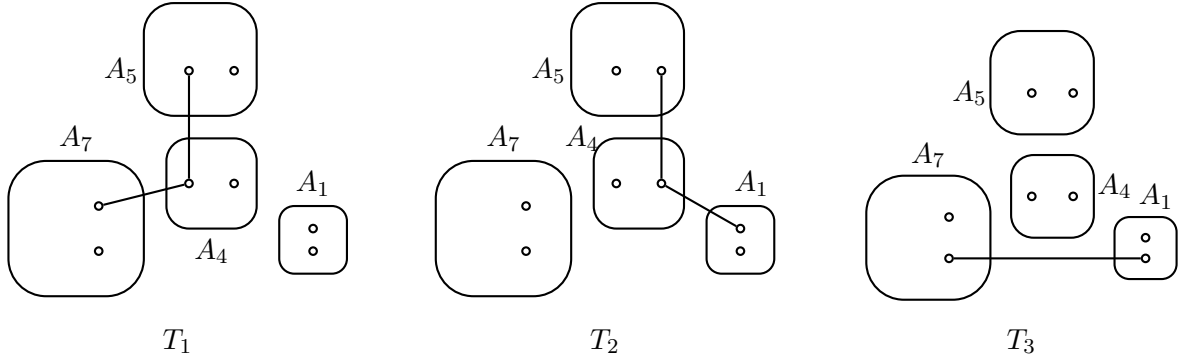

\begin{algorithm}[p]
    \SetKwInOut{Input}{input}
    \SetKwInOut{Output}{output}
    \SetKw{KwTo}{In}
    \SetKw{Set}{set}
    \SetKw{Let}{let}
    \SetKw{Break}{break}
    \SetKw{Forall}{for all}
    \SetKw{With}{with}
    \SetKw{Select}{select}
    \SetCommentSty{mycommfont}    
    
    \DontPrintSemicolon
    \SetNlSkip{2em}
    \SetAlgoSkip{bigskip}
    \SetKwComment{tcc}{}{}
    \SetKwComment{tcp}{}{}
    
    \LinesNumbered
    \Input{A minimally guarded circuit $C=(v_1,\dots,v_{|C|+1})$ in $G_\ALG$ with $v_1=v_{|C|+1}$.\\
		Let $\{\xi_1,\dots,\xi_s\}$ be the set of disjoint vertices on $C$, w.l.o.g. $\xi_1 > \dots > \xi_s$.
		}
    \Output{A set $\moveSet$ of edge disjoint trees in ${G}_\ALG$ consisting of edges from $C$
    }
    \BlankLine
    \tcc{%
    Initialization:
    Observe that $\xi_2$ can only occur once on $C$.
    If $v_{i_1},v_{i_2} \in C$ with $i_1\not=i_2$, but $v_{i_1}=v_{i_2}=\xi_2$, then the subcircuit $(v_{i_1},\dots,v_{i_2})$ certifies that $C$ is not minimally guarded. }
    \Let $v$ be the unique node in $C$ with $v = \xi_2$.\;
    \Let $T=\{ \{v_1,v\} \}$ and \Let $\moveSet_2 = \{(T,v_1)\}$\tcc*[f]{$T$ is stored with root $v_1$}\;
    \Let $\mathfrak{P}_2 = \{ (v_1,\dots, v), (v,\dots,v_{|C|-1})\}$.\tcc*[f]{the second part of $C$ is unclaimed}\;
      \tcc{Main loop: Iteration $k$ computes $\moveSet_k$ and $\mathfrak{P}_k$}
    \ForEach{$k=3,\dots,s$}{
      \Let $\moveSet_k = \moveSet_{k-1}$ and \Let $\mathfrak{P}_k = \mathfrak{P}_{k-1}$\;
			\tcc{We need to process all occurrences of $\xi_k$ on $C$, so we store their indices in $I$}			
      \Let $I = \cset{j \in \{2,\dots,|C|-1\}}{v_j = \xi_{k}}$.\;
			\tcc{Find the path $P_j$ that $j$ lies on. We assume that $v_{P_j}$ occurs before $w_{P_j}$ on $C$}
      \Let $P_j=(v_{P_j},\dots,w_{P_j})$ be the path in $\mathfrak{P}_{k-1}$ with $j$ as inner node, for all $j\in I$.\;\label{alg:defpj}
			\tcc{First case: Treats all occurences of $\xi_k$ in unclaimed parts of $C$ by creating new trees.}
      \ForEach{$j \in I$ with $\pi_{k-1}(P_j)=\bot$\label{alg:jinfreepartcase}}{
        \Let $T = \{\{v_{P_j}, v_j\}\}$ and \Let $\moveSet_k = \moveSet_k \cup \{(T,v_{P_j})\}$\tcc*[f]{new tree claims left part}\;\label{alg:newtree:treedef}
        \Let $\mathfrak{P}_k = \mathfrak{P}_k \setminus \{P_j\} \cup \{ (v_{P_j},\dots,v_j), (v_j,\dots, w_{P_j}) \}$\tcc*[f]{path $P_j$ is split at $v_j$}\;\label{alg:newtree:pathsplit}
      }
			\tcc{Second case: Treats all occurrences of $\xi_k$ that fall on edges of trees in $\moveSet_k$. This is done by iterating through all trees and processing all occurrences in the same tree together.}
      \ForEach{$(T,r) \in \moveSet_{k-1}$\label{alg:jonTcase}}{
        \Let $I_T = \cset{j \in I}{\pi_{k-1}(P_j) \in T}$\; 
				\lIf(\tcc*[f]{If $I_T$ is empty, then $T$ remains unchanged}){$I_T = \emptyset$ }{continue}
        \Select $j^\ast \in I_T$ such that the path from $j^\ast$ to $r$ in $T$ contains no $j \in I_T\setminus\{j^\ast\}$\;
				\tcc{$T$ is modified to include $j^\ast$. Notice that $\pi_{k-1}(P_{j^\ast})=\{v_{P_j^\ast},w_{P_j^\ast}\}$.}
        \Let $T=T \setminus\{\{v_{P_j^\ast},w_{P_j^\ast}\}\} \cup \{\{v_{P_j^\ast},v_{j^\ast}\}, \{v_{j^\ast},w_{P_j^\ast}\}\}$\;\label{alg:adapttree:jast}
        \Let $\mathfrak{P}_k = \mathfrak{P}_k \setminus\{P_{j^\ast}\} \cup \{(v_{P_j^\ast},\dots,v_{j^\ast}), (v_{j^\ast},\dots,w_{P_j^\ast})\}$\;\label{alg:adapttree:pathsplitjast}
        \ForEach{$j \in I_T\backslash\{j^\ast\}$}{
					\tcc{Any edge containing a $j\neq j^\ast$ is split: Half of the edge becomes a new tree, and the other half is used to keep the tree connected. Notice that $\pi_{k-1}(P_{j})=\{v_{P_j},w_{P_j}\}$.}
            \Let $T=T\setminus\{\{v_{P_j},w_{P_j}\}\} \cup \{\{v_j,w_{P_j}\}\}$\;\label{alg:adapttree:splitedge}
						\Let $T'=\{\{v_{P_j},v_j\}\}$ and \Let $\moveSet_k = \moveSet_k \cup \{(T',v_{P_j})\}$\;\label{alg:adapttree:newtree}
          \Let $\mathfrak{P}_k = \mathfrak{P}_k \setminus\{P_j\} \cup \{(v_{P_j},\dots,v_j), (v_j,\dots,w_{P_j})\}$\;\label{alg:adapttree:pathsplitj}
        }
      }
    }
    \Return $\moveSet_{s}$\;
    \caption{A charging algorithm~\label{alg:charging}}
  \end{algorithm}
\begin{figure}[p]
\begin{center}
\begin{tikzpicture}[thick, scale=0.7, treenode/.style={draw, circle,inner sep=0cm, minimum size=0.4cm},streenode/.style={draw, circle,inner sep=0cm, minimum size=0.2cm}]
\foreach \x/\xi in {1/7,2/1,3/2,4/1,5/4,6/1,7/2,8/5,9/1,10/3,11/2,12/7}{
 \draw (\x,-0.1) -- (\x,0.1);
 \node at (\x,0.5) {\xi};
 \node at (\x,1.2) {$v_{\x}$};
}
\draw (1,0) -- (12,0);
\node at (16,1) {Trees in $\hat G_\ALG$};

\node at (0,-1) {$\moveSet_2$};
\node [streenode,fill=black!20] (nt11) at (1,-1) {}; \node [streenode] (nt12) at (8,-1) {};
\draw [-] (nt11) -- (nt12);
\node at (0,-2) {$\mathfrak{P}_2$};
\draw [>=|,<->] (1,-2) -- (7.95,-2); \draw [>=|,<->] (8.05,-2) -- (11,-2);
\node [treenode,fill=black!20] (n11) at (15,-1) {{\small $7$}};
\node [treenode] (n12) at (16,-1) {{\small $5$}};
\draw [-] (n11)--(n12);

\begin{scope}[yshift=-2.5cm]
\node at (0,-1) {$\moveSet_3$};
\node [streenode,fill=black!20] (nt21) at (1,-1) {}; \node [streenode] (nt22) at (5,-1) {}; \node [streenode] (nt23) at (8,-1) {};
\draw [-] (nt21) -- (nt22); \draw [-] (nt22) -- (nt23);
\node at (0,-2) {$\mathfrak{P}_3$};
\draw [>=|,<->] (1,-2) -- (4.95,-2); \draw [>=|,<->] (5.05,-2) -- (7.95,-2); \draw [>=|,<->] (8.05,-2) -- (11,-2);
\node [treenode,fill=black!20] (n21) at (15,-1) {{\small $7$}};
\node [treenode] (n22) at (16,-1) {{\small $4$}};
\node [treenode] (n23) at (17,-1) {{\small $5$}};
\draw [-] (n21)  --(n22); \draw [-] (n22) --(n23);
\end{scope}

\begin{scope}[yshift=-5cm]
\node at (0,-1) {$\moveSet_4$};
\node [streenode,fill=black!20] (nt31) at (1,-1) {}; \node [streenode] (nt32) at (5,-1) {}; \node [streenode] (nt33) at (8,-1) {};
\node [streenode,fill=black!20] (nt34) at (8,-1.5) {}; \node [streenode] (nt35) at (10,-1.5) {};
\draw [-] (nt31) -- (nt32); \draw [-] (nt32) -- (nt33); \draw [-] (nt34) -- (nt35);
\node at (0,-2.5) {$\mathfrak{P}_4$};
\draw [>=|,<->] (1,-2.5) -- (4.95,-2.5); \draw [>=|,<->] (5.05,-2.5) -- (7.95,-2.5); \draw [>=|,<->] (8.05,-2.5) -- (9.95,-2.5); \draw [>=|,<->] (10.05,-2.5)-- (11,-2.5);
\node [treenode,fill=black!20] (n31) at (15,-1) {{\small $7$}};
\node [treenode] (n32) at (16,-1) {{\small $4$}};
\node [treenode] (n33) at (17,-1) {{\small $5$}};
\node [treenode,fill=black!20] (n34) at (18,-1.5) {{\small $5$}};
\node [treenode] (n35) at (19,-1.5) {{\small $3$}};
\draw [-] (n31)--(n32); \draw [-] (n32) --(n33); \draw [-] (n34) -- (n35);
\end{scope}

\begin{scope}[yshift=-8cm]
\node at (0,-1) {$\moveSet_5$};
\node [streenode,fill=black!20] (nt41) at (1,-1) {}; \node [streenode] (nt41b) at (3,-1) {}; 
\node [streenode] (nt42) at (5,-1) {}; \node [streenode] (nt42b) at (7,-1) {};  \node [streenode] (nt43) at (8,-1) {};
\node [streenode,fill=black!20] (nt44) at (8,-1.5) {}; \node [streenode] (nt45) at (10,-1.5) {};
\node [streenode,fill=black!20] (nt46) at (10,-2) {}; \node [streenode] (nt47) at (11,-2) {};
\node [streenode,fill=black!20] (nt48) at (5,-2.5) {}; \node [streenode] (nt49) at (7,-2.5) {}; 
\draw [-] (nt41) -- (nt41b); \draw [-] (nt41b) -- (nt42); \draw [-] (nt42b) -- (nt43); \draw [-] (nt44) -- (nt45); \draw [-] (nt46)--(nt47);
\draw [-] (nt48) -- (nt49); 
\node at (0,-3.5) {$\mathfrak{P}_5$};
\draw [>=|,<->] (1,-3.5) -- (2.95,-3.5); \draw [>=|,<->] (3.05,-3.5) -- (4.95,-3.5); \draw [>=|,<->] (5.05,-3.5) -- (7.95,-3.5); 
\draw [>=|,<->] (8.05,-3.5) -- (9.95,-3.5); \draw [>=|,<->] (10.05,-3.5) -- (10.95,-3.5); 
\node [treenode,fill=black!20] (n41) at (15,-1) {{\small $7$}};
\node [treenode] (n42) at (16,-1) {{\small $2$}};
\node [treenode] (n43) at (17,-1) {{\small $4$}};
\node [treenode] (n43b) at (16,-1.75) {{\small $5$}};
\node [treenode,fill=black!20] (n44) at (18,-1.5) {{\small $5$}};
\node [treenode] (n45) at (19,-1.5) {{\small $3$}};
\node [treenode,fill=black!20] (n46) at (20,-2) {{\small $3$}};
\node [treenode] (n47) at (21,-2) {{\small $2$}};
\node [treenode,fill=black!20] (n48) at (15,-2.5) {{\small $4$}};
\node [treenode] (n49) at (16,-2.5) {{\small $2$}};
\draw [-] (n41)--(n42); \draw [-] (n42) --(n43); \draw [-] (n42) --(n43b); 
\draw [-] (n44) -- (n45); \draw [-](n46)--(n47); \draw [-] (n48) -- (n49);
\end{scope}

\begin{scope}[yshift=-12cm]
\node at (0,-1) {$\moveSet_6$};
\node [streenode,fill=black!20] (nt51) at (1,-1) {}; \node [streenode] (nt51a) at (2,-1) {};  \node [streenode] (nt51b) at (3,-1) {}; 
\node [streenode] (nt51c) at (4,-1) {}; \node [streenode] (nt52) at (5,-1) {};  \node [streenode] (nt52b) at (7,-1) {};  \node [streenode] (nt53) at (8,-1) {};
\node [streenode,fill=black!20] (nt54) at (8,-1.5) {}; \node [streenode] (nt54b) at (9,-1.5) {}; \node [streenode] (nt55) at (10,-1.5) {};
\node [streenode,fill=black!20] (nt56) at (10,-2) {}; \node [streenode] (nt57) at (11,-2) {};
\node [streenode,fill=black!20] (nt58) at (5,-2.5) {}; \node [streenode] (nt58b) at (6,-2.5) {};  \node [streenode] (nt59) at (7,-2.5) {}; 
\node [streenode,fill=black!20] (nt59a) at (3,-3) {};  \node [streenode] (nt59b) at (4,-3) {};
\draw [-] (nt51) -- (nt51a); \draw [-] (nt51a) -- (nt51b); 
\draw [-] (nt51c) -- (nt52); \draw [-] (nt52b) -- (nt53); \draw [-] (nt54) -- (nt54b); 
\draw [-] (nt54b) -- (nt55); \draw [-] (nt56)--(nt57);
\draw [-] (nt58) -- (nt58b);  \draw [-] (nt58b) -- (nt59); 
\draw [-] (nt59a) -- (nt59b);
\node at (0,-4) {$\mathfrak{P}_6$};
\draw [>=|,<->] (1,-4) -- (1.95,-4); \draw [>=|,<->] (2.05,-4) -- (2.95,-4); \draw [>=|,<->] (3.05,-4) -- (3.95,-4); 
\draw [>=|,<->] (4.05,-4) -- (4.95,-4); \draw [>=|,<->] (5.05,-4) -- (5.95,-4); \draw [>=|,<->] (6.05,-4) -- (6.95,-4); 
\draw [>=|,<->] (7.05,-4) -- (7.95,-4); \draw [>=|,<->] (8.05,-4) -- (8.95,-4); \draw [>=|,<->] (9.05,-4) -- (9.95,-4); 
\draw [>=|,<->] (10.05,-4) -- (10.95,-4); 
\node [treenode,fill=black!20] (n50) at (15,-1) {{\small $7$}};
\node [treenode] (n51) at (16,-1) {{\small $1$}};
\node [treenode] (n52) at (17,-1) {{\small $2$}};
\node [treenode] (n53) at (16,-1.75) {{\small $4$}};
\node [treenode] (n53b) at (17,-1.75) {{\small $5$}};
\node [treenode,fill=black!20] (n54) at (18,-1.3) {{\small $5$}};
\node [treenode] (n55) at (19,-1.3) {{\small $1$}};
\node [treenode] (n55b) at (20,-1.3) {{\small $3$}};
\node [treenode,fill=black!20] (n56) at (20,-2) {{\small $3$}};
\node [treenode] (n57) at (21,-2) {{\small $2$}};
\node [treenode,fill=black!20] (n58) at (15,-2.5) {{\small $4$}};
\node [treenode] (n58b) at (16,-2.5) {{\small $1$}};
\node [treenode] (n59) at (17,-2.5) {{\small $2$}};
\node [treenode,fill=black!20] (n59b) at (18,-3) {{\small $2$}};
\node [treenode] (n59c) at (19,-3) {{\small $1$}};
\draw [-] (n50)--(n51); \draw [-] (n51)--(n52); \draw [-] (n51) --(n53); \draw [-] (n52) --(n53b); 
\draw [-] (n54) -- (n55); \draw [-] (n55) -- (n55b); \draw [-](n56)--(n57); \draw [-] (n58) -- (n58b); \draw [-] (n58b) -- (n59);
\draw [-] (n59b) -- (n59c);
\end{scope}

\end{tikzpicture}
\end{center}
\caption{An example for Algorithm~\ref{alg:charging}. On top, we see a circuit $v_1,\ldots,v_s$ drawn in a path form with the only occurrences of $\xi_1=7$ at the endpoints. Below that, we see how $\moveSet_k$ and $\mathfrak{P}_k$ develop through the iterations $k=2,\ldots,6$, after which the algorithm stops. Iteration $k=5$ is the first where two occurrences of $\xi_k$ fall into the same tree, which changes the structure of the tree, because the edge between $v_5$ and $v_8$ is split and distributed between two trees: the edge $v_7$ and $v_8$ stays in the tree, and the edge between $v_5$ and $v_7$ forms a new tree. Notice that connectivity is maintained by this operation.\label{fig:charging-algo-example}}
\end{figure}

\subsubsection{The Partitioning Algorithm}\label{sec:partitioningalgo}

We show how to partition minimally guarded circuits by providing Algorithm~\ref{alg:charging}. 
It computes a sequence of sets $\moveSet_k$ of trees. 
For $k\in \{2,\ldots,s\}$, $\moveSet_k$ contains a partitioning with $n_i$ tree that each pay for $\xi_i$ once, for all $i \in \{1,\ldots,k\}$. 
The output of the algorithm is $\moveSet_s$.

The algorithm maintains a partitioning of $C$ into a set of sub-paths $\mathfrak{P}_k$. While these paths are not necessarily simple, they are at all times edge-disjoint. The algorithm iteratively splits non-simple subpaths into simple subpaths. At the same time, it needs to make sure that the subpaths can be combined to trees that satisfy the conditions of Lemma~\ref{lem:treedecomp-to-lowerbound}. This is accomplished by building the trees of $\moveSet_k$ in the transitive closure $\hat{G}_\ALG$ of $G_\ALG$: In this way, we can ensure that any edge of each tree in $\moveSet_k$ corresponds to a path in the current partitioning $\mathfrak{P}_k$. More precisely, if a tree in $\moveSet_k$ contains an edge $(v,w)$, then the partitioning contains a subpath $(v,\dots,w)$. This is why we say that a tree $T \in \moveSet_k$ \emph{claims} a sub-path $p$ of $C$ if one of the edges in $T$ corresponds to $p$. Each time the algorithm splits a sub-path, it also splits the corresponding edge of a tree in $\moveSet_k$. To represent the correspondence of trees and subpaths, we define the mapping $\pi_k : \mathfrak{P}_k \to \cup_{T \in \moveSet_k} T \times \{\bot\}$ that maps a path $p=(v,\dots,w) \in \mathfrak{P}_k$ to an edge $e \in \cup_{T \in \moveSet_k} T$ if and only if $e \cap p = \{v,w\}$. If no such edge exists, then $\pi_k(P) = \bot$. This mapping is well-defined. The trees in the final set $\moveSet_s$ do not contain transitive edges, i.e., they are subgraphs of $G_\ALG$. They also leave no part of $C$ unclaimed.


We now describe the algorithm in more detail and simultaneously observe its main property:
\begin{invariant}\label{obs:mainalginvariant}
For all $k=2,\ldots,s$, it holds after iteration $k$ that for all $i\in \{2,\ldots,k\}$ there are at least $n_i$ trees in $\moveSet_k$ that pay for $\xi_i$.
\end{invariant}

For presentation purposes, we assume that we already know that the following invariants are true and prove them later in Lemma~\ref{lem:invariantsproof}.
An example run of the algorithm to accompany the explanation can be found in Figure~\ref{fig:charging-algo-example}. Notice that the trees in $\moveSet_k$ are rooted, \ie we store each tree as a tuple consisting of the actual tree plus a root. The trees in connecting moves are unrooted, the roots in $\moveSet_k$ are only needed for the computation.

\begin{invariants}\label{obs:easy-invariants}
For all $k=2,\ldots,s$, the following holds:
\begin{enumerate}[label=\arabic*.,ref={\ref{obs:easy-invariants}-\arabic*}]
  \item The trees in $\moveSet_k$ are edge disjoint. \label{alg-easyproperty-1}
  \item The paths in $\mathfrak{P}_k$ are edge-disjoint and it holds that $\bigcup_{p\in \mathfrak{P}_k} p = C \backslash \{\{v_{|C|},v_{|C|+1}\}\}$. \label{alg-easyproperty-2}
  \item If $v$ is an outer node of some $p \in \mathfrak{P}_k$, then $v \in \{\xi_1,\dots,\xi_k\}$. If $v$ is an inner node, then $v \in \{\xi_{k+1},\dots,\xi_s\}$. \label{alg-easyproperty-3}
  \item For any $e \in T$, $T\in \moveSet_k$, $\pi_k^{-1}(e)$ consists of one path from $\mathfrak{P}_k$\label{alg-easyproperty-3b}.
	\item If $\{v_{j_1},v_{j_2}\}$ with $j_1 < j_2$ is an edge in $T, (T,r) \in \moveSet_k$, then $v_{j_1}$ is closer to $r$ than $v_{j_2}$. \label{alg-easyproperty-4}
\end{enumerate}
\end{invariants}

The initialization consists of setting $\moveSet_2$ and $\mathfrak{P}_2$.
Observe that $\xi_2$ is visited exactly once by $C$: 
If there were $v_i,v_j \in C$ with $i\not=j$, but $v_i=v_j=\xi_2$, the sub-circuit $(v_i,\dots,v_j)$ would be such that $v_i > v_j$ for all $j \in \{i+1,\dots,j-1\}$, because $\xi_1$ only occurs at $v_1$ and $v_{|C|+1}$. This would be a contradiction to the assumption that $C$ is minimally guarded.

The algorithm splits $(v_1,\ldots,v_{|C|})$ at the unique occurrence $v$ of $\xi_2$ on $C$.
This is done by setting $\mathfrak{P}_2$ to consist of the paths $(v_1,\dots,v)$ and $(v,\dots,v_{|C|})$ and by inserting the tree $T=\{\{{v_1,v}\}\}$ with root $v_1$ into $\moveSet_2$. 
Notice that now $\pi_2((v_1,\dots,v))=\{v_1,v\}$ in $T$ and $\pi_2((v,\ldots,v_{|C|}))=\bot$.
Invariant~\ref{obs:mainalginvariant} is true because there is now one tree paying for $\xi_2$.

For $k \ge 3$, we assume that the properties are true for $k-1$ by induction.
The algorithm starts by setting $\moveSet_k = \moveSet_{k-1}$ and $\mathfrak{P}_k = \mathfrak{P}_{k-1}$. 
Then it considers the set $I = \cset{j}{v_j = \xi_{k}}$ of all occurrences of $\xi_{k}$ on $C$. 
For all $j\in I$ it follows from Properties~\ref{alg-easyproperty-2} and ~\ref{alg-easyproperty-3} that there is a unique path $P_j=(v_{P_j},\dots,w_{P_j})$ in $\mathfrak{P}_{k-1}$ with $j \in P_j$. The algorithm defines $P_j$ in Line~\ref{alg:defpj}. 
We know that the paths $P_j$ are different for all $j \in I$: By Property~\ref{alg-easyproperty-3}, all inner nodes are $\xi_{k}$ or a $\xi_\ell$ with higher index. If at least two occurrences of $\xi_{k}$ would fall on the same $P_j$, take the two that are closest together: All nodes between them would be equal to a $\xi_{j'}$ with $j' > k$, which contradicts the assumption that $C$ is minimally guarded (because the $\xi_\ell$ are sorted decreasingly). Thus, all $j\in I$ have a distinct $P_j$ that they lie on.
During the whole algorithm, the endpoints of $P_j$ are always named $v_{P_j}$ and $w_{P_j}$, where $v_{P_j}$ occurs first on $C$. 

Suppose that $\pi_{k-1}(P_j)\not= \bot$. 
The algorithm deals with all $j$ that satisfy this in the for loop that starts in Line~\ref{alg:jonTcase}. For any $T$ with occurences of $\xi_k$, it considers $I_T = \cset{j \in I}{\pi_{k-1}(P_j) \in T}$, 
the set of all occurrences of $\xi_{k}$ that fall into the same tree $T\in \moveSet_{k-1}$. It select a node $j^\ast \in I_T$ whose unique path to the root of $T$ does not contain any other $j \in I_T$. This node must exist because $T$ is a tree. The algorithm updates $T$, and adds a new tree for any $j \in I_T\backslash\{j^\ast\}$.
The idea is that the edge that $j^\ast$ falls on is divided into two edges that stay in $T$, while all other edges are split into an edge that stays in $T$ and an edge that forms a new tree. Since all $v_j$ with $j\in I_T$ represent the same connected component, $T$ stays connected.

More precisely, for $j^\ast$ the edge $\{v_{p_{j^\ast}}, w_{p_{j^\ast}}\}$ is divided in $T$ and is replaced by two edges $\{v_{p_{j^\ast}}, v_{j^\ast}\}$ and $\{v_{j^\ast}, w_{p_{j^\ast}}\}$. For all $j \in I_T\setminus\{j^\ast\}$, the algorithm replaces $\{v_{p_{j}}, w_{p_{j}}\}$ by $\{v_j, w_{P_j}\}$ in $T$. By Property~\ref{alg-easyproperty-4}, we know that $v_{P_j}$ is closer to the root $r$ of $T$. Thus, removing $\{v_{p_{j}}, w_{p_{j}}\}$ disconnects the subtree at $w_{p_{j}}$ from $r$. However, adding $\{v_j, w_{P_j}\}$ reconnects the tree because $v_j$ and $v_{j^\ast}$ are the same node and the algorithm assured that $v_{j^\ast}\in V[T]$. Thus, $T$ stays connected. The algorithm also adds the new tree $\{\{v_{P_j},v_{j}\}\}$ to $\moveSet_k$ which it can do because this part of $C$ is now free.
 
The algorithm also updates $\mathfrak{P}_k$ by removing $P_{j}$, inserting $\{v_{P_j},\dots,v_j\}$ and $\{v_j,\dots,w_{P_j}\}$ instead, thus splitting $P_{j}$ at $v_{j}$ for all $j \in I_T$ including $j^\ast$.

The algorithm also processes all $j \in I$ where $\pi_{k-1}(P_j)=\bot$. This is done in the for loop in Line~\ref{alg:jinfreepartcase}.
In this case, the path $P_j$ is split at node $j$ by removing $P_j$ from $\mathfrak{P}_k$. Then, the algorithm inserts its two parts $(v_{P_j},\dots,v_j)$ and $(v_j,\dots,w_{P_j})$ into $\mathfrak{P}_k$ to $\mathfrak{P}_k$. It also adds a new tree $\{\{v_{P_j},v_j\}\}$ with root $v_{P_j}$ to $\moveSet_k$. 

Notice that all trees that are created satisfy that there is a vertex with a higher number than $\xi_{k}$: if a vertex is added into a tree, then the other vertices in the tree have higher value, and if a new tree is created, then it consists of an edge to a vertex which previously was an endpoint of a path, and these have numbers in $\{\xi_1,\ldots,\xi_{k-1}\}$. For every $j\in I$, it either happens that a tree is updated or that a new tree is created. Thus, the set $\moveSet_k$ satisfies Invariant~\ref{obs:mainalginvariant} when the iteration is completed. 

\begin{lemma}\label{lem:algoonedoesitsjob}
Invariant~\ref{obs:mainalginvariant} holds.
\end{lemma}

Verifying the other invariants consists of checking all updates on $\moveSet_k$ and $\mathfrak{P}_k$.

\begin{lemma}
Invariants~\ref{obs:easy-invariants} hold.\label{lem:invariantsproof}
\end{lemma}
\begin{proof}
It is easy to verify that all properties hold for $k=2$ since the algorithm sets $\moveSet_2=(\{v_1,v\},v_1)$ and $\mathfrak{P}_2=\{ (v_1,\dots, v), (v,\dots,v_{|C|-1})\}$.
Property~\ref{alg-easyproperty-1} is also true for the new tree created in Line~\ref{alg:newtree:treedef} because it is only executed if $P_j$ is unclaimed. Line~\ref{alg:adapttree:jast} subdivides an edge into two. Lines~\ref{alg:adapttree:splitedge} and ~\ref{alg:adapttree:newtree} split an existing edge and distributes it among $T$ and $T'$. Thus, Property~\ref{alg-easyproperty-1} is preserved.
For Property~\ref{alg-easyproperty-2}, consider Lines~\ref{alg:newtree:pathsplit},~\ref{alg:adapttree:pathsplitjast} and~\ref{alg:adapttree:pathsplitj} to verify that paths are only split into subpaths and no edges are lost.
Property~\ref{alg-easyproperty-3} stays true because iteration $k$ processes all occurrences of $\xi_k$ and always executes one of the Lines~\ref{alg:newtree:pathsplit},~\ref{alg:adapttree:pathsplitjast} and~\ref{alg:adapttree:pathsplitj}, thus splitting the corresponding paths such that $\xi_k$ becomes an outer node.
Lines~\ref{alg:newtree:treedef},~\ref{alg:newtree:pathsplit},~\ref{alg:adapttree:jast},~\ref{alg:adapttree:pathsplitjast},~\ref{alg:adapttree:splitedge},~\ref{alg:adapttree:newtree},~\ref{alg:adapttree:pathsplitj} affect Property~\ref{alg-easyproperty-3b}. In all cases, $\moveSet_k$ and $\mathfrak{P}_k$ are adjusted consistently.

Finally, consider Property~\ref{alg-easyproperty-4}. Line~\ref{alg:newtree:treedef} creates a new tree by claiming $(v_{P_j},\ldots,v_j)$ of the unclaimed edges on $(v_{P_j},\ldots,w_{P_j})$. Notice that we assume that $v_{P_j}$ occurs on $C$ before $v_{P_j}$. Thus, assuming that Property~\ref{alg-easyproperty-4} holds for all trees existing before Line~\ref{alg:newtree:treedef}, we see that it also holds for the new tree. Line~\ref{alg:adapttree:jast} modifies a tree $T$ by inserting $v_{j^\ast}$ into an edge. Since $v_{j^\ast}$ is an inner node of $(v_{P_{j^\ast}},\ldots,w_{j^\ast})$, Property~\ref{alg-easyproperty-4} is preserved.
Line~\ref{alg:adapttree:splitedge} splits edge $\{v_{P_j},w_{P_j}\}$. Again, recall that $v_{P_j}$ occurs before $w_{P_j}$ on $C$ and thus inductively is closer to $r$. Further notice that $v_j$ occurs before $w_{P_j}$ and that $w_{P_j}$ gets disconnected from $r$ when $\{v_{P_j},w_{P_j}\}$ is removed. It is then reconnected to $r$ by adding $\{v_{j},w_{P_j}\}$. This means that $v_j$ is closer to $r$ than $w_{P_j}$.
Line~\ref{alg:adapttree:newtree} creates a new tree that satisfies Property~\ref{alg-easyproperty-4} because $v_{P_j}$ occurs on $C$ before $v_j$.
\end{proof}

\begin{corollary}\label{lem:minimallyguarded}
 Assume that $\ALG$ is $c$-approximate connecting move optimal. 
  Let $C=(v_1,\dots,v_l)$ be a circuit in $G_\ALG$ with edges $(e_1,\dots,e_{l})$.
  If $C$ is minimally guarded, then
  \begin{align*}
    \sum_{i=2}^{l-1} w(T_{v_i}) \leq c \cdot \sum_{i=1}^{l} d_{e_i} = c \cdot d(C)
  \end{align*}
\end{corollary}
\begin{proof}
For all $T\in \moveSet$ and all edges $e \in T \subseteq \moveSet_s$, Property~\ref{alg-easyproperty-3b}  says that there is a unique path $\pi^{-1}(e) \in \mathfrak{P}_s$. Property~\ref{alg-easyproperty-3} for $k=s$ means that paths can no longer have inner nodes. Thus, $\pi^{-1}(e)$ is a single edge, and therefore, $e$ also exist in $G_\ALG$. Thus, all trees in $\moveSet_s$ are trees in $G_\ALG$.
By Property~\ref{alg-easyproperty-3b}, the trees are edge disjoint. Lemma~\ref{lem:algoonedoesitsjob} ensures that they satisfy the precondition of Lemma~\ref{lem:treedecomp-to-lowerbound}. The corollary then follows.
\end{proof}

As before, we set $\xi(v)=j$ for the unique $j\in\{1,\dots,p\}$ with $v \in A_j$, for all $v \in V$.
For any edge set $F$ in $G$, we define $\Fin:=\cset{e=\{u,v\} \in  F}{\xi(u)=\xi(v)}$ and $\Fbetw:=\cset{e=\{u,v\} \in F}{\xi(u)\neq\xi(v)}$ as the subset of edges of $F$ within components of $\ALG$ or between them, respectively. 
Furthermore, if an edge set $F'$ in $G$ satisfies $V[F'] \subseteq V[A_j]$ for a $j\in\{1,\ldots,p\}$, then we set $\xi(F') = j$. Notice that in this case, $F'=\Fin'$.


\newcommand{\replace}{R}

\begin{lemma}\label{lem:guardedcycles}
Let $\bar F$ be a simple path in $G$ that starts and ends in the same connected component $T_{j^\ast}$ of $\ALG$ and satisfies that $\xi(v) \le j^\ast$ for all $v \in V[\bar F]$. Assume that $\bar F \neq \bFin$.
%
Assume that $\ALG$ is \edgeset and \pathset swap-optimal with respect to $\bFbetw$ and that $\ALG$ is $c$-approximate connecting move optimal. 

Then, there exists a set $\replace$ of edges on the vertices $V[\bFbetw]$ with $(\bFin \cup \replace)_{\leftrightarrow}=(\bFin \cup \replace)$ that satisfies the properties listed below.
Let $F_{1}',\ldots,F_{x}'$ be the connected components of $\bFin \cup \replace$ in $(V[F'],E[F'])$.



%
%
\begin{enumerate}[label=\arabic*.,ref={\ref{lem:guardedcycles}-\arabic*}]
\item $\ALG$ is \edgesetswap-optimal with respect to $\replace$.\label{Rproperties-1}
\item It holds that $d(\replace) \le d(\bFbetw)$ and $\sum_{\ell=2}^{x} w(T_{\xi(F_{\ell}')}) \le c \cdot d(\bFbetw)$.\label{Rproperties-4}\item For all $F_\ell'$, there exists an index $j$ such that $V[F_\ell']\subseteq V[A_j]$ (thus, $\xi(F_{\ell}') = j$).\label{Rproperties-2}
\item 
There is only one $F_{\ell}'$ with $\xi(F_{\ell}') = j^\ast$, assume w.l.o.g. that  $\xi(F_{1}')=j^\ast$.\label{Rproperties-3}
\end{enumerate}
\end{lemma}
\begin{proof}
Let $\bar{F} = (s,\dots,v_1,w_1,\dots,w_2,v_2,\dots,t)$ where $(s,\dots,v_1)$ and $(v_2,\dots,t)$ are the prefix and suffix of $\bar{F}$ lying in $T_{j^\ast}$, i.e., we assume that $s,\dots,v_1 \in T_{j^\ast}$, $v_2,\dots,t \in T_{j^\ast}$ and $w_1,w_2 \not\in T_{j^\ast}$. 
The nodes $s$ and $v_1$ may coincide as well as $v_2$ and $t$.
Let $\bar C$ be the circuit that $\bar F$ induces in $G_{\ALG}$.

We do induction on the inclusion-wise hierarchy of guarded circuits.
Thus, our base case is that $\bar C$ is minimally guarded. 
In this case, we know that all vertices $v$ from $(w_1,\dots,w_2)$ satisfy $\xi(v) < j^\ast$.
We set $\replace=\{\bar e\}$ with $\bar e := \{v_1,v_2\}$ and $d(\bar e) := d(\bFbetw)$ and show that $\replace$ satisfies Properties~\ref{Rproperties-1}--\ref{Rproperties-3}.
For Property~\ref{Rproperties-1}, we need \pathsetswap optimality. Picking $v_1$ and $v_2$ uniquely defines a set of edges $X$ which form a shortest path from $v_1$ to $v_2$ in the contracted graph and which every $v_1$-$v_2$-based path move adds. Let $(\bar e,S)$ be a \edgesetswap that adds $\bar e$. We argue that this move cannot be improving because otherwise, the \pathsetswap $(X,S)$ was improving. Let $loss(X)$ be the increase that adding $X$ to $\ALG$ incurs in $\phi$, and let $gain(S)$ be the amount by which deleting $S$ decreases $\phi$. Assume that $(\bar e,S)$ is improving, \ie $d(\bar e) = d(\bFbetw) < gain(S)$.
Notice that $\bFbetw$ is a path from $v_1$ to $v_2$ in the contracted graph. Thus, $d(\bFbetw) \ge d(X)$. Furthermore, notice that $loss(X) \le d(X)$. Thus, $loss(X) \le d(\bFbetw)$, such that our assumption implies $loss(X) < gain (S)$. That is a contradiction to \pathsetswap optimality. Property~\ref{Rproperties-1} holds.


Property~\ref{Rproperties-4} is true because $d(F')=d(\bFbetw)$ and by Corollary~\ref{lem:minimallyguarded} since $d(\bFbetw)=d(\bar C)$.
Now we look at the connected components of $\bFin\cup\replace$. They are equal to the connected components of $\bFin$ except that we add the edge $e$. Notice that $\bFin$ only contains edges that go within the same component, and that that $e$ connects two vertices from the same component. Thus, Property~\ref{Rproperties-3} holds. Furthermore, notice that $\bFin$ has exactly two connected components consisting of vertices from $T_{j^\ast}$: The vertices on the prefix and suffix of $\bar{F}$. These components are connected by $e=\{v_1,v_2\}$. Thus, Property~\ref{Rproperties-4} holds.
%

Now assume that $\bar C$ is not minimally guarded. First assume that $\bar C$ is not guarded.
Define $v_1$ and $v_2$ as before. Since $\bar C$ is not guarded, $\bar F$ has to visit $T_{j^\ast}$ again between $v_1$ and $v_2$. Let $v_3$ and $v_4$ be the first and last vertex of one arbitrary visit to $T_{j^\ast}$ between $v_1$ and $v_2$. It is possible that $v_3 = v_4$, otherwise, notice that $v_3$ and $v_4$ are connected in $\bFin$.
%
We split $\bar F$ into two paths $P_1:=(v_1,\ldots,v_3)$ and $P_2:=(v_4,\ldots,v_2)$ and obtain two sets ${F^1}'$ and ${F^2}'$ by using the induction hypothesis on the two inclusionwise smaller paths. 
By induction hypothesis, $v_1$, $v_3$ and all other occurrences of vertices from $T_{j^\ast}$ in $P_1$ have to be connected in $\bFin \cup {F^1}'$. Also, $v_4$, $v_2$ and all other occurrences of vertices from $T_{j^\ast}$ in $P_2$ have to be connected in $\bFin \cup{F^2}'$. Thus, all occurrences of vertices from $T_{j^\ast}$ on $\bar F$ are connected in $\bFin \cup {F^1}' \cup {F^2}'$ because $v_3$ is connected to $v_4$ in $\bFin$. So, Property~\ref{Rproperties-3} holds. 
Furthermore, since $\bar F$ is a simple path, no other components of ${F_1}'$ and ${F_2}'$ can contain the same vertex because only $v_3$ is in both $P_1$ and $P_2$.
Thus, Property~\ref{Rproperties-2} holds for ${F^1}' \cup {F^2}'$ since it holds for ${F^1}'$ and ${F^2}'$ individually. If $\ALG$ is \edgesetswap-optimal with respect to a set $A$ and also with respect to a set $B$, then it is \edgesetswap-optimal with respect to $A\cup B$, thus Property~\ref{Rproperties-1} holds for ${F^1}' \cup {F^2}'$. Similarly, Property~\ref{Rproperties-4} holds because it holds for ${{F^1_1}'}$ and ${{F^2_1}'}$ individually with respect to disjoint parts of $\bFbetw$.

Finally, assume that $\bar F$ is guarded, but not minimally guarded. 
Define $v_1$ and $v_2$ as before. Since $\bar F$ is guarded, but not minimally guarded, it visits a connected component $T_j$ with $j < j^\ast$ twice between $v_1$ and $v_2$, and between these two visits, it never visits a component $T_{j'}$ with $j' > j$. Pick a $j$ and two visits of $T_j$ with this property, and let $v_3$ be the last vertex in the first of these visits of $T_j$ and let $v_4$ be the first vertex of the second visit of $T_j$. Again, $v_3=v_4$ is possible and otherwise, $v_3$ and $v_4$ are connected in $\bFin$.
We apply the induction  to the path $\bar F'$ which is the subpath $(v_3,\ldots,v_4)$ and obtain a set $F'$ by the induction hypothesis. Additionally, we create $\bar F''$ from $\bar F$ by replacing the subpath $(v_3,\ldots,v_4)$ by the edge $(v_3,v_4)$. Since this path is shorter, we can apply the induction hypothesis to $\bar F''$ to obtain a set $F''$. We claim that $F'\cup F''$ satisfies all properties. Notice that $\bFbetw = \bar F_{\leftrightarrow}' \cup F_{\leftrightarrow}''$. Again, Property~\ref{Rproperties-1} is true because $F' \cup F''$ is the union of two sets that satisfy Property~\ref{Rproperties-1}. 

What are the connected components of $\bFin \cup F' \cup F''$?
Let $\mathcal{CC}$ be the set of connected components of $\bFin$, and define $\mathcal{CC'}$ and $\mathcal{CC}''$ to be the connected components of $\bFin \cup F'$ and $\bFin \cup F''$, respectively. All edges in $F' \cup F''$ go between different components in $\mathcal{CC}$, and $F'$ and $F''$ are defined on vertex sets that are disjoint with the exception of $v_3$ and $v_4$. 
Both $\mathcal{C}'$ and $\mathcal{C}''$ contain exactly one connected component which contains both $v_3$ and $v_4$ (notice that $v_3$ and $v_4$ are connected in $\bar F''$ and $\bFin$). In $\mathcal{C}'$, it is the connected component with the highest width (\ie, $F_1'$), but in $\mathcal{CC}''$, the component with the highest width is the component that contains $v_1$ and $v_2$. Thus, $v_1,v_2 \in F_1''$ and $v_3,v_4 \in F_{j''}''$ with $j'' \neq 1$. In $\bFin \cup F' \cup F''$, $F_1'$ and $F_{j'''}$ are merged into one connected component because they both contain $v_3$ and $v_4$, meaning that $w(F_1')=w(F_{j'''})$ is counted twice when applying Property~\ref{Rproperties-4} for $F'$ and $F''$ compared to applying Property~\ref{Rproperties-4} to $F' \cup F''$. All other components in $\mathcal{C}'$ and $\mathcal{C}''$ are defined on disjoint vertex sets and thus $\bFin \cup F' \cup F''$ is the disjoint union of $\mathcal{CC}'\backslash \{ F_1'\}$, $\mathcal{CC}'' \backslash \{F_{j'''}'\}$ and $\{F_1'\cup F_{j''}'\}$. We see that Property~\ref{Rproperties-2} and~\ref{Rproperties-3} carry over to $F' \cup F''$ from holding for $F'$ and $F''$.
Name the connected components of $\bFin \cup F' \cup F''$ as $F_1''',\ldots,F_y'''$. 
The component with the highest width among these, $F_1'''$, is the one that contains $v_1$ and $v_2$, \ie $F_1''' = F_1''$. We thus have
\[
\sum_{\ell=2}^{y} w(T_{\xi(F_{\ell}''')}) = \sum_{\ell=2}^{x'} w(T_{\xi(F_{\ell}')})  + \sum_{\ell=2}^{x''} w(T_{\xi(F_{\ell}'')}) 
\le c\cdot d(F_{\leftrightarrow}')+c\cdot d(F_{\leftrightarrow}'') = c\cdot d(\bFbetw).
\]
Since $\bFbetw = \bar F_{\leftrightarrow}' \cup F_{\leftrightarrow}''$ and $d(F') \le d(\bar F_{\leftrightarrow}')$ and $d(F'') \le d(\bar F_{\leftrightarrow}'')$ by the induction hypothesis, we have $d(F'\cup F'') \le d(\bFbetw)$ and thus $F' \cup F''$ satisfies Property~\ref{Rproperties-4}.
%
\end{proof}

Figure~\ref{fig:ex:1}-\ref{fig:ex:last} illustrate how $F'$ is constructed recursively.

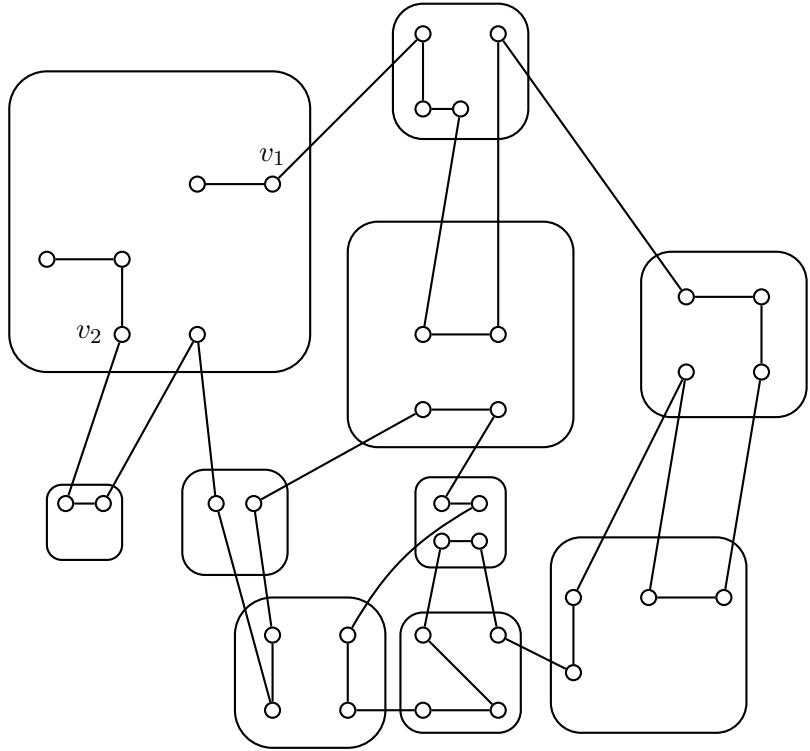
\begin{figure}[ht]
\begin{center}
\begin{tikzpicture}[thick, mynode/.style={draw, circle,inner sep=0cm, minimum size=0.2cm}]
\draw[rounded corners=5mm] (-0.5,5.5) rectangle (3.5,9.5); 

\draw[rounded corners=4mm] (4,4.5) rectangle (7,7.5); 
\draw[rounded corners=4mm] (6.7,0.7) rectangle (9.3,3.3); 
\draw[rounded corners=4mm] (7.9,4.9) rectangle (10.1,7.1); 

\draw[rounded corners=5mm] (2.5,0.5) rectangle (4.5,2.5); 
\draw[rounded corners=4mm] (4.6,8.6) rectangle (6.4,10.4); 
\draw[rounded corners=3mm] (4.7,0.7) rectangle (6.3,2.3); 
\draw[rounded corners=3mm] (1.8,2.8) rectangle (3.2,4.2); 
\draw[rounded corners=2mm] (4.9,2.9) rectangle (6.1,4.1); 
\draw[rounded corners=2mm] (-0.0,3) rectangle (1.0,4.0); 
\node [mynode] (n1) at (2,8) {};
\node [mynode, label=above:{$v_1$}] (n2) at (3,8) {};
\node [mynode] (n3) at (5,10) {};
\node [mynode] (n4) at (5,9) {};
\node [mynode] (n5) at (5.5,9) {};
\node [mynode] (n6) at (5,6) {};
\node [mynode] (n7) at (6,6) {};
\node [mynode] (n8) at (6,10) {};
\node [mynode] (n9) at (8.5,6.5) {};
\node [mynode] (n10) at (9.5,6.5) {};
\node [mynode] (n11) at (9.5,5.5) {};
\node [mynode] (n12) at (9,2.5) {};
\node [mynode] (n13) at (8,2.5) {};
\node [mynode] (n14) at (8.5,5.5) {};
\node [mynode] (n15) at (7,2.5) {};
\node [mynode] (n16) at (7,1.5) {};
\node [mynode] (n17) at (6,2) {};
\node [mynode] (n18) at (5.75,3.25) {};
\node [mynode] (n19) at (5.25,3.25) {};
\node [mynode] (n20) at (5,2) {};
\node [mynode] (n21) at (6,1) {};
\node [mynode] (n22) at (5,1) {};
\node [mynode] (n23) at (4,1) {};
\node [mynode] (n24) at (4,2) {};
\node [mynode] (n25) at (5.75,3.75) {};
\node [mynode] (n26) at (5.25,3.75) {};
\node [mynode] (n27) at (6,5) {};
\node [mynode] (n28) at (5,5) {};
\node [mynode] (n29) at (2.75,3.75) {};
\node [mynode] (n30) at (3,2) {};
\node [mynode] (n31) at (3,1) {};
\node [mynode] (n32) at (2.25,3.75) {};
\node [mynode] (n33) at (2,6) {};
\node [mynode] (n34) at (0.75,3.75) {};
\node [mynode] (n35) at (0.25,3.75) {};
\node [mynode] (n36) [label=left:{$v_2$}] at (1,6) {};
\node [mynode] (n37) at (1,7) {};
\node [mynode] (n38) at (0,7) {};
\draw (n1) -- (n2) -- (n3) -- (n4) -- (n5) -- (n6) -- (n7) -- (n8) -- (n9) -- (n10) -- (n11) -- (n12) -- (n13) -- (n14) -- (n15) -- (n16) -- (n17) -- (n18) -- (n19) -- (n20) -- (n21) -- (n22) -- (n23) -- (n24); \draw (n24) [bend left=15] to (n25); \draw (n25) -- (n26) -- (n27) -- (n28) -- (n29) -- (n30) -- (n31) -- (n32) -- (n33) -- (n34) -- (n35) -- (n36) -- (n37) -- (n38);
\end{tikzpicture}
\caption{The rounded rectangles visualize the connected components of $\ALG$, their size indicates their width. We see a path $\bar F$ that corresponds to a circuit in $G_{\ALG}$ which is guarded, but not minimally guarded. We pick $v_1$ and $v_2$ as described.\label{fig:ex:1}}
\end{center}
\end{figure}

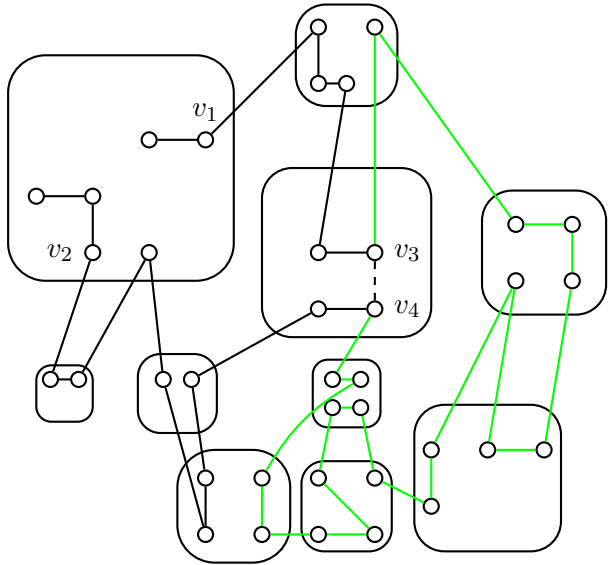
\begin{figure}
\begin{center}
\begin{tikzpicture}[scale=0.75,thick, mynode/.style={draw, circle,inner sep=0cm, minimum size=0.2cm}]
\draw[rounded corners=5mm] (-0.5,5.5) rectangle (3.5,9.5); 

\draw[rounded corners=4mm] (4,4.5) rectangle (7,7.5); 
\draw[rounded corners=4mm] (6.7,0.7) rectangle (9.3,3.3); 
\draw[rounded corners=4mm] (7.9,4.9) rectangle (10.1,7.1); 

\draw[rounded corners=5mm] (2.5,0.5) rectangle (4.5,2.5); 
\draw[rounded corners=4mm] (4.6,8.6) rectangle (6.4,10.4); 
\draw[rounded corners=3mm] (4.7,0.7) rectangle (6.3,2.3); 
\draw[rounded corners=3mm] (1.8,2.8) rectangle (3.2,4.2); 
\draw[rounded corners=2mm] (4.9,2.9) rectangle (6.1,4.1); 
\draw[rounded corners=2mm] (-0.0,3) rectangle (1.0,4.0); 
\node [mynode] (n1) at (2,8) {};
\node [mynode, label=above:{$v_1$}] (n2) at (3,8) {};
\node [mynode] (n3) at (5,10) {};
\node [mynode] (n4) at (5,9) {};
\node [mynode] (n5) at (5.5,9) {};
\node [mynode] (n6) at (5,6) {};
\node [mynode,label=right:{$v_3$}] (n7) at (6,6) {};
\node [mynode] (n8) at (6,10) {};
\node [mynode] (n9) at (8.5,6.5) {};
\node [mynode] (n10) at (9.5,6.5) {};
\node [mynode] (n11) at (9.5,5.5) {};
\node [mynode] (n12) at (9,2.5) {};
\node [mynode] (n13) at (8,2.5) {};
\node [mynode] (n14) at (8.5,5.5) {};
\node [mynode] (n15) at (7,2.5) {};
\node [mynode] (n16) at (7,1.5) {};
\node [mynode] (n17) at (6,2) {};
\node [mynode] (n18) at (5.75,3.25) {};
\node [mynode] (n19) at (5.25,3.25) {};
\node [mynode] (n20) at (5,2) {};
\node [mynode] (n21) at (6,1) {};
\node [mynode] (n22) at (5,1) {};
\node [mynode] (n23) at (4,1) {};
\node [mynode] (n24) at (4,2) {};
\node [mynode] (n25) at (5.75,3.75) {};
\node [mynode] (n26) at (5.25,3.75) {};
\node [mynode,label=right:{$v_4$}] (n27) at (6,5) {};
\node [mynode] (n28) at (5,5) {};
\node [mynode] (n29) at (2.75,3.75) {};
\node [mynode] (n30) at (3,2) {};
\node [mynode] (n31) at (3,1) {};
\node [mynode] (n32) at (2.25,3.75) {};
\node [mynode] (n33) at (2,6) {};
\node [mynode] (n34) at (0.75,3.75) {};
\node [mynode] (n35) at (0.25,3.75) {};
\node [mynode] (n36) [label=left:{$v_2$}] at (1,6) {};
\node [mynode] (n37) at (1,7) {};
\node [mynode] (n38) at (0,7) {};
\draw (n1) -- (n2) -- (n3) -- (n4) -- (n5) -- (n6) -- (n7); \draw [dashed] (n7) -- (n27); \draw (n27) -- (n28) -- (n29) -- (n30) -- (n31) -- (n32) -- (n33) -- (n34) -- (n35) -- (n36) -- (n37) -- (n38);

\draw [green] (n7) -- (n8) -- (n9) -- (n10) -- (n11) -- (n12) -- (n13) -- (n14) -- (n15) -- (n16) -- (n17) -- (n18) -- (n19) -- (n20) -- (n21) -- (n22) -- (n23) -- (n24); \draw [green] (n24) [bend left=15] to (n25); \draw [green] (n25) -- (n26) -- (n27);
\end{tikzpicture}
\caption{We identify suitable $v_3$ and $v_4$ and split the path into two paths: The green path, and the  black path that now contains the dashed edge.}
\end{center}
\end{figure}

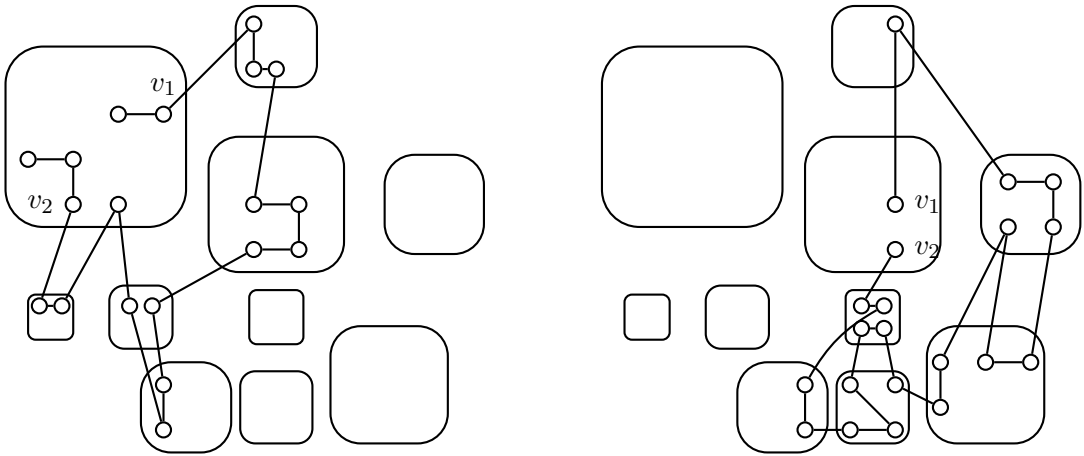
\begin{figure}
\begin{center}
\begin{tikzpicture}[scale=0.6,thick, mynode/.style={draw, circle,inner sep=0cm, minimum size=0.2cm}]
\draw[rounded corners=5mm] (-0.5,5.5) rectangle (3.5,9.5); 

\draw[rounded corners=4mm] (4,4.5) rectangle (7,7.5); 
\draw[rounded corners=4mm] (6.7,0.7) rectangle (9.3,3.3); 
\draw[rounded corners=4mm] (7.9,4.9) rectangle (10.1,7.1); 

\draw[rounded corners=4mm] (2.5,0.5) rectangle (4.5,2.5); 
\draw[rounded corners=3mm] (4.6,8.6) rectangle (6.4,10.4); 
\draw[rounded corners=2mm] (4.7,0.7) rectangle (6.3,2.3); 
\draw[rounded corners=2mm] (1.8,2.8) rectangle (3.2,4.2); 
\draw[rounded corners=1mm] (4.9,2.9) rectangle (6.1,4.1); 
\draw[rounded corners=1mm] (-0.0,3) rectangle (1.0,4.0); 
\node [mynode] (n1) at (2,8) {};
\node [mynode, label=above:{$v_1$}] (n2) at (3,8) {};
\node [mynode] (n3) at (5,10) {};
\node [mynode] (n4) at (5,9) {};
\node [mynode] (n5) at (5.5,9) {};
\node [mynode] (n6) at (5,6) {};
\node [mynode] (n7) at (6,6) {};
\node [mynode] (n27) at (6,5) {};
\node [mynode] (n28) at (5,5) {};
\node [mynode] (n29) at (2.75,3.75) {};
\node [mynode] (n30) at (3,2) {};
\node [mynode] (n31) at (3,1) {};
\node [mynode] (n32) at (2.25,3.75) {};
\node [mynode] (n33) at (2,6) {};
\node [mynode] (n34) at (0.75,3.75) {};
\node [mynode] (n35) at (0.25,3.75) {};
\node [mynode] (n36) [label=left:{$v_2$}] at (1,6) {};
\node [mynode] (n37) at (1,7) {};
\node [mynode] (n38) at (0,7) {};
\draw (n1) -- (n2) -- (n3) -- (n4) -- (n5) -- (n6) -- (n7); \draw  (n7) -- (n27); \draw (n27) -- (n28) -- (n29) -- (n30) -- (n31) -- (n32) -- (n33) -- (n34) -- (n35) -- (n36) -- (n37) -- (n38);
\end{tikzpicture}\qquad\qquad
\begin{tikzpicture}[scale=0.6,thick, mynode/.style={draw, circle,inner sep=0cm, minimum size=0.2cm}]
\draw[rounded corners=5mm] (-0.5,5.5) rectangle (3.5,9.5); 

\draw[rounded corners=4mm] (4,4.5) rectangle (7,7.5); 
\draw[rounded corners=4mm] (6.7,0.7) rectangle (9.3,3.3); 
\draw[rounded corners=4mm] (7.9,4.9) rectangle (10.1,7.1); 

\draw[rounded corners=4mm] (2.5,0.5) rectangle (4.5,2.5); 
\draw[rounded corners=3mm] (4.6,8.6) rectangle (6.4,10.4); 
\draw[rounded corners=2mm] (4.7,0.7) rectangle (6.3,2.3); 
\draw[rounded corners=2mm] (1.8,2.8) rectangle (3.2,4.2); 
\draw[rounded corners=1mm] (4.9,2.9) rectangle (6.1,4.1); 
\draw[rounded corners=1mm] (-0.0,3) rectangle (1.0,4.0); 
\node [mynode,label=right:{$v_1$}] (n7) at (6,6) {};
\node [mynode] (n8) at (6,10) {};
\node [mynode] (n9) at (8.5,6.5) {};
\node [mynode] (n10) at (9.5,6.5) {};
\node [mynode] (n11) at (9.5,5.5) {};
\node [mynode] (n12) at (9,2.5) {};
\node [mynode] (n13) at (8,2.5) {};
\node [mynode] (n14) at (8.5,5.5) {};
\node [mynode] (n15) at (7,2.5) {};
\node [mynode] (n16) at (7,1.5) {};
\node [mynode] (n17) at (6,2) {};
\node [mynode] (n18) at (5.75,3.25) {};
\node [mynode] (n19) at (5.25,3.25) {};
\node [mynode] (n20) at (5,2) {};
\node [mynode] (n21) at (6,1) {};
\node [mynode] (n22) at (5,1) {};
\node [mynode] (n23) at (4,1) {};
\node [mynode] (n24) at (4,2) {};
\node [mynode] (n25) at (5.75,3.75) {};
\node [mynode] (n26) at (5.25,3.75) {};
\node [mynode,label=right:{$v_2$}] (n27) at (6,5) {};

\draw (n7) -- (n8) -- (n9) -- (n10) -- (n11) -- (n12) -- (n13) -- (n14) -- (n15) -- (n16) -- (n17) -- (n18) -- (n19) -- (n20) -- (n21) -- (n22) -- (n23) -- (n24); \draw  (n24) [bend left=15] to (n25); \draw  (n25) -- (n26) -- (n27);
\end{tikzpicture}
\caption{These are the two paths for which we use the induction hypothesis. Both are guarded, but not minimally guarded.\label{fig:ex:3}}
\end{center}
\end{figure}

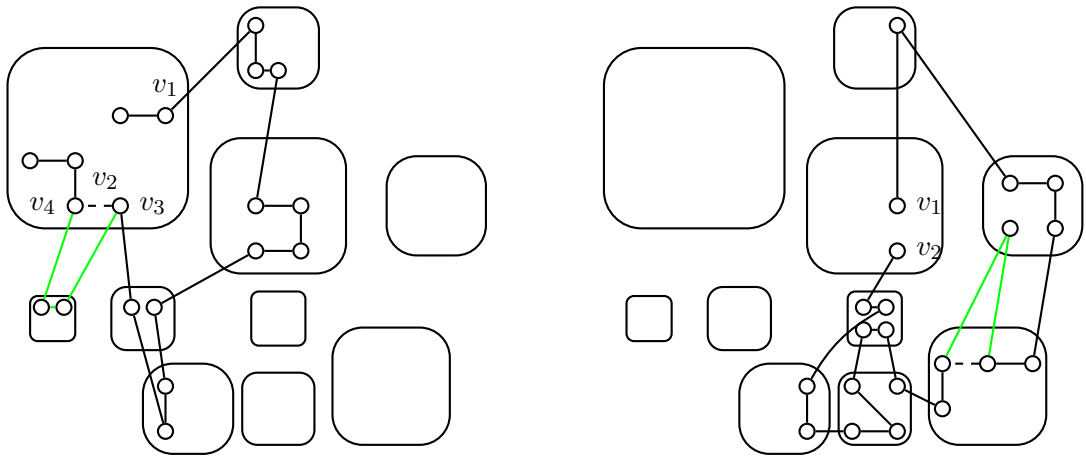
\begin{figure}
\begin{center}
\begin{tikzpicture}[scale=0.6,thick, mynode/.style={draw, circle,inner sep=0cm, minimum size=0.2cm}]
\draw[rounded corners=5mm] (-0.5,5.5) rectangle (3.5,9.5); 

\draw[rounded corners=4mm] (4,4.5) rectangle (7,7.5); 
\draw[rounded corners=4mm] (6.7,0.7) rectangle (9.3,3.3); 
\draw[rounded corners=4mm] (7.9,4.9) rectangle (10.1,7.1); 

\draw[rounded corners=4mm] (2.5,0.5) rectangle (4.5,2.5); 
\draw[rounded corners=3mm] (4.6,8.6) rectangle (6.4,10.4); 
\draw[rounded corners=2mm] (4.7,0.7) rectangle (6.3,2.3); 
\draw[rounded corners=2mm] (1.8,2.8) rectangle (3.2,4.2); 
\draw[rounded corners=1mm] (4.9,2.9) rectangle (6.1,4.1); 
\draw[rounded corners=1mm] (-0.0,3) rectangle (1.0,4.0); 
\node [mynode] (n1) at (2,8) {};
\node [mynode, label=above:{$v_1$}] (n2) at (3,8) {};
\node [mynode] (n3) at (5,10) {};
\node [mynode] (n4) at (5,9) {};
\node [mynode] (n5) at (5.5,9) {};
\node [mynode] (n6) at (5,6) {};
\node [mynode] (n7) at (6,6) {};
\node [mynode] (n27) at (6,5) {};
\node [mynode] (n28) at (5,5) {};
\node [mynode] (n29) at (2.75,3.75) {};
\node [mynode] (n30) at (3,2) {};
\node [mynode] (n31) at (3,1) {};
\node [mynode] (n32) at (2.25,3.75) {};
\node [mynode] (n33) [label=right:{$v_3$}]at (2,6) {};
\node [mynode] (n34) at (0.75,3.75) {};
\node [mynode] (n35) at (0.25,3.75) {};
\node [mynode] (n36) [label=left:{$v_4$},label=north east:{$v_2$}] at (1,6) {};
\node [mynode] (n37) at (1,7) {};
\node [mynode] (n38) at (0,7) {};
\draw (n1) -- (n2) -- (n3) -- (n4) -- (n5) -- (n6) -- (n7); \draw  (n7) -- (n27); \draw (n27) -- (n28) -- (n29) -- (n30) -- (n31) -- (n32) -- (n33);
\draw (n36) -- (n37) -- (n38);
\draw [dashed] (n33) -- (n36);

\draw [green] (n33) -- (n34) -- (n35) -- (n36);
\end{tikzpicture}\qquad\qquad
\begin{tikzpicture}[scale=0.6,thick, mynode/.style={draw, circle,inner sep=0cm, minimum size=0.2cm}]
\draw[rounded corners=5mm] (-0.5,5.5) rectangle (3.5,9.5); 

\draw[rounded corners=4mm] (4,4.5) rectangle (7,7.5); 
\draw[rounded corners=4mm] (6.7,0.7) rectangle (9.3,3.3); 
\draw[rounded corners=4mm] (7.9,4.9) rectangle (10.1,7.1); 

\draw[rounded corners=4mm] (2.5,0.5) rectangle (4.5,2.5); 
\draw[rounded corners=3mm] (4.6,8.6) rectangle (6.4,10.4); 
\draw[rounded corners=2mm] (4.7,0.7) rectangle (6.3,2.3); 
\draw[rounded corners=2mm] (1.8,2.8) rectangle (3.2,4.2); 
\draw[rounded corners=1mm] (4.9,2.9) rectangle (6.1,4.1); 
\draw[rounded corners=1mm] (-0.0,3) rectangle (1.0,4.0); 
\node [mynode,label=right:{$v_1$}] (n7) at (6,6) {};
\node [mynode] (n8) at (6,10) {};
\node [mynode] (n9) at (8.5,6.5) {};
\node [mynode] (n10) at (9.5,6.5) {};
\node [mynode] (n11) at (9.5,5.5) {};
\node [mynode] (n12) at (9,2.5) {};
\node [mynode] (n13) at (8,2.5) {};
\node [mynode] (n14) at (8.5,5.5) {};
\node [mynode] (n15) at (7,2.5) {};
\node [mynode] (n16) at (7,1.5) {};
\node [mynode] (n17) at (6,2) {};
\node [mynode] (n18) at (5.75,3.25) {};
\node [mynode] (n19) at (5.25,3.25) {};
\node [mynode] (n20) at (5,2) {};
\node [mynode] (n21) at (6,1) {};
\node [mynode] (n22) at (5,1) {};
\node [mynode] (n23) at (4,1) {};
\node [mynode] (n24) at (4,2) {};
\node [mynode] (n25) at (5.75,3.75) {};
\node [mynode] (n26) at (5.25,3.75) {};
\node [mynode,label=right:{$v_2$}] (n27) at (6,5) {};

\draw (n7) -- (n8) -- (n9) -- (n10) -- (n11) -- (n12) -- (n13);
\draw (n15) -- (n16) -- (n17) -- (n18) -- (n19) -- (n20) -- (n21) -- (n22) -- (n23) -- (n24); \draw  (n24) [bend left=15] to (n25); \draw  (n25) -- (n26) -- (n27);
\draw [dashed] (n13) -- (n15);
\draw [green] (n13) -- (n14) -- (n15);
\end{tikzpicture}
\end{center}
\caption{We identify $v_3$ and $v_4$ for the two paths in Figure~\ref{fig:ex:3} and split the path into the green part and the black part with the dashed edge.}
\end{figure}

\begin{figure}
\begin{center}
\begin{tikzpicture}[scale=0.6,thick, mynode/.style={draw, circle,inner sep=0cm, minimum size=0.2cm}]
\draw[rounded corners=5mm] (-0.5,5.5) rectangle (3.5,9.5); 

\draw[rounded corners=4mm] (4,4.5) rectangle (7,7.5); 
\draw[rounded corners=4mm] (6.7,0.7) rectangle (9.3,3.3); 
\draw[rounded corners=4mm] (7.9,4.9) rectangle (10.1,7.1); 

\draw[rounded corners=4mm] (2.5,0.5) rectangle (4.5,2.5); 
\draw[rounded corners=3mm] (4.6,8.6) rectangle (6.4,10.4); 
\draw[rounded corners=2mm] (4.7,0.7) rectangle (6.3,2.3); 
\draw[rounded corners=2mm] (1.8,2.8) rectangle (3.2,4.2); 
\draw[rounded corners=1mm] (4.9,2.9) rectangle (6.1,4.1); 
\draw[rounded corners=1mm] (-0.0,3) rectangle (1.0,4.0); 
\node [mynode] (n1) at (2,8) {};
\node [mynode, label=above:{$v_1$}] (n2) at (3,8) {};
\node [mynode] (n3) at (5,10) {};
\node [mynode] (n4) at (5,9) {};
\node [mynode] (n5) at (5.5,9) {};
\node [mynode] (n6) at (5,6) {};
\node [mynode] (n7) at (6,6) {};
\node [mynode] (n27) at (6,5) {};
\node [mynode] (n28) at (5,5) {};
\node [mynode] (n29) at (2.75,3.75) {};
\node [mynode] (n30) at (3,2) {};
\node [mynode] (n31) at (3,1) {};
\node [mynode] (n32) at (2.25,3.75) {};
\node [mynode] (n33) at (2,6) {};
\node [mynode] (n36) [label=left:{$v_2$}] at (1,6) {};
\node [mynode] (n37) at (1,7) {};
\node [mynode] (n38) at (0,7) {};
\draw (n1) -- (n2) -- (n3) -- (n4) -- (n5) -- (n6) -- (n7); \draw  (n7) -- (n27); \draw (n27) -- (n28) -- (n29) -- (n30) -- (n31) -- (n32) -- (n33);
\draw (n36) -- (n37) -- (n38);
\draw (n33) -- (n36);
\end{tikzpicture}\qquad\qquad
\begin{tikzpicture}[scale=0.6,thick, mynode/.style={draw, circle,inner sep=0cm, minimum size=0.2cm}]
\draw[rounded corners=5mm] (-0.5,5.5) rectangle (3.5,9.5); 

\draw[rounded corners=4mm] (4,4.5) rectangle (7,7.5); 
\draw[rounded corners=4mm] (6.7,0.7) rectangle (9.3,3.3); 
\draw[rounded corners=4mm] (7.9,4.9) rectangle (10.1,7.1); 

\draw[rounded corners=4mm] (2.5,0.5) rectangle (4.5,2.5); 
\draw[rounded corners=3mm] (4.6,8.6) rectangle (6.4,10.4); 
\draw[rounded corners=2mm] (4.7,0.7) rectangle (6.3,2.3); 
\draw[rounded corners=2mm] (1.8,2.8) rectangle (3.2,4.2); 
\draw[rounded corners=1mm] (4.9,2.9) rectangle (6.1,4.1); 
\draw[rounded corners=1mm] (-0.0,3) rectangle (1.0,4.0); 
\node [mynode,label=right:{$v_1$}] (n7) at (6,6) {};
\node [mynode] (n8) at (6,10) {};
\node [mynode] (n9) at (8.5,6.5) {};
\node [mynode] (n10) at (9.5,6.5) {};
\node [mynode] (n11) at (9.5,5.5) {};
\node [mynode] (n12) at (9,2.5) {};
\node [mynode] (n13) at (8,2.5) {};
\node [mynode] (n15) at (7,2.5) {};
\node [mynode] (n16) at (7,1.5) {};
\node [mynode] (n17) at (6,2) {};
\node [mynode] (n18) at (5.75,3.25) {};
\node [mynode] (n19) at (5.25,3.25) {};
\node [mynode] (n20) at (5,2) {};
\node [mynode] (n21) at (6,1) {};
\node [mynode] (n22) at (5,1) {};
\node [mynode] (n23) at (4,1) {};
\node [mynode] (n24) at (4,2) {};
\node [mynode] (n25) at (5.75,3.75) {};
\node [mynode] (n26) at (5.25,3.75) {};
\node [mynode,label=right:{$v_2$}] (n27) at (6,5) {};

\draw (n7) -- (n8) -- (n9) -- (n10) -- (n11) -- (n12) -- (n13);
\draw (n15) -- (n16) -- (n17);
\draw [green] (n17) -- (n18) -- (n19) -- (n20);
\draw (n20) -- (n21) -- (n22) -- (n23) -- (n24);
\draw  (n24) [bend left=15] to (n25); \draw  (n25) -- (n26) -- (n27);
\draw (n13) -- (n15);
\draw [dashed] (n17) -- (n20);
\end{tikzpicture}
\caption{These are the two black paths. The left one is now minimally guarded, but the right one is split again.}
\end{center}
\end{figure}
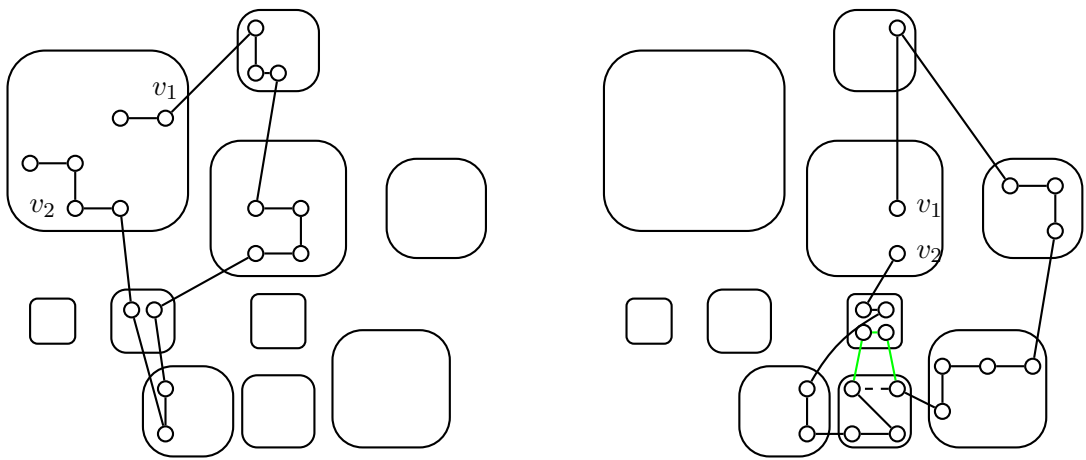

\begin{figure}
\begin{center}
\begin{tikzpicture}[scale=0.4,thick, mynode/.style={draw, circle,inner sep=0cm, minimum size=0.1cm}]
\draw[rounded corners=3mm] (-0.5,5.5) rectangle (3.5,9.5); 

\draw[rounded corners=3mm] (4,4.5) rectangle (7,7.5); 
\draw[rounded corners=3mm] (6.7,0.7) rectangle (9.3,3.3); 
\draw[rounded corners=3mm] (7.9,4.9) rectangle (10.1,7.1); 

\draw[rounded corners=2mm] (2.5,0.5) rectangle (4.5,2.5); 
\draw[rounded corners=2mm] (4.6,8.6) rectangle (6.4,10.4); 
\draw[rounded corners=1mm] (4.7,0.7) rectangle (6.3,2.3); 
\draw[rounded corners=1mm] (1.8,2.8) rectangle (3.2,4.2); 
\draw[rounded corners=1mm] (4.9,2.9) rectangle (6.1,4.1); 
\draw[rounded corners=1mm] (-0.0,3) rectangle (1.0,4.0); 
\node [mynode] (n1) at (2,8) {};
\node [mynode, label=above:{\footnotesize$v_1$}] (n2) at (3,8) {};
\node [mynode] (n3) at (5,10) {};
\node [mynode] (n4) at (5,9) {};
\node [mynode] (n5) at (5.5,9) {};
\node [mynode] (n6) at (5,6) {};
\node [mynode] (n7) at (6,6) {};
\node [mynode] (n27) at (6,5) {};
\node [mynode] (n28) at (5,5) {};
\node [mynode] (n29) at (2.75,3.75) {};
\node [mynode] (n30) at (3,2) {};
\node [mynode] (n31) at (3,1) {};
\node [mynode] (n32) at (2.25,3.75) {};
\node [mynode] (n33) [label=above:{\footnotesize$v_2$}] at (2,6) {};
\node [mynode] (n36) at (1,6) {};
\node [mynode] (n37) at (1,7) {};
\node [mynode] (n38) at (0,7) {};
\draw (n1) -- (n2) -- (n3) -- (n4) -- (n5) -- (n6) -- (n7); \draw  (n7) -- (n27); \draw (n27) -- (n28) -- (n29) -- (n30) -- (n31) -- (n32) -- (n33);
\draw (n36) -- (n37) -- (n38);
\draw (n33) -- (n36);
\end{tikzpicture}\qquad\qquad
\begin{tikzpicture}[scale=0.4,thick, mynode/.style={draw, circle,inner sep=0cm, minimum size=0.1cm}]
\draw[rounded corners=3mm] (-0.5,5.5) rectangle (3.5,9.5); 

\draw[rounded corners=3mm] (4,4.5) rectangle (7,7.5); 
\draw[rounded corners=3mm] (6.7,0.7) rectangle (9.3,3.3); 
\draw[rounded corners=3mm] (7.9,4.9) rectangle (10.1,7.1); 

\draw[rounded corners=2mm] (2.5,0.5) rectangle (4.5,2.5); 
\draw[rounded corners=2mm] (4.6,8.6) rectangle (6.4,10.4); 
\draw[rounded corners=1mm] (4.7,0.7) rectangle (6.3,2.3); 
\draw[rounded corners=1mm] (1.8,2.8) rectangle (3.2,4.2); 
\draw[rounded corners=1mm] (4.9,2.9) rectangle (6.1,4.1); 
\draw[rounded corners=1mm] (-0.0,3) rectangle (1.0,4.0); 
\node [mynode] (n33) [label=above:{\footnotesize $v_1$}] at (2,6) {};
\node [mynode] (n34) at (0.75,3.75) {};
\node [mynode] (n35) at (0.25,3.75) {};
\node [mynode] (n36) [label=above:{\footnotesize$v_2$}] at (1,6) {};
\draw (n33) -- (n34) -- (n35) -- (n36);
\end{tikzpicture}\vspace{\baselineskip}

\begin{tikzpicture}[scale=0.4,thick, mynode/.style={draw, circle,inner sep=0cm, minimum size=0.1cm}]
\draw[rounded corners=3mm] (-0.5,5.5) rectangle (3.5,9.5); 

\draw[rounded corners=3mm] (4,4.5) rectangle (7,7.5); 
\draw[rounded corners=3mm] (6.7,0.7) rectangle (9.3,3.3); 
\draw[rounded corners=3mm] (7.9,4.9) rectangle (10.1,7.1); 

\draw[rounded corners=2mm] (2.5,0.5) rectangle (4.5,2.5); 
\draw[rounded corners=2mm] (4.6,8.6) rectangle (6.4,10.4); 
\draw[rounded corners=1mm] (4.7,0.7) rectangle (6.3,2.3); 
\draw[rounded corners=1mm] (1.8,2.8) rectangle (3.2,4.2); 
\draw[rounded corners=1mm] (4.9,2.9) rectangle (6.1,4.1); 
\draw[rounded corners=1mm] (-0.0,3) rectangle (1.0,4.0); 
\node [mynode,label=right:{\footnotesize  $v_1$}] (n7) at (6,6) {};
\node [mynode] (n8) at (6,10) {};
\node [mynode] (n9) at (8.5,6.5) {};
\node [mynode] (n10) at (9.5,6.5) {};
\node [mynode] (n11) at (9.5,5.5) {};
\node [mynode] (n12) at (9,2.5) {};
\node [mynode] (n13) at (8,2.5) {};
\node [mynode] (n15) at (7,2.5) {};
\node [mynode] (n16) at (7,1.5) {};
\node [mynode] (n17) at (6,2) {};
\node [mynode] (n20) at (5,2) {};
\node [mynode] (n21) at (6,1) {};
\node [mynode] (n22) at (5,1) {};
\node [mynode] (n23) at (4,1) {};
\node [mynode] (n24) at (4,2) {};
\node [mynode] (n25) at (5.75,3.75) {};
\node [mynode] (n26) at (5.25,3.75) {};
\node [mynode,label=right:{\footnotesize$v_2$}] (n27) at (6,5) {};

\draw (n7) -- (n8) -- (n9) -- (n10) -- (n11) -- (n12) -- (n13);
\draw (n15) -- (n16) -- (n17);
\draw (n20) -- (n21) -- (n22) -- (n23) -- (n24);
\draw  (n24) [bend left=15] to (n25); \draw  (n25) -- (n26) -- (n27);
\draw (n13) -- (n15);
\draw (n17) -- (n20);
\end{tikzpicture}\qquad\qquad
\begin{tikzpicture}[scale=0.4,thick, mynode/.style={draw, circle,inner sep=0cm, minimum size=0.1cm}]
\draw[rounded corners=3mm] (-0.5,5.5) rectangle (3.5,9.5); 

\draw[rounded corners=3mm] (4,4.5) rectangle (7,7.5); 
\draw[rounded corners=3mm] (6.7,0.7) rectangle (9.3,3.3); 
\draw[rounded corners=3mm] (7.9,4.9) rectangle (10.1,7.1); 

\draw[rounded corners=2mm] (2.5,0.5) rectangle (4.5,2.5); 
\draw[rounded corners=2mm] (4.6,8.6) rectangle (6.4,10.4); 
\draw[rounded corners=1mm] (4.7,0.7) rectangle (6.3,2.3); 
\draw[rounded corners=1mm] (1.8,2.8) rectangle (3.2,4.2); 
\draw[rounded corners=1mm] (4.9,2.9) rectangle (6.1,4.1); 
\draw[rounded corners=1mm] (-0.0,3) rectangle (1.0,4.0); 
\node [mynode,label=below:{\footnotesize $v_1$}] (n13) at (8,2.5) {};
\node [mynode] (n14) at (8.5,5.5) {};
\node [mynode,label=below:{\footnotesize $v_2$}] (n15) at (7,2.5) {};
\draw (n13) -- (n14) -- (n15);
\end{tikzpicture}\qquad\qquad
\begin{tikzpicture}[scale=0.4,thick, mynode/.style={draw, circle,inner sep=0cm, minimum size=0.1cm}]
\draw[rounded corners=3mm] (-0.5,5.5) rectangle (3.5,9.5); 

\draw[rounded corners=3mm] (4,4.5) rectangle (7,7.5); 
\draw[rounded corners=3mm] (6.7,0.7) rectangle (9.3,3.3); 
\draw[rounded corners=3mm] (7.9,4.9) rectangle (10.1,7.1); 

\draw[rounded corners=2mm] (2.5,0.5) rectangle (4.5,2.5); 
\draw[rounded corners=2mm] (4.6,8.6) rectangle (6.4,10.4); 
\draw[rounded corners=1mm] (4.7,0.7) rectangle (6.3,2.3); 
\draw[rounded corners=1mm] (1.8,2.8) rectangle (3.2,4.2); 
\draw[rounded corners=1mm] (4.9,2.9) rectangle (6.1,4.1); 
\draw[rounded corners=1mm] (-0.0,3) rectangle (1.0,4.0); 
\node [mynode,label=below:{\footnotesize $v_1$}] (n17) at (6,2) {};
\node [mynode] (n18) at (5.75,3.25) {};
\node [mynode] (n19) at (5.25,3.25) {};
\node [mynode,label=below:{\footnotesize$v_2$}] (n20) at (5,2) {};

\draw (n17) -- (n18) -- (n19) -- (n20);
\end{tikzpicture}
\caption{The five base cases that occur.\label{fig:ex:basecases}}
\end{center}
\end{figure}
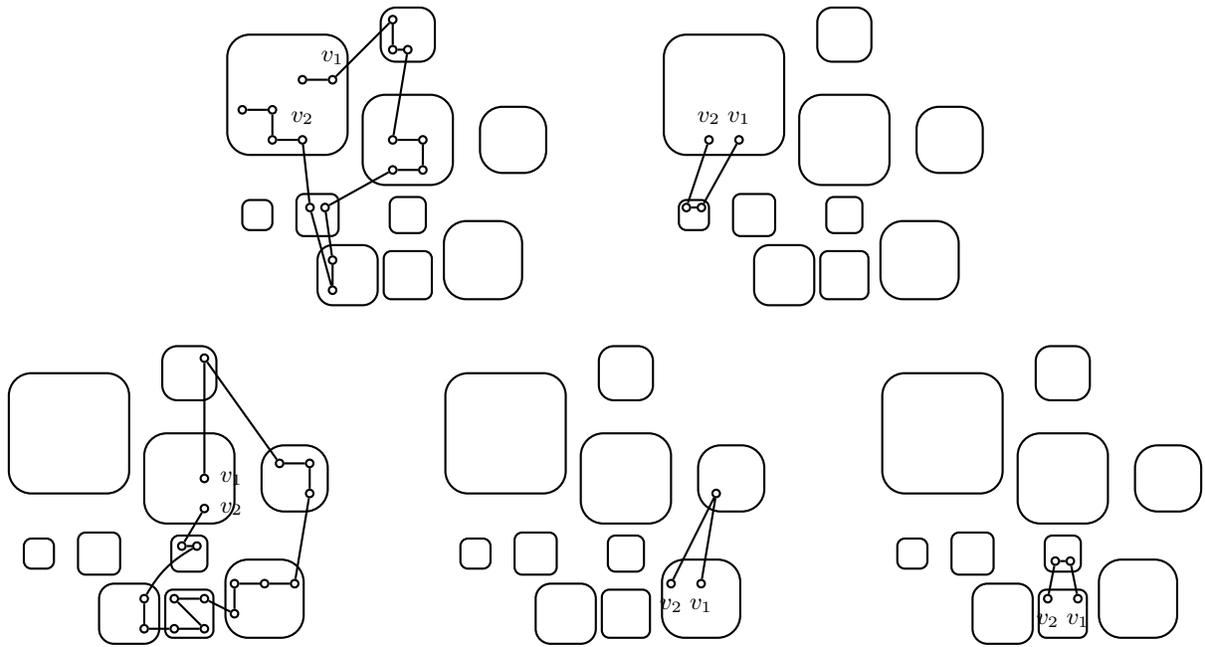

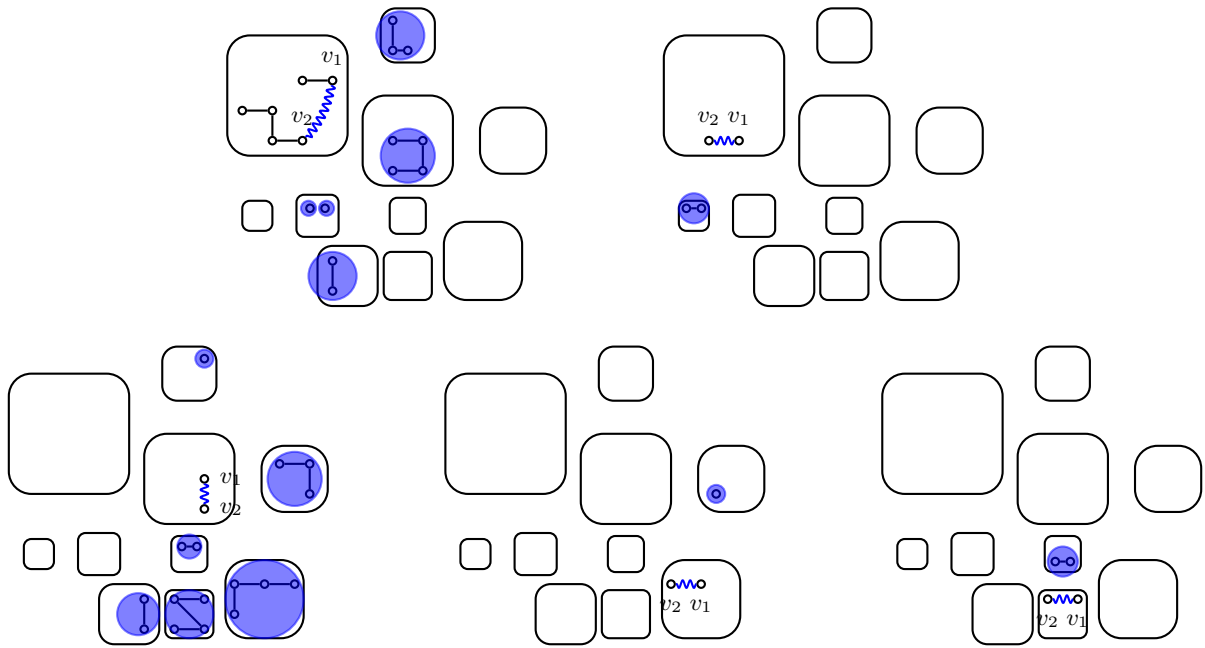
\begin{figure}
\begin{center}
\begin{tikzpicture}[scale=0.4,thick, mynode/.style={draw, circle,inner sep=0cm, minimum size=0.1cm}]
\draw[rounded corners=3mm] (-0.5,5.5) rectangle (3.5,9.5); 

\draw[rounded corners=3mm] (4,4.5) rectangle (7,7.5); 
\draw[rounded corners=3mm] (6.7,0.7) rectangle (9.3,3.3); 
\draw[rounded corners=3mm] (7.9,4.9) rectangle (10.1,7.1); 

\draw[rounded corners=2mm] (2.5,0.5) rectangle (4.5,2.5); 
\draw[rounded corners=2mm] (4.6,8.6) rectangle (6.4,10.4); 
\draw[rounded corners=1mm] (4.7,0.7) rectangle (6.3,2.3); 
\draw[rounded corners=1mm] (1.8,2.8) rectangle (3.2,4.2); 
\draw[rounded corners=1mm] (4.9,2.9) rectangle (6.1,4.1); 
\draw[rounded corners=1mm] (-0.0,3) rectangle (1.0,4.0); 
\node [mynode] (n1) at (2,8) {};
\node [mynode, label=above:{\footnotesize$v_1$}] (n2) at (3,8) {};
\node [mynode] (n3) at (5,10) {};
\node [mynode] (n4) at (5,9) {};
\node [mynode] (n5) at (5.5,9) {};
\node [mynode] (n6) at (5,6) {};
\node [mynode] (n7) at (6,6) {};
\node [mynode] (n27) at (6,5) {};
\node [mynode] (n28) at (5,5) {};
\node [mynode] (n29) at (2.75,3.75) {};
\node [mynode] (n30) at (3,2) {};
\node [mynode] (n31) at (3,1) {};
\node [mynode] (n32) at (2.25,3.75) {};
\node [mynode] (n33) [label=above:{\footnotesize$v_2$}] at (2,6) {};
\node [mynode] (n36) at (1,6) {};
\node [mynode] (n37) at (1,7) {};
\node [mynode] (n38) at (0,7) {};
\draw (n1) -- (n2); \draw (n3) -- (n4) -- (n5);
\draw (n6) -- (n7); \draw  (n7) -- (n27) -- (n28);
\draw (n30) -- (n31);
\draw (n36) -- (n37) -- (n38);
\draw (n33) -- (n36);
\draw [blue,blue,decorate,decoration={snake,segment length=1mm, amplitude=0.5mm},bend left=15] (n2) to (n33);
\draw [blue, fill=blue, fill opacity=0.25,opacity=0.5] (5.25,9.5) circle (0.8cm);
\draw [blue, fill=blue, fill opacity=0.25,opacity=0.5] (5.5,5.5) circle (0.9cm);
\draw [blue, fill=blue, fill opacity=0.25,opacity=0.5] (3,1.5) circle (0.8cm);
\draw [blue, fill=blue, fill opacity=0.25,opacity=0.5] (2.8,3.75) circle (0.25cm);
\draw [blue, fill=blue, fill opacity=0.25,opacity=0.5] (2.2,3.75) circle (0.25cm);
\end{tikzpicture}\qquad\qquad
\begin{tikzpicture}[scale=0.4,thick, mynode/.style={draw, circle,inner sep=0cm, minimum size=0.1cm}]
\draw[rounded corners=3mm] (-0.5,5.5) rectangle (3.5,9.5); 

\draw[rounded corners=3mm] (4,4.5) rectangle (7,7.5); 
\draw[rounded corners=3mm] (6.7,0.7) rectangle (9.3,3.3); 
\draw[rounded corners=3mm] (7.9,4.9) rectangle (10.1,7.1); 

\draw[rounded corners=2mm] (2.5,0.5) rectangle (4.5,2.5); 
\draw[rounded corners=2mm] (4.6,8.6) rectangle (6.4,10.4); 
\draw[rounded corners=1mm] (4.7,0.7) rectangle (6.3,2.3); 
\draw[rounded corners=1mm] (1.8,2.8) rectangle (3.2,4.2); 
\draw[rounded corners=1mm] (4.9,2.9) rectangle (6.1,4.1); 
\draw[rounded corners=1mm] (-0.0,3) rectangle (1.0,4.0); 
\node [mynode] (n33) [label=above:{\footnotesize $v_1$}] at (2,6) {};
\node [mynode] (n34) at (0.75,3.75) {};
\node [mynode] (n35) at (0.25,3.75) {};
\node [mynode] (n36) [label=above:{\footnotesize$v_2$}] at (1,6) {};
\draw (n34) -- (n35);
\draw [blue, fill=blue, fill opacity=0.25,opacity=0.5] (0.5,3.75) circle (0.5cm);
\draw [blue, blue,decorate,decoration={snake,segment length=1mm, amplitude=0.5mm}] (n33) -- (n36);
\end{tikzpicture}\vspace{\baselineskip}

\begin{tikzpicture}[scale=0.4,thick, mynode/.style={draw, circle,inner sep=0cm, minimum size=0.1cm}]
\draw[rounded corners=3mm] (-0.5,5.5) rectangle (3.5,9.5); 

\draw[rounded corners=3mm] (4,4.5) rectangle (7,7.5); 
\draw[rounded corners=3mm] (6.7,0.7) rectangle (9.3,3.3); 
\draw[rounded corners=3mm] (7.9,4.9) rectangle (10.1,7.1); 

\draw[rounded corners=2mm] (2.5,0.5) rectangle (4.5,2.5); 
\draw[rounded corners=2mm] (4.6,8.6) rectangle (6.4,10.4); 
\draw[rounded corners=1mm] (4.7,0.7) rectangle (6.3,2.3); 
\draw[rounded corners=1mm] (1.8,2.8) rectangle (3.2,4.2); 
\draw[rounded corners=1mm] (4.9,2.9) rectangle (6.1,4.1); 
\draw[rounded corners=1mm] (-0.0,3) rectangle (1.0,4.0); 
\node [mynode,label=right:{\footnotesize  $v_1$}] (n7) at (6,6) {};
\node [mynode] (n8) at (6,10) {};
\node [mynode] (n9) at (8.5,6.5) {};
\node [mynode] (n10) at (9.5,6.5) {};
\node [mynode] (n11) at (9.5,5.5) {};
\node [mynode] (n12) at (9,2.5) {};
\node [mynode] (n13) at (8,2.5) {};
\node [mynode] (n15) at (7,2.5) {};
\node [mynode] (n16) at (7,1.5) {};
\node [mynode] (n17) at (6,2) {};
\node [mynode] (n20) at (5,2) {};
\node [mynode] (n21) at (6,1) {};
\node [mynode] (n22) at (5,1) {};
\node [mynode] (n23) at (4,1) {};
\node [mynode] (n24) at (4,2) {};
\node [mynode] (n25) at (5.75,3.75) {};
\node [mynode] (n26) at (5.25,3.75) {};
\node [mynode,label=right:{\footnotesize$v_2$}] (n27) at (6,5) {};

\draw [blue,decorate,decoration={snake,segment length=1mm, amplitude=0.5mm}] (n7) -- (n27);
\draw (n9) -- (n10) -- (n11);
\draw  (n12) -- (n13)-- (n15) -- (n16);
\draw (n17) -- (n20) -- (n21) -- (n22);
\draw (n23) -- (n24);
\draw  (n25) -- (n26); 
\draw [blue, fill=blue, fill opacity=0.25,opacity=0.5] (6,10) circle (0.3cm);
\draw [blue, fill=blue, fill opacity=0.25,opacity=0.5] (9,6) circle (0.9cm);
\draw [blue, fill=blue, fill opacity=0.25,opacity=0.5] (8,2) circle (1.3cm);
\draw [blue, fill=blue, fill opacity=0.25,opacity=0.5] (5.5,1.5) circle (0.8cm);
\draw [blue, fill=blue, fill opacity=0.25,opacity=0.5] (3.8,1.5) circle (0.7cm);
\draw [blue, fill=blue, fill opacity=0.25,opacity=0.5] (5.5,3.75) circle (0.4cm);
\end{tikzpicture}\qquad\qquad
\begin{tikzpicture}[scale=0.4,thick, mynode/.style={draw, circle,inner sep=0cm, minimum size=0.1cm}]
\draw[rounded corners=3mm] (-0.5,5.5) rectangle (3.5,9.5); 

\draw[rounded corners=3mm] (4,4.5) rectangle (7,7.5); 
\draw[rounded corners=3mm] (6.7,0.7) rectangle (9.3,3.3); 
\draw[rounded corners=3mm] (7.9,4.9) rectangle (10.1,7.1); 

\draw[rounded corners=2mm] (2.5,0.5) rectangle (4.5,2.5); 
\draw[rounded corners=2mm] (4.6,8.6) rectangle (6.4,10.4); 
\draw[rounded corners=1mm] (4.7,0.7) rectangle (6.3,2.3); 
\draw[rounded corners=1mm] (1.8,2.8) rectangle (3.2,4.2); 
\draw[rounded corners=1mm] (4.9,2.9) rectangle (6.1,4.1); 
\draw[rounded corners=1mm] (-0.0,3) rectangle (1.0,4.0); 
\node [mynode,label=below:{\footnotesize $v_1$}] (n13) at (8,2.5) {};
\node [mynode] (n14) at (8.5,5.5) {};
\node [mynode,label=below:{\footnotesize $v_2$}] (n15) at (7,2.5) {};
\draw [blue,blue,decorate,decoration={snake,segment length=1mm, amplitude=0.5mm}] (n13)  -- (n15);
\draw [blue, fill=blue, fill opacity=0.25,opacity=0.5] (8.5,5.5) circle (0.3cm);
\end{tikzpicture}\qquad\qquad
\begin{tikzpicture}[scale=0.4,thick, mynode/.style={draw, circle,inner sep=0cm, minimum size=0.1cm}]
\draw[rounded corners=3mm] (-0.5,5.5) rectangle (3.5,9.5); 

\draw[rounded corners=3mm] (4,4.5) rectangle (7,7.5); 
\draw[rounded corners=3mm] (6.7,0.7) rectangle (9.3,3.3); 
\draw[rounded corners=3mm] (7.9,4.9) rectangle (10.1,7.1); 

\draw[rounded corners=2mm] (2.5,0.5) rectangle (4.5,2.5); 
\draw[rounded corners=2mm] (4.6,8.6) rectangle (6.4,10.4); 
\draw[rounded corners=1mm] (4.7,0.7) rectangle (6.3,2.3); 
\draw[rounded corners=1mm] (1.8,2.8) rectangle (3.2,4.2); 
\draw[rounded corners=1mm] (4.9,2.9) rectangle (6.1,4.1); 
\draw[rounded corners=1mm] (-0.0,3) rectangle (1.0,4.0); 
\node [mynode,label=below:{\footnotesize $v_1$}] (n17) at (6,2) {};
\node [mynode] (n18) at (5.75,3.25) {};
\node [mynode] (n19) at (5.25,3.25) {};
\node [mynode,label=below:{\footnotesize$v_2$}] (n20) at (5,2) {};
\draw (n18) -- (n19);
\draw [blue,blue,decorate,decoration={snake,segment length=1mm, amplitude=0.5mm}] (n17) -- (n20);
\draw [blue, fill=blue, fill opacity=0.25,opacity=0.5] (5.5,3.25) circle (0.5cm);

\end{tikzpicture}
\caption{We get that the blue edges cost at most the same as the black edges in Figure~\ref{fig:ex:basecases}, and that summing the widths for all blue components costs at most the cost of the black edges in Figure~\ref{fig:ex:basecases} as well.}
\end{center}
\end{figure}

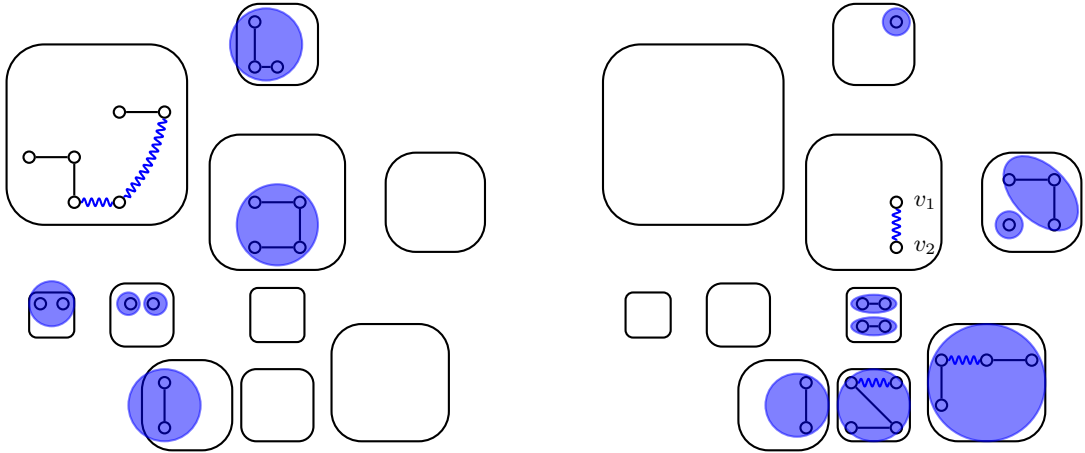
\begin{figure}
\begin{center}
\begin{tikzpicture}[scale=0.6,thick, mynode/.style={draw, circle,inner sep=0cm, minimum size=0.15cm}]
\draw[rounded corners=5mm] (-0.5,5.5) rectangle (3.5,9.5); 

\draw[rounded corners=4mm] (4,4.5) rectangle (7,7.5); 
\draw[rounded corners=4mm] (6.7,0.7) rectangle (9.3,3.3); 
\draw[rounded corners=4mm] (7.9,4.9) rectangle (10.1,7.1); 

\draw[rounded corners=4mm] (2.5,0.5) rectangle (4.5,2.5); 
\draw[rounded corners=3mm] (4.6,8.6) rectangle (6.4,10.4); 
\draw[rounded corners=2mm] (4.7,0.7) rectangle (6.3,2.3); 
\draw[rounded corners=2mm] (1.8,2.8) rectangle (3.2,4.2); 
\draw[rounded corners=1mm] (4.9,2.9) rectangle (6.1,4.1); 
\draw[rounded corners=1mm] (-0.0,3) rectangle (1.0,4.0); 
\node [mynode] (n1) at (2,8) {};
\node [mynode] (n2) at (3,8) {};
\node [mynode] (n3) at (5,10) {};
\node [mynode] (n4) at (5,9) {};
\node [mynode] (n5) at (5.5,9) {};
\node [mynode] (n6) at (5,6) {};
\node [mynode] (n7) at (6,6) {};
\node [mynode] (n27) at (6,5) {};
\node [mynode] (n28) at (5,5) {};
\node [mynode] (n29) at (2.75,3.75) {};
\node [mynode] (n30) at (3,2) {};
\node [mynode] (n31) at (3,1) {};
\node [mynode] (n32) at (2.25,3.75) {};
\node [mynode] (n33) at (2,6) {};
\node [mynode] (n34) at (0.75,3.75) {};
\node [mynode] (n35) at (0.25,3.75) {};
\node [mynode] (n36) at (1,6) {};
\node [mynode] (n37) at (1,7) {};
\node [mynode] (n38) at (0,7) {};
\draw (n1) -- (n2); \draw (n3) -- (n4) -- (n5);
\draw (n6) -- (n7); \draw  (n7) -- (n27) -- (n28);
\draw (n30) -- (n31);
\draw (n36) -- (n37) -- (n38);
\draw [blue,blue,decorate,decoration={snake,segment length=1mm, amplitude=0.5mm},bend left=15] (n33) -- (n36);
\draw [blue,blue,decorate,decoration={snake,segment length=1mm, amplitude=0.5mm},bend left=15] (n2) to (n33);
\draw [blue, fill=blue, fill opacity=0.25,opacity=0.5] (5.25,9.5) circle (0.8cm);
\draw [blue, fill=blue, fill opacity=0.25,opacity=0.5] (5.5,5.5) circle (0.9cm);
\draw [blue, fill=blue, fill opacity=0.25,opacity=0.5] (3,1.5) circle (0.8cm);
\draw [blue, fill=blue, fill opacity=0.25,opacity=0.5] (2.8,3.75) circle (0.25cm);
\draw [blue, fill=blue, fill opacity=0.25,opacity=0.5] (2.2,3.75) circle (0.25cm);
\draw [blue, fill=blue, fill opacity=0.25,opacity=0.5] (0.5,3.75) circle (0.5cm);
\end{tikzpicture}\qquad\qquad
\begin{tikzpicture}[scale=0.6,thick, mynode/.style={draw, circle,inner sep=0cm, minimum size=0.15cm}]
\draw[rounded corners=5mm] (-0.5,5.5) rectangle (3.5,9.5); 

\draw[rounded corners=4mm] (4,4.5) rectangle (7,7.5); 
\draw[rounded corners=4mm] (6.7,0.7) rectangle (9.3,3.3); 
\draw[rounded corners=4mm] (7.9,4.9) rectangle (10.1,7.1); 

\draw[rounded corners=4mm] (2.5,0.5) rectangle (4.5,2.5); 
\draw[rounded corners=3mm] (4.6,8.6) rectangle (6.4,10.4); 
\draw[rounded corners=2mm] (4.7,0.7) rectangle (6.3,2.3); 
\draw[rounded corners=2mm] (1.8,2.8) rectangle (3.2,4.2); 
\draw[rounded corners=1mm] (4.9,2.9) rectangle (6.1,4.1); 
\draw[rounded corners=1mm] (-0.0,3) rectangle (1.0,4.0); 
\node [mynode,label=right:{\footnotesize  $v_1$}] (n7) at (6,6) {};
\node [mynode] (n8) at (6,10) {};
\node [mynode] (n9) at (8.5,6.5) {};
\node [mynode] (n10) at (9.5,6.5) {};
\node [mynode] (n11) at (9.5,5.5) {};
\node [mynode] (n12) at (9,2.5) {};
\node [mynode] (n13) at (8,2.5) {};
\node [mynode] (n14) at (8.5,5.5) {};
\node [mynode] (n15) at (7,2.5) {};
\node [mynode] (n16) at (7,1.5) {};
\node [mynode] (n17) at (6,2) {};
\node [mynode] (n18) at (5.75,3.25) {};
\node [mynode] (n19) at (5.25,3.25) {};
\node [mynode] (n20) at (5,2) {};
\node [mynode] (n21) at (6,1) {};
\node [mynode] (n22) at (5,1) {};
\node [mynode] (n23) at (4,1) {};
\node [mynode] (n24) at (4,2) {};
\node [mynode] (n25) at (5.75,3.75) {};
\node [mynode] (n26) at (5.25,3.75) {};
\node [mynode,label=right:{\footnotesize$v_2$}] (n27) at (6,5) {};

\draw [blue,decorate,decoration={snake,segment length=1mm, amplitude=0.5mm}] (n7) -- (n27);
\draw (n9) -- (n10) -- (n11);
\draw  (n12) -- (n13);
\draw [blue,decorate,decoration={snake,segment length=1mm, amplitude=0.5mm}] (n13) -- (n15);
\draw (n15) -- (n16);
\draw [blue,blue,decorate,decoration={snake,segment length=1mm, amplitude=0.5mm}] (n17) -- (n20);
\draw (n20) -- (n21) -- (n22);
\draw (n23) -- (n24);
\draw (n18) -- (n19);
\draw  (n25) -- (n26); 
\draw [blue, fill=blue, fill opacity=0.25,opacity=0.5] (6,10) circle (0.3cm);
\draw[rotate around={45:(9.2,6.2)},blue, fill=blue, fill opacity=0.25,opacity=0.5] (9.2,6.2) ellipse (0.6cm and 1.0cm);
\draw [blue, fill=blue, fill opacity=0.25,opacity=0.5] (8,2) circle (1.3cm);
\draw [blue, fill=blue, fill opacity=0.25,opacity=0.5] (5.5,1.5) circle (0.8cm);
\draw [blue, fill=blue, fill opacity=0.25,opacity=0.5] (3.8,1.5) circle (0.7cm);
\draw [blue, fill=blue, fill opacity=0.25,opacity=0.5] (5.5,3.75) ellipse (0.5cm and 0.2cm);
\draw [blue, fill=blue, fill opacity=0.25,opacity=0.5] (5.5,3.25) ellipse (0.5cm and 0.2cm);
\draw [blue, fill=blue, fill opacity=0.25,opacity=0.5] (8.5,5.5) circle (0.3cm);
\end{tikzpicture}
\caption{The result of the induction for the paths in Figure~\ref{fig:ex:3}}.
\end{center}
\end{figure}

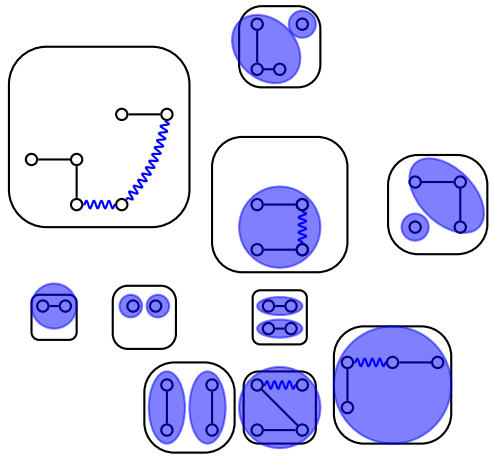
\begin{figure}
\begin{center}
\begin{tikzpicture}[scale=0.6,thick, mynode/.style={draw, circle,inner sep=0cm, minimum size=0.15cm}]
\draw[rounded corners=5mm] (-0.5,5.5) rectangle (3.5,9.5); 

\draw[rounded corners=4mm] (4,4.5) rectangle (7,7.5); 
\draw[rounded corners=4mm] (6.7,0.7) rectangle (9.3,3.3); 
\draw[rounded corners=4mm] (7.9,4.9) rectangle (10.1,7.1); 

\draw[rounded corners=4mm] (2.5,0.5) rectangle (4.5,2.5); 
\draw[rounded corners=3mm] (4.6,8.6) rectangle (6.4,10.4); 
\draw[rounded corners=2mm] (4.7,0.7) rectangle (6.3,2.3); 
\draw[rounded corners=2mm] (1.8,2.8) rectangle (3.2,4.2); 
\draw[rounded corners=1mm] (4.9,2.9) rectangle (6.1,4.1); 
\draw[rounded corners=1mm] (-0.0,3) rectangle (1.0,4.0); 
\node [mynode] (n7) at (6,6) {};
\node [mynode] (n1) at (2,8) {};
\node [mynode] (n2) at (3,8) {};
\node [mynode] (n3) at (5,10) {};
\node [mynode] (n4) at (5,9) {};
\node [mynode] (n5) at (5.5,9) {};
\node [mynode] (n6) at (5,6) {};
\node [mynode] (n7) at (6,6) {};
\node [mynode] (n8) at (6,10) {};
\node [mynode] (n9) at (8.5,6.5) {};
\node [mynode] (n10) at (9.5,6.5) {};
\node [mynode] (n11) at (9.5,5.5) {};
\node [mynode] (n12) at (9,2.5) {};
\node [mynode] (n13) at (8,2.5) {};
\node [mynode] (n14) at (8.5,5.5) {};
\node [mynode] (n15) at (7,2.5) {};
\node [mynode] (n16) at (7,1.5) {};
\node [mynode] (n17) at (6,2) {};
\node [mynode] (n18) at (5.75,3.25) {};
\node [mynode] (n19) at (5.25,3.25) {};
\node [mynode] (n20) at (5,2) {};
\node [mynode] (n21) at (6,1) {};
\node [mynode] (n22) at (5,1) {};
\node [mynode] (n23) at (4,1) {};
\node [mynode] (n24) at (4,2) {};
\node [mynode] (n25) at (5.75,3.75) {};
\node [mynode] (n26) at (5.25,3.75) {};
\node [mynode] (n27) at (6,5) {};
\node [mynode] (n28) at (5,5) {};
\node [mynode] (n29) at (2.75,3.75) {};
\node [mynode] (n30) at (3,2) {};
\node [mynode] (n31) at (3,1) {};
\node [mynode] (n32) at (2.25,3.75) {};
\node [mynode] (n33) at (2,6) {};
\node [mynode] (n34) at (0.75,3.75) {};
\node [mynode] (n35) at (0.25,3.75) {};
\node [mynode] (n36) at (1,6) {};
\node [mynode] (n37) at (1,7) {};
\node [mynode] (n38) at (0,7) {};
\node [mynode] (n27) at (6,5) {};

\draw (n1) -- (n2); \draw (n36) -- (n37) -- (n38); \draw (n6) -- (n7); \draw (n27)--(n28);
\draw (n3) -- (n4) -- (n5); \draw (n30) -- (n31);
\draw [blue,decorate,decoration={snake,segment length=1mm, amplitude=0.5mm}] (n7) -- (n27);
\draw (n9) -- (n10) -- (n11);
\draw  (n12) -- (n13);
\draw [blue,decorate,decoration={snake,segment length=1mm, amplitude=0.5mm}] (n13) -- (n15);
\draw (n15) -- (n16);
\draw [blue,blue,decorate,decoration={snake,segment length=1mm, amplitude=0.5mm}] (n17) -- (n20);
\draw (n20) -- (n21) -- (n22);
\draw (n23) -- (n24);
\draw (n18) -- (n19);
\draw (n34) -- (n35);
\draw  (n25) -- (n26); 
\draw [blue, fill=blue, fill opacity=0.25,opacity=0.5] (6,10) circle (0.3cm);
\draw[rotate around={45:(9.2,6.2)},blue, fill=blue, fill opacity=0.25,opacity=0.5] (9.2,6.2) ellipse (0.6cm and 1.0cm);
\draw [blue, fill=blue, fill opacity=0.25,opacity=0.5] (8,2) circle (1.3cm);
\draw [blue, fill=blue, fill opacity=0.25,opacity=0.5] (5.5,1.5) circle (0.9cm);
\draw [blue, fill=blue, fill opacity=0.25,opacity=0.5] (5.5,3.75) ellipse (0.5cm and 0.2cm);
\draw [blue, fill=blue, fill opacity=0.25,opacity=0.5] (5.5,3.25) ellipse (0.5cm and 0.2cm);
\draw [blue, fill=blue, fill opacity=0.25,opacity=0.5] (8.5,5.5) circle (0.3cm);

\draw [blue,blue,decorate,decoration={snake,segment length=1mm, amplitude=0.5mm},bend left=15] (n33) -- (n36);
\draw [blue,blue,decorate,decoration={snake,segment length=1mm, amplitude=0.5mm},bend left=15] (n2) to (n33);
\draw [rotate around={45:(5.2,9.45)},blue, fill=blue, fill opacity=0.25,opacity=0.5] (5.2,9.45) ellipse (0.65cm and 0.85cm);
\draw [blue, fill=blue, fill opacity=0.25,opacity=0.5] (5.5,5.5) circle (0.9cm);
\draw [blue, fill=blue, fill opacity=0.25,opacity=0.5] (3,1.5) ellipse (0.4cm and 0.8cm);
\draw [blue, fill=blue, fill opacity=0.25,opacity=0.5] (3.9,1.5) ellipse (0.4cm and 0.8cm);
\draw [blue, fill=blue, fill opacity=0.25,opacity=0.5] (2.8,3.75) circle (0.25cm);
\draw [blue, fill=blue, fill opacity=0.25,opacity=0.5] (2.2,3.75) circle (0.25cm);
\draw [blue, fill=blue, fill opacity=0.25,opacity=0.5] (0.5,3.75) circle (0.5cm);
\end{tikzpicture}
\caption{Final set $F'$ and connected components of $F' \cup \bFin$.\label{fig:ex:last}}
\end{center}
\end{figure}

\clearpage

\subsubsection{Wrapping Up}

\reductionforesttotree
\begin{proof}
By Observation~\ref{obs:f-assumptions}, we know that by accepting a factor of $2$ in the cost, we can assume that the connected components $F_1,\ldots,F_q$ of $\REFF$ are node disjoint cycles and that $V[\REFF]=V[\ALG]$ equals the set of all terminals. The connected components of $\ALG$ are $T_1,\ldots,T_p$. Recall that we use $\xi(v)$ for the index of the component $T_j$ that $v$ lies in, and we use $\Fin:=\cset{e=\{u,v\} \in  F}{\xi(u)=\xi(v)}$ and $\Fbetw:=\cset{e=\{u,v\} \in F}{\xi(u)\neq\xi(v)}$ for any edge set $F$ for the set of edges within components of $\ALG$ or between them, respectively. Also, if an edge set $F'$ in $G$ satisfies $V[F'] \subseteq V[T_j]$ for a $j\in\{1,\ldots,p\}$, then we use $\xi(F') = j$.
We want to replace all cycles $F_i$ by sets $F_i'$ which satisfy $F_i' = (F_i')_{\circlearrowright}$ while keeping the solution feasible and within a constant factor of $\phi(\REFF)$. Let $F=F_i$ be one of the cycles.

Let $j^\ast := \max_{v \in V[F]} \xi(v)$ be the index of the component with the highest width among the components that $F$ visits. There have to be at least two vertices on $F$ from $T_{j^\ast}$ (every vertex is a terminal by our assumption, and since the cycles are disjoint, a lone vertex would not be connected to its mate, but $\REFF$ is a feasible solution). If the two vertices are adjacent in $F$, then we can delete the edge that connects them and obtain a path that satisfies the preconditions of Lemma~\ref{lem:guardedcycles}. We get an edge set $F'$. Otherwise, let $v_1$ and $v_2$ be two vertices from $T_{j^\ast}$ that are not connected by an edge in $F$. 
Then the cycle is partitioned into two paths, both with endpoints $v_1$ and $v_2$, that both satisfy the preconditions of Lemma~\ref{lem:guardedcycles}. 
In this case, we get two solutions $F_l'$ and $F_r'$ and set $F' := F_l' \cup F_r'$.
Notice that, either way, we get a set of edges $F'$ on the vertices $V[F]$ inducing connected components $F_1',\ldots,F_x'$ of $\Fin \cup F'$ with the following properties:
\begin{enumerate}
\item $\ALG$ is \edgesetswap-optimal with respect to $F'$
\item For all $F_\ell'$, there exists an index $j$ such that $V[F_\ell']\subseteq V[T_j]$ (thus, $\xi(F_{\ell}') = j$). 
When $F' = F_l' \cup F_r'$, then notice that the connected components of $F_l'$ and $F_r'$ are disjoint with the exception of those containing $v_1$ and $v_2$. Thus, no components with vertices from different $T_j$ will get connected.
\item There is only one $F_{\ell}'$ with $\xi(F_{\ell}') = j^\ast$, assume w.l.o.g. that  $\xi(F_{1}')=j^\ast$.
When $F' = F_l' \cup F_r'$, then notice that all occurrences of vertices from $T_{j^\ast}$ are connected to $v_1$ and $v_2$ in either $F_l'$ or $F_r'$, thus, they are all in the same connected component of $F'$.
\item It holds that $d(F') \le d(\Fbetw)$ and $\sum_{\ell=2}^{x} w(T_{\xi(F_{\ell}')}) \le c\cdot d(\Fbetw)$.
When $F' = F_l' \cup F_r'$, notice that $d(F')=d(F_l')+d(F_r')\le d(\Fbetw)$, and that $\sum_{\ell=2}^{x} w(T_{\xi(F_{\ell}')})$, which does not include the components with $v_1$ and $v_2$, can be split according to the \lq side\rq\ of the cycle that the components belong to.
\end{enumerate}
The solution $\REFF'$ that arises from substituting $\Fbetw$ by $F'$ is not necessarily feasible because $\Fin \cup F'$ can consist of multiple connected components. We need to transform $F'$ such that all terminal pairs in $\Fin \cup F'$ are connected. Notice that a terminal pair $u, \bar u$ always satisfies $\xi(u)=\xi(\bar u)$ because $\ALG$ is feasible, \ie we do not need to connect connected components with vertices from different $T_j$.
Furthermore, all vertices in $V[F]$ from $T_{j^\ast}$ are already connected because of 3.
Fix a $j < j^\ast$ and consider all connected components $F_\ell'$ with $\xi(F_\ell')=j$. Notice that $j < j^\ast$ implies that the widths of these components are part of $\sum_{\ell=2}^{x} w(T_{\xi(F_{\ell}')})$. Start with an arbitrary $F_\ell'$. If there is a terminal $u \in F_\ell'$ whose partner $\bar u$ is in $F_{\ell'}'$, $\ell'\neq \ell$, then connect $u$ to $\bar u$. Since $u, \bar u \in T_j$, their distance is at most $w(T_j)$. Since $w(T_{\xi(F_{\ell'})}) = w(T_j)$, the contribution of $F_{\ell'}$ to the width sum is large enough to cover the connection cost. Now, $F_\ell'$ and $F_{\ell'}'$ are merged into one component, we keep calling it $F_\ell'$. Repeat the process until all terminals in $F_\ell'$ are connected to their partner, while always spending a connection cost that is bounded by the contribution of the component that gets merged into $F_\ell'$. When $F_\ell'$ is done, pick a component that was not merged and continue in the same fashion. Repeat until all components are merged or processed. In the end, $F'$ is a feasible solution, and the money spent for the additional edges is bounded by $\sum_{\ell=2}^{x} w(T_{\xi(F_{\ell}')}) \le c\cdot d(\Fbetw)$. Thus, the cost of the new solution is at most $(1+c)\cdot d(\Fbetw)$. 


We process all $F_i$ with $F_i \neq (F_i)_{\circlearrowright}$ in this manner to obtain a solution $\REFF'$ with $\REFF' = \REFF_{\circlearrowright}'$.
Notice that $\ALG$ is \edgeedge and \edgeset swap-optimal with respect to $\REFF'$. This holds for the new edges because they are from $G$ and $\ALG$ is swap-optimal with respect to $G$, and it holds for the edges that we get from Lemma~\ref{lem:guardedcycles} by property 1.

Thus, we have found $\REFF'$ with the necessary properties.
\end{proof}

We can now apply Corollary~\ref{treecaseresult} to bound the cost of $\ALG$.

\begin{corollary}\label{cor:locality-gap}
Let $G=(V,E)$ be a complete graph, let $d: E \to \Rp$ be a metric that assigns a cost $d_e$ to every edge $e \in E$ and   let $\terms \subseteq V\times V$ be a terminal set. 
  Let $\ALG, \REFF \subseteq E$ be two feasible Steiner Forest solutions for $(G,\terms)$.
  Furthermore, suppose that $\ALG$ is \edgeedge, \edgeset and \pathset swap-optimal with respect to $E$ and $\phi$, that $\ALG$ is $c$-approximate connecting move optimal and that $\ALG$ only uses edges between terminals. Then $d(\ALG') \le 23(1+c) \cdot d(\REFF)$.
\end{corollary}
\begin{proof}
Lemma~\ref{lem:main-forest:reduction-to-tree} ensures that there is a solution $\REFF'$ with $d(\REFF') \le 2(1+c)\cdot d(\REFF)$ that satisfies $\REFF'_{\circlearrowright}=\REFF'$.
Every connected component $A_j$ of $\ALG$ can now be treated separately by using Corollary~\ref{treecaseresult} on $A_j$ and the part of $\REFF'$ that falls into $A_j$. By combining the conclusions for all connected components, we get that
\[
d(\ALG') \le 11.5 \phi(\REFF') \le 23 (1+c) \cdot d(\REFF)
\]
for any feasible solution $\REFF$. 
\end{proof}

Theorem~\ref{thm:main} follows directly. Since the $2$-approximation for $k$-MST~\cite{G05} can be adapted to the weighted case~\cite{G16}, $c=2$ is possible and we can achieve an approximation guarantee of $69$.

{\small 
\bibliography{referencesSFL,bibonline}
\bibliographystyle{amsalpha}
}

\newpage
\appendix

\appendix
\section{Notes on simpler local search algorithms}

\subsection{Adding an edge and removing a constant number of edges}\label{appendix:introexample}
Let $\ell$ and $k < \ell$ be integers and consider Figure~\ref{fig:no-imp-edge-edge}. Notice that adding a single edge and removing $k$ edges does not improve the solution. However, the current solution costs more than $\ell^2/k$ and the optimal solution costs less than $2\ell$, which is a factor of $\ell/(2k)$ better.
		
\begin{figure}[htbp]
\centering
\begin{tikzpicture}
  \tikzstyle{node}=[circle,shade,top color=gray!30,bottom color=gray!70,draw=gray]
  \tikzstyle{optedge}=[very thick,blue,-,dashed]
  \tikzstyle{algedge}=[very thick,black,-,solid]
	\tikzstyle{both}=[postaction={draw,blue!80,dash pattern= on 3pt off 5pt,dash phase=4pt,very thick},black,dash pattern= on 3pt off 5pt,very thick]
  \begin{scope}[yshift=2.5mm,scale=0.6] 
   \matrix (v) [matrix of math nodes, nodes={circle,shade,top color=gray!30,bottom color=gray!70,draw=gray,inner sep=1pt,anchor=base,text depth=.5ex,text height=2ex,minimum width=1.2em,text centered},column sep=5mm,row sep=4.5mm] {
    s_1 & s_2 & t_2 & s_3 & t_3 & |[draw=none,fill=none,shade=none,scale=2.0]| \dots & s_\ell & t_\ell & t_1 \\
   };
   \draw[algedge] (v-1-1) edge node[auto,black] {$\frac{\ell}{k}$} (v-1-2);
   \draw[algedge,transform canvas={yshift=0.6pt}] (v-1-2) -- node[auto] {$1$} (v-1-3);
	 \draw[optedge,transform canvas={yshift=-0.6pt}] (v-1-2) -- (v-1-3);	
	
   \draw[algedge] (v-1-3) edge node[auto,black] {$\frac{\ell}{k}$} (v-1-4);
   \draw[algedge,transform canvas={yshift=0.6pt}] (v-1-4) -- node[auto] {$1$} (v-1-5);
	 \draw[optedge,transform canvas={yshift=-0.6pt}] (v-1-4) -- (v-1-5);	
   \draw[algedge,path fading=east] (v-1-5) -- ++(0.6,0);
   \draw[algedge,path fading=west] (v-1-7) -- ++(-0.6,0);
   \draw[algedge,transform canvas={yshift=0.6pt}] (v-1-7) -- node[auto] {$1$} (v-1-8);
	 \draw[optedge,transform canvas={yshift=-0.6pt}] (v-1-7) -- (v-1-8);	
   \draw[algedge] (v-1-8) edge node[auto,black] {$\frac{\ell}{k}$} (v-1-9);
   \draw[optedge] (v-1-1) -- ++(0,-1) -- node[auto,swap,black,above,pos=0.55] {$\ell$} ($(v-1-9)+(0,-1)$) -- (v-1-9);
  \end{scope}
 \end{tikzpicture}
\caption{A bad example for \edgesetswap{}s that remove a constant number of edges.\label{fig:no-imp-edge-edge}}
\end{figure}
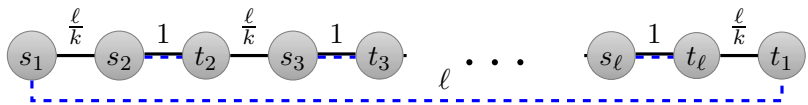


\subsection{Regular graphs with high girth and low degree}\label{sec:bad-gap}

Assume that $G$ is a degree-$3$ graph with girth $g = c \log n$ like the graph used in Chen \etal~\cite{CRV10}. Such graphs can be constructed, see~\cite{B98}. Select a spanning tree $\REFF$ in $G$ which will be the optimal solution. Let $E'$ be the non-tree edges, notice that $|E'| \ge n/2$, and let $M$ be a maximum matching in $E'$. Because of the degrees, we know that $|M| \ge n/10$. The endpoints of the edges in $M$ form the terminal pairs $\terms$. Set the length of all edges in $\REFF$ to $1$ and the length of the remaining edges to $g/4$. 
The solution $\REFF$ is feasible and costs $n-1$. The solution $M$ costs $\Omega(\log n)$. 

\msnote{Can there be a true \pathsetswap, adding and deleting multiple edges at the same time?}
Assume we want to remove an edge $e = \{v,w\} \in M$ and our swap even allows us to add a path to reconnect $v$ and $w$ (in the graph where $M\backslash \{e\}$ is contracted). Let $P$ be such a path. Since $M$ is a matching, at most every alternating edge on $P$ is in $M$. Thus, we have to add $|P|/2-1 \ge g/2 -1$ edges of length one at a total cost that is larger than the cost $g/4$ of $e$. Thus, no $d$-improving swap of this type exists (note that, in particular, \pathsetswap{}s are not $d$-improving for $M$). As a consequence, any oblivious local search with constant locality gap needs to sport a move that removes edges from multiple components of the current solution.
In order to restrict to local moves that only remove edges from a single component, we therefore introduced the potential $\phi$.

\section{Making the Algorithm Polynomial}
\label{sec:polytime}

So far, we have shown that any locally optimal solution is also within a
constant factor of the global optimum. In order to ensure that the local
search can find such a local optimum in polynomial time, two issues have
to be addressed. 

First, we need to show that each improving move can be carried out in polynomial time (or it can be decided in polynomial time that no such move exists). While it is easy to see that improving \edgeedge, \edgeset, and \pathset swaps can be found in polynomial time, finding an improving connecting move is NP-hard in general. However, as we saw in Section~\ref{sec:happytreepacking}, 
it is sufficient to restrict the neighborhood of the local search to $c$-approximate connecting moves. In Section~\ref{sec:howto-treemove}, we show that the task of finding an approximate connecting move reduces to approximating the weighted $k$-MST problem. In Section~\ref{sec:wkmst} we discuss constant factor approximation algorithms for this problem.
In particular, we get the following theorem:

\begin{theorem}\label{thm:alg-approx-tree-opt}
 Let $\varepsilon>0$. There exists a polynomial time algorithm, called \emph{Improving-Connecting-Move}, such that given a Steiner forest $\ALG$ and a metric distance $d:E\rightarrow \Rp$ on the edges then either: (i) the algorithm finds an improving connecting move with respect to $\phi$, or (ii) it guarantees that there is no $c(1+\varepsilon)$-approximate connecting move, that is, for every tree $T$ of $G_{\ALG}$ it holds that\msnote{Here, $5$ is $c$, too}
\[
 \sum_{e\in \ALG)} d(e) \ge c(1+\varepsilon)\cdot \left(\sum_{i\in V_{\ALG}} w(A_i) - \max_{i\in V_{\ALG}} w(A_i) \right), 
\]
 where $\{A_1,\ldots,A_p\}$ is the set of connected components of $\ALG$.
\end{theorem}

The second thing we need to show for guaranteeing polynomiality of the local search is that the total number of improving moves is bounded by a polynomial in the input size. This can easily be achieved done via a standard rounding technique incurring only a loss of a factor of $1 + \varepsilon$ over the original guarantee for local optima, for an arbitrarily small $\varepsilon > 0$; see Section~\ref{sec:convergence} for details.

We finally get the following theorem:

\begin{theorem}
For every $\varepsilon > 0$, there is a local search algorithm that computes in polynomial time a solution $\ALG$ to Steiner Forest such that $d(\ALG) \leq (1+\varepsilon) 69 \cdot \OPT$.\jmcom{Replace by 69 if $c = 2$.}
\end{theorem}

\subsection{How to ensure approximate connecting move optimality}\label{sec:howto-treemove}

Assume that we are given an algorithm \emph{Tree-Approx} that computes a $c$-approximation for the following minimization problem. We call the problem \emph{weighted (rooted) $k$-MST problem}, and approximating it is further discussed in Section~\ref{sec:wkmst}.

\begin{quotation}
Given $G=(V,E)$ with a root $r$, a metric $d: V\times V \to \mathbb{R}^+$, a function $\gamma : V\to \mathbb{R}^+$ with $\gamma(r)=0$, and a lower bound $\Gamma$, find a tree $T$ in $G$ with $r \in V[T]$ and $\sum_{v \in V[T]} \gamma(v)\ge \Gamma$ that minimizes $\sum_{e \in T} d(e)$.
\end{quotation}

We see how to use \emph{Tree-Approx} to ensure $((1+\varepsilon)\cdot c)$-approximate connecting move optimality.
We apply \emph{Tree-Approx} to $G_\ALG^{\text{all}}$. Recall that the vertices in $G_\ALG^{\text{all}}$ are $\{1,\ldots,p\}$, 
 corresponding to the components $A_1,\ldots,A_p$ of our solution. We try $|V(G_\ALG^{\text{all}})|$ possibilities for the component with the largest width in the connecting move.
After choosing the largest component to be the one with index $i$, we delete all vertices from $G_\ALG^{\text{all}}$ with indices larger than $i$.
Then we set $\gamma(i) := 0$ and $\gamma(j) := w(A_j)$ for $j < i$.
We can collect prices between $w_{\min} := \min \{w(A_i) \mid i \in \{1, \dots, p\}, w(A_i) > 0\}$, the smallest strictly positive width of any component, and $\sum_{j=1}^{i-1} w(A_j) < p \cdot w(A_p)$. 
We call \emph{Tree-Approx} for $\Gamma = (1+\varepsilon/2)^\ell w_{\min}$ for all $\ell\ge 1$ until $(1+\varepsilon/2)^\ell \ge p w(A_p)$.
The largest $\ell$ that we have to test is at most $\log_{1+\varepsilon/2} p \frac{ w(A_p)}{w_{\min}}$. 
Thus, our total number of calls of \emph{Tree-Approx} is bounded by $p \cdot \log_{1+\varepsilon/2} p \frac{ w(A_p)}{w_{\min}} \le n \cdot \log_{1+\varepsilon/2} n \Delta$, where $\Delta$ is the largest distance between a terminal and its partner divided by the smallest such distance that is non-zero.

Assume that \emph{Tree-Approx} returns a solution $T$ with $\sum_{e \in E[T]} d(e) > \sum_{v \in T} \gamma(v)$ for all calls, which means that it does not find an improving connecting move.
Furthermore, assume that there exists a $((1+\varepsilon)\cdot c)$-approximate connecting move $T^\ast$ that we should have found, \ie which satisfies that $\sum_{e \in E[T^\ast]} d(e) \le \frac{1}{(1+\varepsilon)c} \sum_{v \in T^\ast} \gamma(v)$. Set $\Gamma^\ast := \sum_{v \in T^\ast} \gamma(v)$. Let $\ell$ be the index that satisfies $(1+\varepsilon/2)^\ell w_{\min} \le \Gamma^\ast < (1+\varepsilon/2)^{\ell+1} w_{\min}$ and consider the run with the correct choice of the largest width and the lower bound $\Gamma' := (1+\varepsilon/2)^\ell w_{\min}$.

Notice that $T^\ast$ is a feasible solution for this run: $w(T^\ast) = \Gamma^\ast \ge \Gamma'$ satisfies the lower bound. Thus, the optimal solution to the input has a cost of at most $\sum_{e \in E[T^\ast]} d(e)$. \emph{Tree-Approx} computes a $c$-approximation, \ie a solution $\hat T$ with $\sum_{v \in V[\hat T]} \gamma(v) \ge \Gamma'$ and 
\begin{align*}
\sum_{e \in \hat T} d(e) \le c \sum_{e \in E[T^\ast]} d(e) 
& \le \frac{c}{(1+\varepsilon)c} \sum_{v \in V[T^\ast]} \gamma(v)\\
& \le (1+\varepsilon/2) \frac{c}{(1+\varepsilon)c} \Gamma' < \Gamma' \le \sum_{v \in V[\hat T]} \gamma(v),
\end{align*}
which means that \emph{Tree-Approx} computes an improving connecting move.

\subsection{Weighted \texorpdfstring{$k$-MST}{k-MST}}\label{sec:wkmst}

This section is about ways to provide the algorithm \emph{Tree-Approx}. The problem we want to solve is a weighted version of the \emph{rooted $k$-MST problem}. Given $G=(V,E)$ with a root $r$, a metric $d: V\times V \to \mathbb{R}^+$ and a lower bound $k \in \mathbb{N}$, the rooted $k$-MST problem is to compute a tree $T$ in $G$ with $r \in V[T]$ and $|V[T]|\ge k$ that minimizes $\sum_{e \in E[T]} d(e)$. The unrooted $k$-MST problem is defined verbatim except that no distinguished root has to be part of the tree.

\paragraph{Work on $k$-MST.}
Fischetti \etal~\cite{FHJM94} show that the unrooted $k$-MST problem is NP-hard. 
Any algorithm for the rooted $k$-MST problem transfers to an algorithm for the unrooted case with the same approximation guarantee by testing all possible nodes and returning the best solution that was found. This in particular holds for optimal algorithms, so the rooted $k$-MST problem is also NP-hard.

As for example observed by Garg~\cite{G05}, we can also use algorithms for the unrooted $k$-MST problem to compute solutions for the rooted $k$-MST problem with the same approximation guarantee. To do so, create $n$ vertices with zero distance to the designated root vertex and search for a tree with $n+k$ vertices. Any such tree has to include at least one copy of the root, and at least $k-1$ other vertices. Thus, any solution for the unrooted $k$-MST problem is a feasible solution for the rooted $k$-MST problem, and the cost is the same.
Thus, the rooted and unrooted version of the $k$-MST problem are equivalent.

Blum, Ravi and Vempala~\cite{BRV96} develop the first constant-factor approximation for the $k$-MST problem, the factor is $17$. Subsequently, Garg~\cite{G96} gave a $3$-approximation, Arya and Ramesh~\cite{AR98} developed a $2.5$-approximation, Arora and Karakostas~\cite{AK00} proposed a $(2+\varepsilon)$-approximation, and, finally, Garg~\cite{G05} published a $2$-approximation algorithm for the $k$-MST problem. Chudak, Roughgarden and Williamson~\cite{CRW04} show that an easier $5$-approximation also proposed by Garg~\cite{G96} bears resemblances to the primal dual algorithm by Jain and Vazirani~\cite{JV01} for the $k$-median problem, in particular to the utilization of Lagrangean relaxation.

\paragraph{Connection to weighted $k$-MST.}
Johnson, Minkoff and Phillips~\cite{JMP00} observe the following reduction from the weighted $k$-MST problem to the unweighted $k$-MST problem, assuming all $\gamma(v)$ are integers. To create the unweighted instance $G'=(V',E')$, start with $V'=V$. Then, for any vertex $v$, add $2 \gamma(v)n-1$ vertices at distance zero of $v$ (thus, there are $2 \gamma(n)n$ vertices \lq at\rq\ $v$), and set $k$ to $2n \Gamma$. Any solution $T$ for the weighted $k$-MST problem can be interpreted as a solution $T'$ for the modified unweighted instance with $\sum_{v \in V[T]} 2n\gamma(v)=2n\Gamma$ vertices. Given a solution for the unweighted input, we can first change the solution thus that for any $v \in V$, either $v' \in V'$ is not picked or $v'$ is picked and its $2 \gamma(v) n-1$ copies are picked as well. This is possible since picking more vertices at the same location incurs no additional cost. After this step, the solution can be transformed into a weighted solution with enough weight by picking the corresponding vertices in $V$. 

This reduction constructs an input for the unweighted $k$-MST problem that is of pseudo-polynomial size. Johnson \etal~\cite{JMP00} note that algorithms for the unweighted $k$-MST problem can typically be adapted to handle the clouds of vertices at the same location implicitly without incurring a super-polynomial running time. They specifically state that this is true for the $3$-approximation algorithm by Garg~\cite{G96} for the rooted $k$-MST problem. 
The more recent $2$-approximation algorithm by Garg~\cite{G05} can also be adapted for the weighted case such that the running time remains independent of the weights~\cite{G16}. This yields a polynomial $2$-approximation algorithm for weighted $k$-MST.

\subsection{Convergence} 
\label{sec:convergence}

It is easy to see that a straightforward execution of the local search algorithm discussed in this paper runs in pseudo-polynomial time, as each improving move decreases the potential at least by the length of the shortest edge. We apply a standard rounding technique to make the algorithm polynomial.

\textbf{Algorithm X}
\begin{enumerate}
 \item Set $i := 0$ and let $\ALG_0 = \terms$.
 \item $\beta := \frac{\varepsilon \cdot \max_{\{v, \bar{v}\} \in \terms} d(v, \bar{v})}{|E|}, \quad d_{\beta}(e) :=  \quad  \ceil{\frac{d(e)}{\beta}}\beta.$
 \item While $\ALG_i$ admits an improving \pathset swap w.r.t to $\phi_{\beta}$, or \emph{Improving-Connecting-Move} finds an improving connecting move w.r.t. $\phi_{\beta}$, set $\ALG_{i+1}$ to be the resulting solution after applying the move, and $i:= i+1$.
 \item Return the solution $\ALG'$ obtained by dropping all inessential edges of $\ALG_i$.
\end{enumerate}

\begin{lemma}\label{ref:poly-time}
 Assuming that the locality gap for swap-optimal and $c$-approximate connecting move optimal solutions is $C$, Algorithm X computes in polynomial time a $(1 + \varepsilon)C$-approximation to Steiner Forest.
\end{lemma}

\begin{proof}
  We first observe that the algorithms runs in polynomial time. To see this, first note that $d_{\beta}(e) \geq \beta$ for all $e \in E$. Therefore, every improving {\pathsetswap} and every successful run of \emph{Improving-Connecting-Move} decreases the potential by at least $\beta$, \ie $\phi_{\beta}(\ALG_{i+1}) \leq \phi_{\beta}(\ALG_{i}) - \beta$. As $\phi_\beta(\ALG_0) = 2 \sum_{\{v, \bar{v}\} \in \terms} d_{\beta}(v, \bar{v}) \leq \frac{2 n_t |E|}{\varepsilon} \beta$, we conclude that the algorithm terminates after at most $2 \frac{n_t |E|}{\varepsilon}$ iterations, each of which can be executed in polynomial time.
  
  Now consider the output $\ALG'$ of Algorithm X. As it is {\pathset} swap optimal and $c$-approximately connecting move optimal w.r.t.~$d_{\beta}$, our assumption on the locality gap implies that $d_{\beta}(\ALG') \leq C d_{\beta}(\REFF_{\beta})$, where $\REFF_{\beta}$ is the optimal solution of the Steiner Forest instance defined by the metric $d_{\beta}$. Furthermore $d_{\beta}(\REFF_{\beta}) \leq d_{\beta}(\REFF)$, where $\REFF$ is an optimal solution to the original Steiner Forest instance defined by the metric $d$. We hence observe that
 \ama
   d(\ALG') & \leq d_{\beta}(\ALG') \leq C d_{\beta}(\REFF_{\beta}) \leq C d_{\beta}(\REFF) = C \cdot \sum_{e \in \REFF} \ceil{\frac{d(e)}{\beta}}\beta\\
   & \leq C \cdot \sum_{e \in \REFF} \left(\frac{d(e)}{\beta} + 1\right)\beta \leq C \cdot (d(\REFF) + |\REFF|\beta) \leq (1 + \varepsilon) C \cdot d(\REFF),\\
 \ema 
 where the last inequality follows from $|\REFF| \leq |E|$ and $d(\REFF) \geq \max_{\{v, \bar{v}\} \in \terms} d(v, \bar{v})$.
\end{proof}

Note that Corollary~\ref{cor:locality-gap} asserts that we can choose $C = 23(1 + c)$, and 
\cite{G05,G16} yields that $c=2$ is a feasible choice.
Thus Lemma~\ref{ref:poly-time} yields a polynomial-time $(1 + \varepsilon)\cdot 69$-approximation local search algorithm, proving Theorem~\ref{thm:main}.\jmcom{final occurrences of $c=5$}

\listoftodos

\end{document}